\documentclass[12pt,english]{article}
\usepackage[T1]{fontenc}
\usepackage[latin9]{inputenc}
\usepackage{geometry}
\usepackage{babel}
\usepackage{array}
\usepackage{verbatim}
\usepackage{booktabs}
\usepackage{url}
\usepackage{multirow}
\usepackage{amsmath}
\usepackage{amsthm}
\usepackage{amssymb}
\usepackage{graphicx}
\usepackage{rotating}
\usepackage{setspace}
\usepackage[authoryear,longnamesfirst]{natbib}
\usepackage[unicode=true,
 bookmarks=true,bookmarksnumbered=false,bookmarksopen=false,
 breaklinks=false,pdfborder={0 0 0},pdfborderstyle={},backref=false,colorlinks=false]
 {hyperref}
\hypersetup{pdftitle={Efficient and Convergent Sequential Pseudo-Likelihood Estimation of Dynamic Discrete Games},
 pdfauthor={Adam Dearing and Jason R. Blevins}}

\makeatletter

\providecommand{\tabularnewline}{\\}
\usepackage{tikz}
\usetikzlibrary{calc}

\theoremstyle{plain}
\newtheorem{assumption}{\protect\assumptionname}
\theoremstyle{definition}
 \newtheorem{example}{\protect\examplename}
\theoremstyle{plain}
\newtheorem{lem}{\protect\lemmaname}
\theoremstyle{plain}
\newtheorem{lyxalgorithm}{\protect\algorithmname}
\theoremstyle{plain}
\newtheorem{thm}{\protect\theoremname}
\theoremstyle{plain}
\newtheorem{prop}{\protect\propositionname}

\usepackage{graphicx}
\usepackage[nolists]{endfloat}

\@ifundefined{showcaptionsetup}{}{%
 \PassOptionsToPackage{caption=false}{subfig}}
\usepackage{subfig}
\makeatother

\providecommand{\algorithmname}{Algorithm}
\providecommand{\assumptionname}{Assumption}
\providecommand{\examplename}{Example}
\providecommand{\lemmaname}{Lemma}
\providecommand{\propositionname}{Proposition}
\providecommand{\theoremname}{Theorem}

\begin{document}
\title{Efficient and Convergent Sequential Pseudo-Likelihood Estimation of
Dynamic Discrete Games}
\author{Adam Dearing\thanks{Cornell University and NBER. S.C. Johnson Graduate School of Management,
S.C. Johnson College of Business, 304 Sage Hall, 141 E. Ave., Ithaca
NY 14853. Email: \texttt{aed237@cornell.edu}.} ~and Jason R. Blevins\thanks{The Ohio State University. Department of Economics, 410 Arps Hall,
1945 N. High St., Columbus OH 43210. Email: \texttt{blevins.141@osu.edu}.}\\
}
\date{April 24, 2024}
\maketitle
\begin{abstract}
We propose a new sequential Efficient Pseudo-Likelihood ($k$-EPL)
estimator for dynamic discrete choice games of incomplete information.
$k$-EPL considers the joint behavior of multiple players simultaneously,
as opposed to individual responses to other agents' equilibrium play.
This, in addition to reframing the problem from conditional choice
probability (CCP) space to value function space, yields a computationally
tractable, stable, and efficient estimator. We show that each iteration
in the $k$-EPL sequence is consistent and asymptotically efficient,
so the first-order asymptotic properties do not vary across iterations.
Furthermore, we show the sequence achieves higher-order equivalence
to the finite-sample maximum likelihood estimator with iteration and
that the sequence of estimators converges almost surely to the maximum
likelihood estimator at a nearly-superlinear rate when the data are
generated by any regular Markov perfect equilibrium, including equilibria
that lead to inconsistency of other sequential estimators. When utility
is linear in parameters, $k$-EPL iterations are computationally simple,
only requiring that the researcher solve linear systems of equations
to generate pseudo-regressors which are used in a static logit/probit
regression. Monte Carlo simulations demonstrate the theoretical results
and show $k$-EPL's good performance in finite samples in both small-
and large-scale games, even when the game admits spurious equilibria
in addition to one that generated the data. We apply the estimator
analyze competition in the U.S. wholesale club industry.

\textbf{Keywords:} dynamic discrete games, dynamic discrete choice,
multiple equilibria, pseudo maximum likelihood estimation.

\textbf{JEL Classification:} C57, C63, C73, L13.

\newpage{}
\end{abstract}

\section{\label{sec:Introduction}Introduction}

Estimation of dynamic discrete choice models -- particularly dynamic
discrete games of incomplete information -- is a topic of considerable
interest in economics. Broadly, likelihood-based estimation of these
models takes the form

\begin{align*}
\max_{(\theta,Y)\in\Theta\times\mathcal{Y}} & Q_{N}(\theta,Y)\\
\text{{s.t.\:}} & G(\theta,Y)=0,
\end{align*}
where $Q_{N}$ is the log-likelihood function based on $N$ independent
markets, $\theta$ is a finite-dimensional vector of parameters, $Y$
is a vector of important auxiliary parameters, and $G(\theta,Y)=0$
is a vector equality constraint that represents equilibrium conditions.
The parameters $\theta$ usually consist of the structural parameters
of the model. Common examples of auxiliary parameters $Y$ include
expected/integrated value functions or conditional choice probabilities,
since the equality constraint is often derived from an equilibrium
fixed point condition of the form $G(\theta,Y)\equiv Y-\Gamma(\theta,Y)=0$.

One approach to estimating these models is to directly impose the
fixed point equation for each trial value of $\theta$ visited by
the optimization algorithm by solving for $Y_{\theta}$ such that
$G(\theta,Y_{\theta})=0$. This approach was pioneered by \citet{Miller1984},
\citet{Wolpin1984}, \citet{Pakes1986}, and \citet{RustGMC1987}
for single-agent models, where the fixed point is unique and can be
computed via standard value function iteration or backwards induction.
Solution algorithms are available for dynamic games \citep{PakesMcguire1994,PakesMcGuire2001},
but it is often infeasible to nest those within an estimation routine
because the computational burden can be quite large. Furthermore,
in games the model may be incomplete due to multiple solutions $Y_{\theta}$
\citep{Tamer2003}.

These issues with the nested fixed-point approach led researchers
to extend conditional choice probability (CCP) estimators -- first
introduced in the seminal work of \citet{HotzMil1993} for single-agent
models -- to the case of dynamic discrete games. Of particular interest
here is the nested pseudo-likelihood approach of Aguirregabiria and
Mira \citeyearpar{AgMir2002,AgMir2007}.\footnote{Some other examples of CCP estimators are described in \citet{HotzMSS1994,BajBenkLevin2007,PakesOstBerry2007,PesSD2008}.}
They suggest using a $k$-step nested pseudo-likelihood ($k$-NPL)
approach, which defines a sequence of estimators, as an algorithm
for computing the nested pseudo-likelihood (NPL) estimator, a fixed
point of the sequence. In single-agent models, \citet{AgMir2002}
show that the $k$-NPL estimator is efficient for $k\geq1$ when initialized
with a consistent CCP estimate, in the sense that it has the same
limiting distribution as the (partial) maximum likelihood estimator.
Furthermore, \citet{KasShim2008} showed that the sequence converges
to the true parameter values with probability approaching one in large
samples. Indeed, the Monte Carlo simulations in \citet{AgMir2002}
show that single-agent $k$-NPL reliably converges to the maximum
likelihood estimate. The combination of computational simplicity,
efficiency, and convergence stability make $k$-NPL an attractive
alternative to other approaches, such as nested fixed point (computationally
burdensome) or standard Newton or Fisher scoring steps on the full
maximum likelihood problem (often diverges in finite samples in practice).\footnote{This is a well-known limitation of standard Newton steps, and further
regularization is often needed to ensure convergence.}

Unfortunately, these attractive properties of $k$-NPL are lost in
dynamic games. \citet{AgMir2007} show that $k$-NPL estimates are
in general not efficient for $k\leq\infty$, although they show that
the $\infty$-NPL (taking $k\to\infty$ or until convergence) estimator
outperforms the $1$-NPL estimator in efficiency when both are consistent.
But \citet{PesSD2010} show that the sequence may fail to converge
to the equilibrium that generated the data, even with very good starting
values, so that $\infty$-NPL may not be consistent (see also, \citet{KasShim2012},
\citet{EgLaiSu2015}, and \citet{AgMarcoux2019}). \citet{KasShim2012}
show that inconsistency occurs when the NPL mapping is unstable at
the data-generating equilibrium, which is essentially equivalent to
best-response instability of the equilibrium.\footnote{The NPL mapping is first-order equivalent to the best-response mapping
at the equilibrium.}

Another type of CCP estimator is the minimum-distance estimator (\citet{AltMill1998,PesSD2008}).
This estimator is both consistent and efficient in dynamic games,
and \citet{BugniBunt2020} develop a sequential version of this estimator,
which we refer to as $k$-MD. However, Monte Carlo simulations show
that these econometric properties come at the expense of greatly increased
computational burden, taking 12 to 26 times longer per iteration than
$k$-NPL in even the simplest small-scale dynamic games.\footnote{\citet{BugniBunt2020} perform Monte Carlo simulations with a very
small-scale game (two players, two actions, four states), and even
in that setting they remark, ``... computing the optimally weighted
1-MD, 2-MD, and 3-MD estimators takes us roughly 33\%, 75\%, and 80\%
more time than computing the 20-{[}NPL{]} estimator, respectively.''
These translate to roughly 26, 17.5, and 12 times longer per iteration
for each respective case.} This difference in computation time is likely to grow with the size
of the game, which would be a serious concern for empirical applications.
This leaves the researcher with an undesirable tradeoff between the
computationally simple $k$-NPL sequence and the more burdensome $k$-MD
sequence with better efficiency properties.

But there is another concern shared by the $k$-NPL and $k$-MD sequences:
finite sample performance of successive iterations. The $k$-MD estimator
uses the NPL mapping to update choice probabilities between iterations,
so there is reason to be concerned that it may mimic $k$-NPL's finite
sample properties when the data-generating equilibrium is NPL-unstable.
While the first-order asymptotic analysis of \citet{BugniBunt2020}
implies that both sequences are consistent for \emph{finite} $k$,
this asymptotic consistency does not necessarily lead to good performance
in finite samples. Indeed, the finite-sample performance of $k$-NPL
deteriorates with $k$ when the equilibrium is NPL-unstable. For example,
\citet{PesSD2008} consider a finite number of iterations for $k$-NPL
and find that it is severely biased in large-but-finite samples when
the equilibrium is NPL-unstable. In our own Monte Carlo simulations
presented later, we find that substantial bias appears rather quickly,
even for low values of $k$. This issue may also be a concern for
$k$-MD, since it is also consistent for \emph{finite} $k$ but has
unknown stability properties as $k\to\infty$, which may lead to deterioration
in finite-sample performance even for fixed $k$, similar to $k$-NPL.
This concern arises because \citet{KasShim2012} show that instability/inconsistency
of the $k$-NPL sequence arises from instability of the NPL mapping
used to update the choice probabilities between iterations.\footnote{However, a rigorous econometric analysis of the the $k$-MD estimator's
behavior as $k\to\infty$ is beyond the scope of this paper.}

With these concerns about various sequential methods in mind, an important
question arises: is there a CCP-based sequence that achieves a balance
between computational simplicity, asymptotic efficiency, and good
finite-sample properties with any number of iterations -- including
as $k\to\infty$ (iterating to convergence) -- in dynamic games?
The primary contribution of this paper is to provide such a method,
which we name the $k$-step Efficient Pseudo-Likelihood ($k$-EPL)
estimator. We show that $k$-EPL estimates are first-order asymptotically
equivalent to the maximum likelihood estimate for any number of iterations,
$k\geq1$. Thus, every estimate in the sequence is efficient. Furthermore,
we also show that higher-order improvements are achieved with iteration,
so the $k$-EPL sequence converges to the finite-sample maximum likelihood
estimator almost surely. The convergence rate is fast, approaching
super-linear as $N\to\infty$. This convergence result for $k$-EPL
holds even when the data-generating equilibrium is best-response unstable,
rendering $k$-NPL inconsistent.

One key distinction between $k$-EPL and $k$-NPL lies in how we incorporate
simultaneous play by multiple agents. While $k$-NPL focuses on single
agents' responses to a combination of an exogenous state transition
process and other agents' equilibrium play, $k$-EPL incorporates
the simultaneous nature of the game and is based on the joint behavior
of multiple players. Incorporating this additional information yields
increased asymptotic precision.

Despite this conceptual modification, our  $k$-EPL estimator retains
a simple computational structure similar to $k$-NPL. When utility
is linear in the parameters of interest, both $k$-EPL and $k$-NPL
iterations proceed in two stages: (i) solve a set of linear systems
to generate pseudo-regressors; and (ii) use the pseudo-regressors
in a static logit/probit maximum likelihood problem. The linear systems
only need to be solved once per iteration in the sequence, and the
static logit/probit problem is a low-dimensional, strictly concave
problem that has a unique solution and is easy to solve with out-of-the-box
optimization software. Because $k$-EPL incorporates all players simultaneously,
 the linear systems in $k$-EPL have larger dimension than those in
k-NPL. While this increases the relative computation time in large-dimensional
problems, the increase appears much less severe than the computation
time required for $k$-MD. We also find that iterating $k$-EPL to
convergence can be faster than doing so with $k$-NPL when the game
is not too large, since the need for fewer iterations results in lower
overall computation time.

One interesting implication of our higher-order analysis is that iterating
$k$-EPL can also provide a convenient algorithm for computing the
maximum likelihood estimator. We explore this in some of our Monte
Carlo simulations and find that it performs quite well, even with
random starting values. However, our primary focus is on the entire
$k$-EPL sequence, beginning with consistent initial estimates. We
find that much of the practical improvement from iteration is achieved
with just a few iterations in our Monte Carlo experiments, suggesting
that low values of $k$ can be quite effective when the initial CCP
estimates are consistent.

In recent related work, \citet{AgMarcoux2019} studied the finite-sample
properties of $\infty$-NPL ($k$-NPL with $k\to\infty$) and introduced
a variant of the $k$-NPL algorithm that updates the conditional choice
probabilities by applying spectral methods to the CCP updates. The
goal of their algorithm is to improve convergence properties of $\infty$-NPL
for unstable fixed points.\footnote{In the population or in finite samples, the NPL operator may have
spectral radius larger than one for some equilibria, rendering it
unstable. Conversely, the spectral radius of the EPL operator is zero
in the population and near zero in finite samples.} However, there are several differences between their work and ours.
First, they limit their analysis of the spectral algorithm to computing
the best fixed point of the $k$-NPL sequence, whereas we provide
analysis for all the iterations in the $k$-EPL sequence. Second,
upon convergence, the $k$-NPL and spectral $k$-NPL algorithms do
not produce the maximum likelihood estimator, and convergence can
require many iterations. In contrast, our $k$-EPL estimator has the
same limiting distribution as the MLE at each iteration and usually
converges locally to the MLE after few iterations in finite samples.
We verify these properties in our simulation studies.

Our work is also related to methods leveraging Neyman orthogonalization,
which has played a central role in recent advances in the broader
econometrics literature (see, e.g., \citet{ChernozhukovEtAl_TEJ2018,ChernozhukovEtAl_ECMA2022a,ChernozhukovEtAL_ECMA2022b}).
$k$-EPL leverages a type of quasi-Newton step in its construction,
leading to an important ``zero Jacobian'' property. Consequently,
each estimate in the $k$-EPL sequence is asymptotically Neyman orthogonal
to the previous estimate, and which leads to many of $k$-EPL's attractive
econometric properties.

We demonstrate the application of the $k$-EPL estimator through an
empirical analysis of the U.S. wholesale club industry, with a specific
focus on its three major players: Sam's Club, Costco, and BJ's. We
construct a structural dynamic model to examine the industry's competitive
landscape. Leveraging data on club stores operating across the United
States from 2009 to 2021, we employ $k$-EPL to estimate structural
parameters such as fixed costs, entry costs, the effect of market
size, and the competitive effect. Additionally, we consider a counterfactual
experiment designed to identify the key determinants of market entry
behavior and to explore their potential influence on the industry's
structure.

The remainder of the paper proceeds as follows. Section \ref{sec:Dynamic-Discrete-Game}
describes a generic dynamic discrete choice game of incomplete information.
Section \ref{sec:EPL} describes the $k$-EPL estimator, its asymptotic
and finite-sample properties, and its numerical implementation. Section
\ref{sec:Monte-Carlo-Simulations} provides Monte Carlo simulations,
and additional simulation results are included in the appendix. Section
\ref{sec:Empirical-Application} describes our empirical application
to the U.S. wholesale club industry. Section \ref{sec:Conclusion}
concludes. All proofs appear in the Appendix.

\section{\label{sec:Dynamic-Discrete-Game}Dynamic Discrete Games of Incomplete
Information}

Here we describe a canonical stationary dynamic discrete game of complete
information in the style of \citet{AgMir2007} and \citet{PesSD2008}.
Time is discrete, indexed by $t=1,2,3.\dots$. In a given market,
there are $J$ firms indexed by $j\in\mathcal{J}=\{1,2,\dots,|\mathcal{J}|\}$.
Given a vector of state variables observable to all agents and the
econometrician, $x_{t}$, and its own private information $\varepsilon_{t}^{j}$,
each firm chooses an action, $a_{t}^{j}\in\mathcal{A}=\{0,1,2,\dots,|\mathcal{A}|-1\}$.
Action zero is the ``outside option'' when applicable. All players
choose their actions simultaneously.

Agents have flow utilities (profits), $U^{j}(x_{t},a_{t}^{j},a_{t}^{-j},\varepsilon_{t}^{j};\theta_{u})$,
where $a_{t}^{-j}$ are the actions of the other players. States transition
according to $p(x_{t+1},\varepsilon_{t+1}\mid a_{t},x_{t},\varepsilon_{t};\theta_{f})$,
and the discount factor is $\beta\in(0,1)$. Agents choose actions
to maximize expected discounted utility,
\[
\mathrm{E}\left\{ \sum_{s=0}^{\infty}\beta^{s-t}U^{j}(x_{s},a_{s}^{j},a_{s}^{-j},\varepsilon_{s}^{j};\theta_{u})\,\bigg|\,x_{t},\varepsilon_{t}^{j}\right\} .
\]
The primary parameter of interest is $\theta=(\theta_{u},\theta_{f})$.
Furthermore, we impose the following standard assumptions on the primitives.
\begin{assumption}
\label{assu:Additive-Separability}(Additive Separability) $U^{j}(x_{t},a_{t}^{j},a_{t}^{-j},\varepsilon_{at}^{j};\theta_{u})=\bar{u}(x_{t},a_{t}^{j},a_{t}^{-j};\theta_{u})+\varepsilon_{t}^{j}(a_{t}^{j})$.
\end{assumption}
\begin{assumption}
\label{assu:Conditional-Independence}(Conditional Independence) $p(x_{t+1},\varepsilon_{t+1}\mid x_{t},a_{t},\varepsilon_{t};\theta_{f})=g(\varepsilon_{t+1})f(x_{t+1}\mid x_{t},a_{t};\theta_{f})$,
where $g(\varepsilon_{t+1})$ is absolutely continuous with respect
to the Legesgue measure on $\mathbb{R}^{|\mathcal{A}|\times|\mathcal{J}|}$.
\end{assumption}
\begin{assumption}
\label{assu:Independent-Private-Values}(Independent Private Values)
Private values are independently distributed across players.
\end{assumption}
\begin{assumption}
\label{assu:Finite-Observed-State}(Finite Observed State Space) $x_{t}\in\mathcal{X}=\{1,2,\dots,|\mathcal{X}|\}$.
\end{assumption}
Assumptions 1--4 here correspond to Assumptions 1--4 in \citet{AgMir2007}.
In most applications in the literature, the private shocks are assumed
to be either i.i.d. Type 1 Extreme Value or normal, both of which
satisfy the assumptions.
\begin{example}
\label{exa:Wholesale-Club-Store}Wholesale Club Store Entry and Exit:
Firms are wholesale club stores (Costco, Sam's Club, BJ's, etc.) making
decisions of whether to operate in a market ($a_{it}^{j}=1$, ``entry'')
or not ($a_{it}^{j}=0$, ``exit''). After the entry/exit decision
is made in each period, the firms receive profits that result from
a static equilibrium competition (e.g., in prices or quantities).
The outcome of the static competition equilibrium depends on (i) the
number of active firms; and (ii) the market size, $s_{it}$. The observed
state going into period $t$ is $x_{it}=(s_{it},a_{i,t-1})$, which
includes the market size and indicators for incumbency. The per-period
profit for active firm $j$ is then given by 
\[
\bar{u}^{j}(x_{it},a_{it}^{j}=1,a_{it}^{-j};\theta)=\theta_{\text{FC},j}+\theta_{\text{RS}}s_{it}-\theta_{\text{RN}}\ln\left(1+\sum_{l\neq j}a_{it}^{l}\right)-\theta_{\text{EC}}(1-a_{i,t-1}^{j}).
\]
 Here, $\theta_{FC,j}<0$ is the fixed cost of operation for firm
$j$, $\theta_{EC}>0$ is the entry cost (which is not paid by incumbents),
$\theta_{RS}>0$ represents the effect of market size on flow profit,
and $\theta_{RN}>0$ represents the effect of competition on flow
profits. Flow profit for an inactive firm is normalized to zero: $\bar{u}^{j}(x_{it},a_{it}^{j}=0,a_{it}^{-j};\theta)=0$,
which is required for identification.\footnote{See, e.g., the discussion of Example 1 in \citet{AgMir2007}. We note
that this is not without loss of generality for counterfactuals, as
discussed in \citet{KalScoSouz2017}.}
\end{example}
The operative equilibrium concept here will be that of a Markov Perfect
Nash equilibrium. We will consider stationary equilibria only, so
from here we drop the time subscript. Because moves are simultaneous,
the actions of player $j$ do not depend directly on $a^{-j}\in\mathcal{A}^{|\mathcal{J}|-1}$,
but rather on $P^{-j}(x)\in\Delta^{|\mathcal{J}|-1}$, where $P^{-j}(x)$
is player $j$'s belief about the other players' probability of playing
the corresponding actions in state $x$ and $\Delta$ is the unit
simplex in $\mathbb{R}^{\left|\mathcal{A}\right|-1}$. So, from here
on out we will work with the following utility function and transition
probabilities:
\[
u^{j}(x,a^{j};P^{-j},\theta_{u})=\sum_{a^{-j}\in\mathcal{A}^{|\mathcal{J}|-1}}P^{-j}(x,a^{-j})\bar{u}(x,a^{j},a^{-j};\theta_{u})
\]
 
\[
f^{j}(x'\mid x,a^{j};P^{-j},\theta_{f})=\sum_{a^{-j}\in\mathcal{A}^{|\mathcal{J}|-1}}P^{-j}(x,a^{-j})f(x'\mid x,a^{j},a^{-j};\theta_{f})
\]

Now consider the vector of player $j$'s (expected) choice-specific
value functions, $v^{j}\in\mathbb{R}^{|\mathcal{X}|\times|\mathcal{A}|}$,
and define the corresponding choice probabilities as $\Lambda^{j}(x,a^{j};v^{j})$,
which is the probability agent $j$ chooses action $a^{j}$ in state
$x$, conditional on having choice-specific value function $v^{j}$.\footnote{We note that the choice-specific value functions, $v^{j}$, are also
often referred to as conditional value functions.} In equilibrium, the choice probabilities will be $P^{j}(a)=\Lambda^{j}(a;v^{j})$.
And let 
\[
\Lambda^{-j}(v^{-j})=(\Lambda^{1}(v^{1}),\dots,\Lambda^{j-1}(v^{j-1}),\Lambda^{j+1}(v^{j+1}),\dots,\Lambda^{|\mathcal{J}|}(v^{|\mathcal{J}|})),
\]
so that in equilibrium $P^{-j}=\Lambda^{-j}(v^{-j})$. \citet{AgMir2007}
show that in equilibrium, the choice-specific value functions are
equal to 
\[
v^{j}(x,a^{j})=u^{j}(x,a^{j};P^{-j},\theta_{u})+\beta\sum_{x'}f^{j}(x'\mid x,a^{j};P^{-j},\theta_{f})\Gamma^{j}(x';\theta,P),
\]
 where 
\[
\Gamma^{j}(\theta,P)=\left(I-\beta F(\theta_{f},P)\right)^{-1}\sum_{a^{j}}P^{j}(a^{j})*\left(u^{j}(a^{j};P^{-j},\theta_{u})+e(a^{j};P^{j})\right)
\]
 maps into \emph{ex-ante} (or integrated) value function space. Here,
$e\left(a^{j};P^{j}\right)$ stacks the values of $e(x,a^{j};P^{j})\equiv E[\varepsilon^{j}(a^{j})\mid x,a^{j},P^{j}]$,
as defined in Aguirregabiria and Mira (\citeyear{AgMir2007}, Equation
11), and $F(\theta_{f},P)$ is an unconditional state transition matrix
with elements $F(\theta_{f},P)\{k,l\}=\sum_{a\in\mathcal{A}^{J}}\left(\Pi_{j=1}^{J}P^{j}(a^{j}\mid x=k)\right)f(x'=l\mid x,a;\theta_{f})$.\footnote{See footnote 6 in \citet{AgMir2007}.}
They then define the NPL operator, $\Psi(\theta,P)$, such that $\Psi^{j}(\theta,P)=\Lambda^{j}(\Gamma^{j}(\theta,P))$
and combining all players yields the following fixed-point condition
that describes any Markov perfect equilibrium (Aguirregabiria and
Mira \citeyear{AgMir2007}, Lemma 1): 
\[
P=\Psi(\theta,P).
\]

While this equilibrium representation based on the NPL operator is
often useful, we will ultimately want to work with an alternative
representation when implementing our new estimator. This alternative
arises due to a change of variables from $P$ space to $v$ space.
(See Section~\ref{subsec:Computational-Implementation} for a detailed
discussion of the importance of this change.) Define the function
\begin{align*}
\Phi^{j}(x,a^{j};v^{j},v^{-j},\theta) & =u^{j}(x,a;\Lambda^{-j}(v^{-j}),\theta_{u})+\beta\sum_{x'}f^{j}(x'\mid x,a^{j};\Lambda^{-j}(v^{-j}),\theta_{f})S(v^{j}(x')),
\end{align*}
where $\Phi:\Theta\times\mathbb{R}^{|\mathcal{J}|\times|\mathcal{X}|\times|\mathcal{A}|}\to\mathbb{R}^{|\mathcal{J}|\times|\mathcal{X}|\times|\mathcal{A}|}$
and $S(\cdot)$ is McFadden's social surplus function.\footnote{For example, $S(v^{j}(x))=\ln(\sum_{a^{j}}\exp(v^{j}(x,a^{j}))+\bar{\gamma}$
when the private values are i.i.d. and follow the type 1 extreme value
distribution, where $\bar{\gamma}$ is the Euler-Mascheroni constant.} This $\Phi(\cdot)$ function allows us to characterize the equilibrium
with an alternative fixed-point equation, as described in the following
Lemma.
\begin{lem}
\label{lem:Rep-Lemma}(Representation Lemma) Under Assumptions \ref{assu:Additive-Separability}--\ref{assu:Finite-Observed-State},
choice-specific value functions characterize a Markov perfect equilibrium
for $\theta$ if and only if $v^{j}(x,a^{j})=\Phi^{j}(x,a^{j};\theta,v^{j},v^{-j})$
for all $(j,x,a)\in\mathcal{J}\times\mathcal{X}\times\mathcal{A}$.
Or more succinctly, 
\[
v=\Phi(\theta,v).
\]
\end{lem}

\section{\label{sec:EPL}The $k$-EPL Estimator}

This section describes the $k$-EPL estimator and discusses its asymptotic
and finite-sample properties, as well as computational aspects of
its implementation in dynamic discrete choice games. We begin by discussing
maximum likelihood estimation, subject to an equilibrium constraint
based on some nuisance parameter, $Y$. The model is parameterized
by a finite-dimensional vector, $\theta\in\Theta\subset\mathbb{R}^{|\Theta|}$,
and a constraint $G(\theta,Y)=0$ where $Y\in\mathcal{Y}\subset\mathbb{R}^{|\mathcal{Y}|}$
and $G:\Theta\times\mathcal{Y}\to\mathbb{R}^{|\mathcal{Y}|}$. The
true parameter values are $\theta^{*}$ and $Y^{*}$, with $G(\theta^{*},Y^{*})=0$.
Note that there may be other values of $Y$ satisfying the constraint
at $\theta^{*}$, but we will assume that the data are generated from
only one such value, a common assumption in the literature.

The alternative statements of the equilibrium conditions in Section
\ref{sec:Dynamic-Discrete-Game} yield two potential choices for $Y$
and the corresponding constraint. Our asymptotic consistency and efficiency
results in this section do not depend on the choice of nuisance parameter,
but this choice can have serious implications for the computational
implementation of the $k$-EPL estimator. So, for computational purposes
we ultimately use the choice-specific value functions as the nuisance
parameter ($Y\equiv v$), and the constraint comes from the equation
in Lemma \ref{lem:Rep-Lemma} ($v-\Phi(\theta,v)=0$). We discuss
the computational concerns in more detail in Section \ref{subsec:Computational-Implementation}.
But for now, we present the asymptotic theory with the generic nuisance
parameter, $Y$.

Let $w_{i}$ for $i=1,\dots,N$ denote the observations from $N$
independent markets, and define
\begin{align*}
Q_{N}(\theta,Y) & \equiv N^{-1}\sum_{n=1}^{N}q_{i}(\theta,Y)=N^{-1}\sum_{i=1}^{N}\ln Pr(w_{i}\mid\theta,Y),
\end{align*}
where $q_{i}(\theta,Y)\equiv\ln Pr(w_{i}\mid\theta,Y)$. Furthermore,
let $Q^{*}(\theta,Y)\equiv\mathrm{E}[Q_{N}(\theta,Y)]$.
\begin{assumption}
\label{assu:Basics}(a) The observations $\{w_{i}:i=1,\dots,N\}$
are i.i.d. and generated by a single equilibrium $(\theta^{*},Y^{*})$.
(b) $\Theta$ and $\mathcal{Y}$ are compact and convex and $(\theta^{*},Y^{*})\in int(\Theta\times\mathcal{Y})$.
(c) $Q_{N}(\theta,Y)$ and $Q^{*}(\theta,Y)$ are twice continuously
differentiable. $Q^{*}$ has a unique maximum in $\Theta\times\mathcal{Y}$
subject to $G(\theta,Y)=0$, and the maximum occurs at $(\theta^{*},Y^{*})$.
(d) $G(\theta,Y)$ is thrice continuously differentiable and $\nabla_{Y}G(\theta^{*},Y^{*})$
is non-singular.
\end{assumption}
Assumptions \ref{assu:Basics}(a)-(c) echo standard identification
assumptions. We note that assuming that $Q^{*}$ has a unique maximum
does not rule out games with multiple equilibria. Non-singularity
of the Jacobian in (d) is the defining feature of regular Markov perfect
equilibria in the sense of \citet{DorEscobar2010}. Regularity essentially
means that the equilibrium is locally isolated and we can apply the
implicit function theorem to obtain $Y(\theta)$ locally.\footnote{\citet{AgMir2007} directly assume the local existence of $Y(\theta)$,
instead of appealing to the implicit function theorem.}

One method to estimate these models is via constrained maximum likelihood:
\begin{align*}
\left(\hat{\theta}_{\text{MLE}},\hat{Y}_{\text{MLE}}\right)=\arg & \max_{(\theta,Y)\in\Theta\times\mathcal{{Y}}}Q_{N}(\theta,Y)\\
 & \text{s.t. }G(\theta,Y)=0.
\end{align*}
 An equivalent statement is $\hat{\theta}_{\text{MLE}}=\textrm{arg max}_{\theta}\quad Q_{N}(\theta,Y(\theta))$,
where $G(\theta,Y(\theta))=0$ and $\hat{Y}_{\text{MLE}}=Y(\hat{\theta}_{\text{MLE}})$.
Pseudo-likelihood estimation replaces $Y(\theta)$ with some other
mapping. \citet{AgMir2007} define $Y\equiv P$ and replace $Y(\theta)\equiv P(\theta)$
with their $\Psi(\theta,\hat{P}_{k-1})$ for the $k$-th iteration
in the $k$-NPL sequence. However, this procedure suffers from the
issues discussed in Section \ref{sec:Introduction}.

Our $k$-step Efficient Pseudo-Likelihood ($k$-EPL) sequence instead
uses a ``Newton-like'' step, which provides a good approximation
to the full Newton step but uses a fixed value of $\nabla_{Y}G(\theta,Y)^{-1}$;
this value varies between steps but does not vary as the optimizer
searches over different values $\theta$ within each step. Algorithm
\ref{alg:EPL-Algorithm} below defines our sequential estimation procedure.
It uses our Newton-like mapping, $\Upsilon(\cdot)$, which is a function
of an initial compound parameter vector, $\gamma=(\theta,Y)$, and
an additional (possibly different) value of $\theta$.
\begin{lyxalgorithm}
\label{alg:EPL-Algorithm}($k$-step Efficient Pseudo-Likelihood,
or $k$-EPL)
\begin{itemize}
\item Step 1: Obtain strongly $\sqrt{N}$-consistent initial estimates $\hat{\gamma}_{0}=(\hat{\theta}_{0},\hat{Y}_{0})$.
\item Step 2: For $k\ge1$, obtain parameter estimates iteratively: 
\[
\hat{\theta}_{k}=\underset{\theta\in\Theta}{\textrm{\ensuremath{\arg\max}}}\quad Q_{N}\left(\theta,\Upsilon(\theta,\hat{\gamma}_{k-1})\right)
\]
 where 
\[
\Upsilon(\theta,\hat{\gamma}_{k-1})=\hat{Y}_{k-1}-\nabla_{Y}G(\hat{\theta}_{k-1},\hat{Y}_{k-1})^{-1}G(\theta,\hat{Y}_{k-1})
\]
 and update the auxiliary parameters:
\[
\hat{Y}_{k}=\Upsilon(\hat{\theta}_{k},\hat{\gamma}_{k-1}).
\]
\item Step 3: Increment $k$ and repeat Step 2 until desired value of $k$
is reached or until numerical convergence.
\end{itemize}
\end{lyxalgorithm}
This procedure enjoys some nice econometric properties, both asymptotically
and in finite samples. These properties arise from some convenient
features of the $\Upsilon(\cdot)$ function, which are detailed in
the following lemma.
\begin{lem}
\label{lem:Newt-Props} Let $\Upsilon(\theta,\gamma)$ denote the
operator defined in Algorithm \ref{alg:EPL-Algorithm} and define
$Y_{\theta}\equiv Y(\theta)$ and $\gamma_{\theta}\equiv(\theta,Y_{\theta})$.
Under Assumption \ref{assu:Basics}, if $\nabla_{Y}G(\theta,Y_{\theta})$
is non-singular, then the following properties hold:
\end{lem}
\begin{enumerate}
\item Roots of $G$ and fixed points of $\Upsilon$ are identical: $\Upsilon(\theta,\gamma_{\theta})=Y(\theta)\iff G(\theta,Y_{\theta})=0$.
\item $\nabla_{\theta}\Upsilon(\theta,\gamma_{\theta})=\nabla_{\theta}Y(\theta)$.
\item $\nabla_{\gamma}\Upsilon(\theta,\gamma_{\theta})=0$ (Zero Jacobian
Property).
\end{enumerate}
Lemma \ref{lem:Newt-Props} is the key to most of the results in this
section, which arise from applying the lemma at $(\theta^{*},\gamma^{*})$
and $(\hat{\theta}^{MLE},\hat{\gamma}^{MLE})$. For now, we note that
Result 3 of Lemma \ref{lem:Newt-Props} is analogous to the ``zero
Jacobian'' property from Proposition 2 of \citet{AgMir2002}, which
was the key to both their efficiency results and the finite-sample
convergence results of \citet{KasShim2008} for single-agent $k$-NPL.
By utilizing Newton-like steps on the equilibrium constraint, the
$k$-EPL algorithm restores this zero Jacobian property in dynamic
games.

One notable difference between $k$-EPL and some other sequential
estimators is that $k$-EPL is initialized from a consistent estimate
of both the structural parameter and the nuisance parameter, while
other estimators -- including $k$-NPL -- only require an initially
consistent estimate of the (nuisance) CCPs. Other examples include
the finite-dependence-based estimators (\citet{HotzMil1993,ArcidMiller2011})
and inequality-based estimators (\citet{BajBenkLevin2007}). While
none of these other estimators offer efficiency or convergence guarantees
as $k\to\infty$, they can be used to obtain a consistent parameter
estimate for initializing $k$-EPL. If the model exhibits finite-dependence,
then finite-dependence-based estimators can be particularly attractive
for obtaining initial estimates because of there computational simplicity.
However, we note that the example models in our Monte Carlo simulations
and empirical application may not have the finite dependence property
because exit is not permanent. So, we instead use a $1$-NPL estimate
for initialization.\footnote{Finite dependence is more difficult to establish in entry/exit games
without a permanent exit decision that leads directly to  a terminal
state property. See, e.g., the discussion in Section 5.2 of \citet{ArcidEllick2011}.
However, \citet{ArcidiaconoMiller2019} show how to derive finite
dependence in a class dynamic games that does not require a terminal
state.}

\subsection{Asymptotic Properties of $k$-EPL}

One implication of Lemma~\ref{lem:Newt-Props} is that $k$-EPL then
gives a sequence of asymptotically efficient estimators that converges
almost surely in large samples. We state this result formally in the
following Theorem.
\begin{thm}
\label{thm:kEPL-Properties}(Asymptotic Properties of $k$-EPL) Under
Assumption \ref{assu:Basics}, the $k$-EPL estimates computed with
Algorithm \ref{alg:EPL-Algorithm} satisfy the following for any $k\geq1$:
\begin{enumerate}
\item (Consistency) $\hat{\gamma}_{k}=(\hat{\theta}_{k},\hat{Y}_{k})$ is
a strongly consistent estimator of $(\theta^{*},Y^{*})$.
\item (Efficiency) $\sqrt{N}(\hat{\theta}_{k}-\theta^{*})\overset{d}{\to}\mathcal{\mathrm{N}}(0,\Omega_{\theta\theta}^{*-1})$,
where $\Omega_{\theta\theta}^{*}$ is the information matrix evaluated
at $\theta^{*}$.
\item (Large Sample Convergence) There exists a neighborhood of $\gamma^{*}=(\theta^{*},Y^{*})$,
$\mathcal{B}^{*}$, such that $\lim_{k\to\infty}\hat{\gamma}_{k}=\hat{\gamma}_{\text{MLE}}$
almost surely for any $\hat{\gamma}_{0}\in\mathcal{B}^{*}$. In other
words,
\[
\Pr\left[\lim_{k\rightarrow\infty}\hat{\gamma}_{k}=\hat{\gamma}_{\text{MLE}}\;\middle|\;\hat{\gamma}_{0}\in\mathcal{B}^{*}\right]=1.
\]
\end{enumerate}
\end{thm}
The results of Theorem \ref{thm:kEPL-Properties} for $k$-EPL in
games are shared by $k$-NPL in single-agent models, but not games.
In short, the zero Jacobian property ensures that $\hat{\gamma}_{k}=(\hat{\theta}_{k},\hat{Y}_{k})$
is asymptotically orthogonal to $\hat{\gamma}_{k-1}$, so that using
$\hat{\gamma}_{k-1}$ is asymptotically equivalent to using $\gamma^{*}=(\theta^{*},Y^{*})$
at each step. This drives the consistency (Result 1) and asymptotic
equivalence to MLE (Result 2) of each step. Intuitively, an EPL step
is similar to a Newton step on the full maximum likelihood problem,
although iterating on that procedure is notoriously unstable unless
properly regularized.\footnote{\citet{Train2009} discusses the need to alter the step size in order
to obtain global convergence (pp. 189-191). \citet{Nesterov2004}
also discuses divergence (Section 1.2).} Our Monte Carlo simulations in Section (\ref{sec:Monte-Carlo-Simulations})
show that $k$-EPL is stable without further regularization.

\subsection{Iteration to the Maximum Likelihood Estimate}

While the asymptotic distribution is insensitive to iteration, we
can obtain substantial finite-sample improvements by iterating and
can even compute the finite-sample MLE by iterating to convergence.
In this section, we proceed with a formal econometric analysis of
the local convergence rate of the iterations to the maximum likelihood
estimator, and we discuss the implications for finite sample performance.
Later on, the finite sample properties are illustrated in the Monte
Carlo simulations in Section \ref{sec:Monte-Carlo-Simulations}.

Our results for the convergence rate to MLE are similar to those of
the single-agent version of $k$-NPL in \citet{AgMir2002}, which
also has the zero Jacobian property. In their Monte Carlo simulations,
single-agent $k$-NPL iterations exhibit rapid finite-sample improvements
and reliably converged to the finite-sample MLE. \citet{KasShim2008}
then provided a formal econometric explanation for these results.
The analysis of $k$-EPL's finite sample properties in this section
is similar but also applies to games.

The only additional requirement for our finite-sample results is that
the Jacobian of the equality constraints, $G$, with respect to $Y$
is nonsingular at the finite-sample MLE.
\begin{assumption}
\label{assu:MLE-non-singular}$\nabla_{Y}G(\hat{\theta}_{\text{{MLE}}},\hat{Y}_{\text{{MLE}}})$
is non-singular.
\end{assumption}
Assumption \ref{assu:MLE-non-singular} guarantees the existence of
an implicit function, $Y(\theta)$, around $\hat{\theta}_{MLE}$ and
also that the quasi-Newton mapping, $\Upsilon(\theta,\hat{\gamma}_{MLE})$
is valid. Assumption \ref{assu:Basics} is enough to guarantee that
Assumption \ref{assu:MLE-non-singular} is satisfied almost surely
as $N\to\infty$, since it implies that $\det\left(\nabla_{Y}G\left(\hat{\theta}_{MLE},\hat{Y}_{MLE}\right)\right)\overset{a.s.}{\to}\det\left(\nabla_{Y}G\left(\theta^{*},Y^{*}\right)\right)\neq0$
by continuity of $\det(\cdot)$, continuity of $\nabla_{Y}G(\cdot)$
(Assumption 5(d)), and strong consistency of $\hat{\gamma}_{MLE}$.
Furthermore, the set of singular matrices has measure zero. So, we
view this as a relatively mild (but important) assumption.
\begin{thm}
\label{thm:Finite-Sample-Properties}(Local Convergence Results for
Iterating to MLE) Suppose Assumptions \ref{assu:Basics} and \ref{assu:MLE-non-singular}
hold and that the optimization problem in Step 2 of Algorithm \ref{alg:EPL-Algorithm}
has a unique solution for all $k\geq1$. Then,
\begin{enumerate}
\item The MLE is a fixed point of the EPL iterations: if $\hat{\gamma}_{k-1}=\hat{\gamma}_{\text{MLE}}$,
then $\hat{\gamma}_{k}=\hat{\gamma}_{\text{MLE}}$.
\item For all $k\geq1$, 
\[
\hat{\gamma}_{k}-\hat{\gamma}_{\text{MLE}}=O_{p}(N^{-1/2}||\hat{\gamma}_{k-1}-\hat{\gamma}_{\text{MLE}}||+||\hat{\gamma}_{k-1}-\hat{\gamma}_{\text{MLE}}||^{2}).
\]
\item W.p.a. 1 as $N\to\infty$, for any $\varepsilon>0$ there exists some
neighborhood of $\hat{\gamma}_{\text{MLE}}$, $\mathcal{B}$, such
that the EPL iterations define a contraction mapping on $\mathcal{B}$
with Lipschitz constant, $L<\varepsilon$.
\end{enumerate}
\end{thm}
The first result of Theorem \ref{thm:Finite-Sample-Properties} establishes
that the MLE is a fixed point of the $k$-EPL iterations in a finite
sample, similar to Aguirregabiria and Mira (\citeyear{AgMir2002},
Proposition 3) for single-agent $k$-NPL. The second result gives
an asymptotic analysis of convergence to MLE, which provides a theoretical
explanation of why we should expect iteration to yield improvements
in finite samples. This result is analogous to Proposition 2 of \citet{KasShim2008},
although their result only applies in the single-agent case. In short,
even though iteration on EPL provides no improvement up to $O_{p}(N^{-1/2})$,
it still yields higher-order improvements. To see why, suppose the
initial estimates are such that $\hat{\gamma}_{0}-\gamma^{*}=O_{p}(N^{-b})$
for $b\in(1/4,1/2]$, so that $||\hat{\gamma}_{0}-\hat{\gamma}_{\text{MLE}}||=O_{p}(N^{-b})$.\footnote{For $b\in(1/4,1/2]$, we have $\hat{\gamma}_{0}-\hat{\gamma}_{\text{MLE}}=\hat{\gamma}_{0}-\gamma^{*}-(\hat{\gamma}_{\text{MLE}}-\gamma^{*})=O_{p}(N^{-b})+O_{p}(N^{-1/2})=O_{p}(N^{-b})$
. Additionally, this can be used to show higher-order equivalence
to the MLE.} Repeated substitution gives $||\hat{\gamma}_{k}-\hat{\gamma}_{\text{MLE}}||=O_{p}(N^{-(k-1)/2-2b})$.
In particular, in the case where the state space is finite and frequency
or kernel estimates are used, $b=1/2$ and $||\hat{\gamma}_{k}-\hat{\gamma}_{\text{MLE}}||=O_{p}(N^{-(k+1)/2})$,
where $N^{-(k+1)/2}\to0$ as $k\to\infty$ for $N>1$. Our own Monte
Carlo simulations in Section \ref{sec:Monte-Carlo-Simulations} exhibit
such improvements.

The third result in Theorem \ref{thm:Finite-Sample-Properties} allows
us to consider EPL iterations as a computationally attractive algorithm
for computing the MLE. It establishes that we can expect the EPL iterations
to be a local contraction around the MLE in the finite sample with
a very fast convergence rate. The full proof appears in the appendix,
but it essentially proceeds by noting that the $k$-EPL sequence satisfies
$\hat{\gamma}_{k}=H_{N}(\hat{\gamma}_{k-1})$, where $\hat{\gamma}_{MLE}$
is a fixed point of the function $H_{N}$. And due to the zero-Jacobian
property in Lemma \ref{lem:Newt-Props} (Result 3), we obtain $\nabla_{\gamma}H_{N}(\hat{\gamma}_{MLE})\overset{a.s.}{\to}0$.\footnote{This drives the Neyman orthogonality discussed in the introduction.}

For the population analogue of the EPL iterations, the convergence
rate is then super-linear. However, we only have access to finite
samples in practice, so we should expect the convergence rate to be
linear with a small Lipschitz constant, implying that we'll need only
a few iterations to achieve convergence. We can therefore use EPL
iterations to compute the MLE even when a consistent $\hat{\gamma}_{0}$
is unavailable. We can simply use multiple starting values, iterate
to convergence, and use the converged estimate that provides the highest
log-likelihood.\footnote{\citet{AgMir2007} suggest a similar procedure to find $\infty$-NPL
when no initial consistent estimate is available, and multiple starting
values are often used when computing the maximum likelihood estimate
with other methods (\citet{SuJudd2012,EgLaiSu2015}).} We demonstrate this usage of $k$-EPL with Monte Carlo simulations
in the appendix and find that it works well.

Aside from $k$-EPL, there are two potential alternative algorithms
for computing the MLE: the nested fixed-point (NFXP) algorithm \'a
la \citet{RustGMC1987} and the MPEC approach proposed by \citet{SuJudd2012}
and extended to dynamic games by \citet{EgLaiSu2015}. The NFXP algorithm
searches over $\theta$ in an outer loop and finds $Y_{\theta}$ such
that $G(\theta,Y_{\theta})=0$ in an inner loop. MPEC leverages modern
optimization software to search over $\theta$ and $Y$ simultaneously,
only imposing that $G(\theta,Y)=0$ at the solution. The algorithm
of choice may depend on the structure of the model.

While this section discusses the $k$-EPL algorithm in the context
of a general constrained maximum-likelihood problem, we are ultimately
focused on estimating dynamic discrete choice games of incomplete
information. As discussed in the introduction, NFXP is often computationally
unattractive---or even infeasible---in such games.\footnote{We note that NFXP still performs well in single-agent dynamic models.
See \citet{DorJudd2012} and \citet{ABBE2016} for details on the
computational burden of computing equilibria in discrete-time dynamic
discrete games.} MPEC, however, remains feasible and performs well, as demonstrated
by \citet{EgLaiSu2015}. The key difference here between MPEC and
$k$-EPL is that $k$-EPL will be able to more heavily exploit the
structure of the problem.\footnote{\citet{EgLaiSu2015} exploit sparsity patterns in their MPEC implementation
but do not further exploit other features of the problem structure.} In Section \ref{subsec:Computational-Implementation}, we show that
--- much like $k$-NPL in single-agent models --- common modeling
assumptions lead to EPL iterations composed of two easily-computed
parts: solving linear systems to form pseudo-regressors, followed
by solving an unconstrained, globally concave maximization problem
\'a la static logit/probit with those pseudo-regressors. Neither
of these operations require sophisticated commercial optimization
software; and repeating them just a few times may ultimately be more
computationally attractive than using MPEC to simultaneously solve
for all variables in a non-concave, large-scale, constrained maximization
problem.

\subsection{\label{subsec:Computational-Implementation}Computational Details
and Choice of Nuisance Parameter}

$k$-EPL is particularly useful when we are interested in estimating
the flow utility parameters, $\theta_{u}$. In many cases -- including
our Monte Carlo experiments -- the transition parameter, $\theta_{f}$,
is known in advance. So, we focus on estimating only the flow utility
parameters and let $\theta\equiv\theta_{u}$. Alternatively, $\theta_{f}$
can be estimated in a first stage, with $\theta_{u}$ then estimated
via partial maximum likelihood. Similarly to single-agent $k$-NPL,
$k$-EPL iterations based on this partial MLE problem yield asymptotic
equivalence and finite-sample convergence to partial MLE because the
zero Jacobian property still applies.
\begin{assumption}
\label{assu:Linear-Utility-Index}(Linear Utility Index) $u^{j}(\theta,x,a^{j},P^{-j})=h(x,a^{j},P^{-j})'\theta$.
\end{assumption}
\begin{assumption}
\label{assu:Log-Concave-CCP-Mapping}(Log-Concave CCP Mapping) $\Lambda(\cdot)$
is log-concave.
\end{assumption}
Assumption \ref{assu:Linear-Utility-Index} requires the flow utilities
to be linear in $\theta$, which is a standard assumption for dynamic
discrete choice models.\footnote{See, e.g., \citet{RustGMC1987,AgMir2002,AgMir2007,BajBenkLevin2007,PakesOstBerry2007,PesSD2008,ArcidMiller2011,EgLaiSu2015,BugniBunt2020}.}
Assumption \ref{assu:Log-Concave-CCP-Mapping} requires log-concavity
of the mapping from choice-specific values into CCPs. A sufficient
condition for this is that the distribution of private shocks, $g(\cdot)$,
is log-concave (\citet{CaplinNalebuff1991}). Consequently, Assumption
\ref{assu:Log-Concave-CCP-Mapping} is satisfied in the ubiquitous
case of logit shocks, as well as when shocks follow a normal distribution.
These two assumptions have important computational implications, which
is the focus of the rest of this section.

The choice of nuisance parameter, $Y$, does not affect the asymptotic
results in Section \ref{sec:EPL}, but it \textit{does} have tremendous
implications for computation. In order to provide a computationally
simple estimator, we use the choice-specific values as the nuisance
parameter and the equilibrium condition from Lemma \ref{lem:Rep-Lemma}.
\begin{assumption}
\label{assu:Equilibrium-in-CCV}(Equilibrium in Choice-Specific Values)
$Y\equiv v$ and $G(\theta,v)\equiv v-\Phi(\theta,v)$.
\end{assumption}
Coupled with Assumptions \ref{assu:Linear-Utility-Index} and \ref{assu:Log-Concave-CCP-Mapping},
the choice of nuisance parameter and constraint in Assumption \ref{assu:Equilibrium-in-CCV}
leads to some convenient computational properties.

By Assumption \ref{assu:Linear-Utility-Index}, we have $u^{j}(\theta,x,a^{j},P^{-j})=h(x,a^{j},P^{-j})'\theta$.
Because $P^{-j}=\Lambda^{-j}(v^{-j})$, we can re-write this in terms
of $v$: $u^{j}(\theta,x,a^{j},v^{-j})=h(x,a^{j},v^{-j})'\theta$.
Inspecting the form of $\Phi(\theta,v)$, we see that it will be linear
in $\theta$ and therefore so will $G(\theta,v)=v-\Phi(\theta,v)$:
\[
G(\theta,v)=H(v)\theta+z(v),
\]
where $H(\cdot)$ is a matrix and $z(\cdot)$ is a vector. As a result,
$\Upsilon(\theta,\hat{\gamma}_{k-1})$ is also linear in $\theta$:
\begin{align*}
\Upsilon(\theta,\hat{\gamma}_{k-1}) & =\hat{v}_{k-1}-\nabla_{v}G(\hat{\theta}_{k-1},\hat{v}_{k-1})^{-1}G(\theta,\hat{v}_{k-1})\\
 & =\hat{v}_{k-1}-\nabla_{v}G(\hat{\theta}_{k-1},\hat{v}_{k-1})^{-1}\left(H(\hat{v}_{k-1})\theta+z(\hat{v}_{k-1})\right)\\
 & =-\nabla_{v}G(\hat{\theta}_{k-1},\hat{v}_{k-1})^{-1}H(\hat{v}_{k-1})\theta\\
 & \quad\quad+\hat{v}_{k-1}-\nabla_{v}G(\hat{\theta}_{k-1},\hat{v}_{k-1})^{-1}z(\hat{v}_{k-1})\\
 & \equiv A(\hat{\gamma}_{k-1})\theta+b(\hat{\gamma}_{k-1}).
\end{align*}
 Additionally, the optimization step in Algorithm \ref{alg:EPL-Algorithm}
($k$-EPL) becomes 
\[
\hat{\theta}_{k,EPL}=\underset{\theta\in\Theta}{\textrm{arg max}}\quad N^{-1}\sum_{i=1}^{N}\sum_{t}\sum_{j}\ln\Lambda\left(\Upsilon(x_{t},a_{t}^{j};\theta,\hat{\gamma}_{k-1})\right).
\]
 It turns out that this is a concave optimization problem, as described
in the next proposition.
\begin{prop}
\label{prop:Computation-Simplicity}Under Assumptions \ref{assu:Linear-Utility-Index}-\ref{assu:Equilibrium-in-CCV},
(i) $\Upsilon(\theta,\hat{\gamma}_{k-1})=A(\hat{\gamma}_{k-1})\theta+b(\hat{\gamma}_{k-1})$,
where $A(\hat{\gamma}_{k-1})\equiv-\nabla_{v}G(\hat{\theta}_{k-1},\hat{v}_{k-1})^{-1}H(\hat{v}_{k-1})$
and $b(\hat{\gamma}_{k-1})\equiv\hat{v}_{k-1}-\nabla_{v}G(\hat{\theta}_{k-1},\hat{v}_{k-1})^{-1}z(\hat{v}_{k-1})$;
and (ii) For $k$-EPL, $\hat{\theta}_{k}=\textrm{arg max}_{\theta\in\Theta}\quad N^{-1}\sum_{i=1}^{N}\sum_{t}\sum_{j}\ln\Lambda\left(\Upsilon(x_{t},a_{t}^{j};\theta,\hat{\gamma}_{k-1})\right),$
where the objective function is concave in $\theta$.
\end{prop}
\begin{proof}
Result (i) follows from the analysis immediately preceding the proposition.
Result (ii) arises because $\ln\Lambda(\cdot)$ is concave by Assumption
\ref{assu:Log-Concave-CCP-Mapping} and $\Upsilon(\cdot)$ is linear
in $\theta$ (Result (i)).
\end{proof}
Proposition \ref{prop:Computation-Simplicity} shows how our choice
of nuisance parameter and constraint lead to a computationally simple
estimation sequence. Computing $A(\hat{\gamma}_{k-1})$ and $b(\hat{\gamma}_{k-1})$
in Proposition \ref{prop:Computation-Simplicity} requires computing
$\nabla_{v}G(\hat{\theta}_{k-1},\hat{v}_{k-1})^{-1}H(\hat{v}_{k-1})$
and $\nabla_{v}G(\hat{\theta}_{k-1},\hat{v}_{k-1})^{-1}z(\hat{v}_{k-1})$,
respectively, which are the solutions to linear systems. Importantly,
these linear systems can be solved outside the optimization search
in Step 2 of Algorithm \ref{alg:EPL-Algorithm}. So, the computation
procedure alternates between i) computing ``pseudo-regressors''
by solving linear systems; and ii) maximizing a concave optimization
problem that using the pseudo-regressors as inputs. Furthermore, $\nabla_{v}G(\hat{\theta}_{k-1},\hat{v}_{k-1})$
can be computed analytically when $\Lambda(\cdot)$ and its derivative
have an analytic form, such as the logit or probit cases.

Notably, this computational simplicity is \textit{not} available
if the CCPs are chosen as the nuisance parameter. That is, when $Y\equiv p$
and $G(\theta,P)\equiv P-\Psi(\theta,P)$. In this case, 
\[
\Upsilon(\theta,\hat{\gamma}_{k-1})=\hat{P}_{k-1}-\left(I-\nabla_{P}\Psi(\hat{\theta}_{k-1},\hat{P}_{k-1})\right)^{-1}\left(P-\Psi(\theta,\hat{P}_{k-1})\right)
\]
 and the optimization step in the $k$-EPL algorithm solves 
\[
\hat{\theta}_{k,\text{EPL}}=\underset{\theta\in\Theta}{\textrm{arg max}}\quad N^{-1}\sum_{i=1}^{N}\sum_{t}\sum_{j}\ln\Upsilon(x_{t},a_{t}^{j};\theta,\hat{\gamma}_{k-1}).
\]
 Several computational issues arise. First, instead of solving linear
systems once before the optimization step, we must repeatedly solve
linear systems throughout the optimization because of the need to
compute $(I-\nabla_{P}\Psi(\hat{\theta}_{k-1},\hat{P}_{k-1}))^{-1}\Psi(\theta,\hat{P}_{k-1})$
for each new value of $\theta$ in the search. Second, we will lose
the guarantee of concavity of the optimization problem in each step.
Even though $\Psi(\theta,\hat{P}_{k-1})$ is log-concave in $\theta$,
this does not guarantee log-concavity of $\Upsilon(\theta,\hat{\gamma}_{k-1})$
because affine transformations of log-concave functions are not necessarily
log-concave. And third, the Newton-like steps rely on an implicit
linearization: even though $\Psi(\theta,\hat{P}_{k-1})$ maps into
the probability simplex, $\Upsilon(\theta,\hat{\gamma}_{k-1})$ can
arrive at values outside the simplex.\footnote{Technically, $\Psi(\cdot)$ maps into a Cartesian product of the interior
of the unit simplex due to each player having their own strategies.} Thus, we would need to add constraints to the optimization problem
to ensure $\Upsilon(\theta,\hat{\gamma}_{k-1})$ does not leave the
unit simplex. Whereas, the formulation in $v$-space does not require
any constraints because $v$ can take any value on a Cartesian product
of the real line.

\subsubsection{\label{subsec:Comparison-to-Other-Methods}Comparison to Other Methods}

While our EPL iterations with $Y\equiv v$ have a similar computational
structure to NPL iterations (with $Y\equiv P$) insofar as both require
solving linear systems then a globally concave optimization problem,
the dimension of the linear systems in $k$-EPL is larger. \citet{AgMir2007}
show that $k$-NPL requires solving $|J|(|\Theta|+1)$ different systems
of linear equations, each of dimension $|\mathcal{X}|$ , resulting
in a worst-case bound of $O((|\Theta|+1)|\mathcal{X}|^{3}|\mathcal{J}|)$
flops.\footnote{There are $|\Theta|+1$ systems for each of the $|\mathcal{J}|$ players.}
On the other hand, $k$-EPL requires solving $|\Theta|+1$ different
systems of linear equations, each of dimension $|\mathcal{J}||\mathcal{X}||\mathcal{A}|$,
resulting in a larger worst-case bound of $O((|\Theta|+1)|\mathcal{X}|^{3}|\mathcal{J}|^{3}|\mathcal{A}|^{3})$
flops. Sparsity of the linear systems -- a common feature in dynamic
discrete choice models (\citet{EgLaiSu2015}) -- can lower these
bounds for both $k$-NPL and $k$-EPL. Fortunately, in practice, in
our largest Monte Carlo experiment the relative difference is much
lower than suggested by the worst-case bounds.

We then see a tradeoff between efficiency and computational burden,
a common theme in estimating dynamic games. This theme appears in
Bugni and Bunting's \citeyearpar{BugniBunt2020} analysis comparing
their efficient $k$-MD (minimum distance) estimator to $k$-NPL.
Even in what is essentially the smallest scale game possible -- 2
players, 2 actions, 4 states -- they report for $k$-MD a large increase
in computational burden over $k$-NPL, with individual iterations
taking about 12 to 26 times longer on average, depending on the number
of iterations. It is perhaps reasonable to expect that this difference
will grow with the size of the game, which is a concern for practitioners
who must balance econometric efficiency with computational feasibility.
Our $k$-EPL estimator, on the other hand, induces a much less severe
increase in computational burden while retaining efficiency. In Section
\ref{sec:Monte-Carlo-Simulations}, we also explore a 2 player, 2
action, 4 state game and find that the difference between computational
time for $k$-EPL and $k$-NPL iterations is negligible in that setting.
Additionally, we explore a much larger-scale game -- 5 players, 2
actions, 160 states -- that is more representative of empirically
relevant models. In this larger-scale game, we find that there is
indeed an increase in time per iteration for $k$-EPL relative to
$k$-NPL, but this increase in the larger-scale game (between 4 to
8 times) is not even as large as the 12 to 26-fold increase for $k$-MD
in the much smaller-scale game.

Even with an increase in computation time per iteration relative to
$k$-NPL, $k$-EPL can still ultimately be more attractive than $k$-NPL.
First, its asymptotic efficiency, convergence properties, and rapid
finite-sample improvements are attractive features that may be worth
the increased computational burden of each iteration. Second, even
in cases where both $k$-EPL and $k$-NPL converge to consistent estimates,
$k$-EPL enjoys a much faster convergence rate than $k$-NPL, resulting
in fewer iterations to convergence. So, iterating to convergence on
$k$-EPL to obtain the finite-sample MLE can still be faster than
computing the $\infty$-NPL estimator (if it converges), even though
each individual iteration takes longer.

In many applications, the dominant source of computational burden
for either estimator will often be the size of the state space, $|\mathcal{X}|$,
since it can be large when $|\mathcal{A}|$ and $|\mathcal{J}|$ are
small and also tends to grow with both $|\mathcal{A}|$ and $|\mathcal{J}|$
in dynamic games. One simple yet salient illustration arises when
the state is determined by the previous actions of the players, so
that $|\mathcal{X}|=|\mathcal{A}|^{|\mathcal{J}|}$. Thus, as $|\mathcal{J}|$
grows, $|\mathcal{X}|$ ultimately becomes the main source of computational
burden for the linear systems in both $k$-NPL and $k$-EPL. To help
deal with large state spaces, \citet{AgMir2007} show that the linear
systems required for $k$-NPL can be solved via an iterative process
reminiscent of value function Bellman iteration, so that their worst-case
computational burden reduces to $O((|\Theta|+1)|\mathcal{X}|^{2}|\mathcal{J}|)$.
Similarly, $k$-EPL can also use alternative iterative methods to
solve the linear systems such as Krylov subspace methods, although
it cannot use Bellman-style iteration.

\subsection{\label{subsec:Single-Agent-Equivalence}Single-Agent Dynamic Discrete
Choice}

We conclude this section by showing that $k$-NPL in a single-agent
dynamic discrete choice model (\citealp{AgMir2002}) is equivalent
to $k$-EPL with a slightly modified definition of $\Upsilon(\cdot)$.
Here, we can work directly in probability space, $Y\equiv P$, and
let 
\[
G(\theta,P)=P-\Psi(\theta,P),
\]
 with $\Psi(\theta,P)$ defined as in Section \ref{sec:Dynamic-Discrete-Game}
but with only a single agent.

We now have
\[
\nabla_{P}G(\theta,P)=I-\nabla_{P}\Psi(\theta,P).
\]

Proposition 2 from \citet{AgMir2002} shows that $\nabla_{P}\Psi(\theta,P_{\theta})=0$,
where $P_{\theta}=\Psi(\theta,P_{\theta})$. Thus, $\nabla_{P}G(\theta,P_{\theta})=I$
for all $\theta$. So, we can use a modified definition of $\Upsilon(\cdot)$,
where $\nabla_{P}G(\hat{\theta}_{k-1},\hat{P}_{k-1})$ is simply replaced
with the identity matrix, $I$, and we obtain
\begin{align*}
\Upsilon(\theta,\hat{\gamma}_{k-1}) & =\hat{P}_{k-1}-I^{-1}\left(\hat{P}_{k-1}-\Psi(\theta,\hat{P}_{k-1})\right)\\
 & =\Psi(\theta,\hat{P}_{k-1}).
\end{align*}
This modified implementation of $k$-EPL is identical to $k$-NPL.

This equivalence of $k$-NPL to $k$-EPL in single agent models is
unsurprising for two reasons. First, we stated in the introduction
that the motivation for $k$-EPL is to extend the nice properties
of $k$-NPL from single-agent models to dynamic games. So there should
be, at the very least, substantial conceptual overlap between the
methods. Second, Aguirregabiria and Mira (\citeyear{AgMir2002}, Proposition
1(c)) show that their policy iterations are equivalent to Newton-like
iterations on the (ex-ante) value function in single-agent models.
Since $k$-EPL is built around Newton iterations, such an equivalence
is again suggestive of the relationship shown here.

\section{\label{sec:Monte-Carlo-Simulations}Monte Carlo Simulations}

In this section, we present Monte Carlo simulation results to illustrate
$k$-EPL's finite sample properties. The simulations presented here
are based on the dynamic game of entry and exit in Example \ref{exa:Wholesale-Club-Store},
parameterized to match a model with five heterogeneous firms from
\citet{AgMir2007}. The appendix includes further Monte Carlo simulations
for two other models: (i) a small-scale dynamic model from \citet{PesSD2008};
and (ii) a static game from \citet{PesSD2010}. The dynamic model
in \citet{PesSD2008} exhibits multiple -- possibly best-reply-unstable
-- equilibria (with data generated from only one of them), which
can be challenging for other iterative methods. The static game in
\citet{PesSD2010} provides a setting where we can easily compare
$\infty$-EPL to the MLE computed via the nested fixed-point algorithm.

The simulations in this section are based on an empirically relevant
model that forms the basis of many applications but has also been
used (sometimes in simplified forms) in simulation studies by \citet{KasShim2012},
\citet{EgLaiSu2015}, \citet{BugniBunt2020}, \citet{BlevKim2019},
and \citet{AgMarcoux2019}. In particular, \citet{AgMarcoux2019}
discuss in detail how the spectral radius of the NPL operator increases
with the strength of competition in the model. As such, we take as
our baseline case the parameters as Experiment 2 of \citet{AgMir2007}.
We then follow \citet{AgMarcoux2019} and increase the competitive
effect parameter, $\theta_{\text{RN}}$ to investigate how $k$-EPL
and $k$-NPL behave as the spectral radius of the NPL operator increases
to the point of instability and beyond.

The model is a dynamic entry-exit game with $\lvert\mathcal{J\rvert}=5$
firms that operate in $N$ independent markets. The firms have heterogeneous
fixed costs. There is a single common market state, the market size,
which can take one of five values: $s_{it}\in\lbrace1,2,3,4,5\rbrace$.
Market size follows a $5\times5$ transition matrix and we use the
same transition matrix as \citet{AgMir2007}. The other observable
states are the incumbency statuses of the five firms, denoted $a_{i,t-1}^{j}$.
There are therefore $5\times2^{5}=160$ distinct states in the model.
Hence the state in market $i$ at time $t$ can be represented in
vector form as $x_{it}=(s_{it},a_{i,t-1}^{1},a_{i,t-1}^{2},a_{i,t-1}^{3},a_{i,t-1}^{4},a_{i,t-1}^{5})$.

Given the state of the model at the beginning of the period, firms
simultaneously choose whether to operate in the market, $a_{it}^{j}=1$,
or not, $a_{it}^{j}=0.$ They make these decisions in order to maximize
expected discounted profits, where the period profit function for
an active firm is

\[
\bar{u}^{j}(x_{it},a_{it}^{j}=1,a_{it}^{-j};\theta)=\theta_{\text{FC},j}+\theta_{\text{RS}}s_{it}-\theta_{\text{RN}}\ln\left(1+\sum_{l\neq j}a_{it}^{l}\right)-\theta_{\text{EC}}(1-a_{i,t-1}^{j})
\]
and $\bar{u}^{j}(x_{it},a_{it}^{j}=0,a_{it}^{-j};\theta)=0$ for inactive
firms. The game is dynamic because firms must pay an entry cost $\theta_{\text{EC}}$
to enter the market and because firms have forward-looking expectations
about market size and entry decisions of rival firms. The private
information shocks $\varepsilon_{it}^{j}(a_{it}^{j})$ are independent
and identically distributed across time, markets, players, and actions
and follow the standard type I extreme value distribution.

We choose the model parameters following \citet{AgMir2007} and \citet{AgMarcoux2019}.
In particular, the fixed costs for the five firms are $\theta_{\text{FC},1}=-1.9$,
$\theta_{\text{FC},2}=-1.8$, $\theta_{\text{FC},3}=-1.7$, $\theta_{\text{FC},4}=-1.6$,
and $\theta_{\text{FC},5}=-1.5$. The coefficient on market size is
$\theta_{\text{RS}}=1$ and the common firm entry cost is $\theta_{\text{EC}}=1$.
The only parameter that differs across our three experiments is the
competitive effect $\theta_{\text{RN}}$, which we set to be $\theta_{\text{RN}}=1$
in Experiment 1, $\theta_{\text{RN}}=2.5$ in Experiment 2, and $\theta_{\text{RN}}=4$
in Experiment 3. For easy comparison, these parameter values correspond
closely to the ``very stable,'' ``mildly unstable,'' and ``very
unstable'' cases of \citet{AgMarcoux2019}.

For each experiment, we carry out $1000$ replications using two sample
sizes, $N=1600$ and $N=6400$, noting that $N=1600$ is the sample
size used by \citet{AgMir2007}. For each replication and each sample
size, we draw a sample of size $N$. With this sample we calculate
the iterative $k$-NPL and $k$-EPL estimates. For $k$-NPL, we follow
the original \citet{AgMir2007} implementation. We initialize $k$-NPL
with estimated semiparametric logit choice probabilities. We then
initialize $1$-EPL using the parameter estimates and value function
from the $1$-NPL iteration, which is consistent even in cases where
further iteration may lead to inconsistency. Aside from the initialization,
and using the same sample, the $k$-EPL iterations proceed independently
from the $k$-NPL iterations. For $k$-EPL, as before we represent
the equilibrium condition in terms of $v$ as $G(\theta,v)=v-\Phi(\theta,v)$
and use analytical derivatives for the Jacobian $\nabla_{v}G(\theta,v)$.

We report estimates from one, two, and three iterations of each estimator
as well as the converged values, which we denote as $\infty$-NPL
and $\infty$-EPL. We limit the number of iterations to 100 for both
estimators, and if the algorithm has not converged we use the estimate
from the final iteration. We use the same convergence criteria for
both estimators: at each iteration we check the sup norm of the change
in the parameter values and choice probabilities. Based on the criteria
used by \citet{AgMir2007}, if both are below $10^{-2}/K$, where
$K$ is the number of parameters, we terminate the iterations and
return the final converged estimate.

\begin{figure}
\begin{minipage}[t]{0.49\columnwidth}%
\subfloat[Experiment 1: $\hat{\theta}_{\text{RN}}$ when $\theta_{\text{RN}}=1$]{\includegraphics[width=1\textwidth]{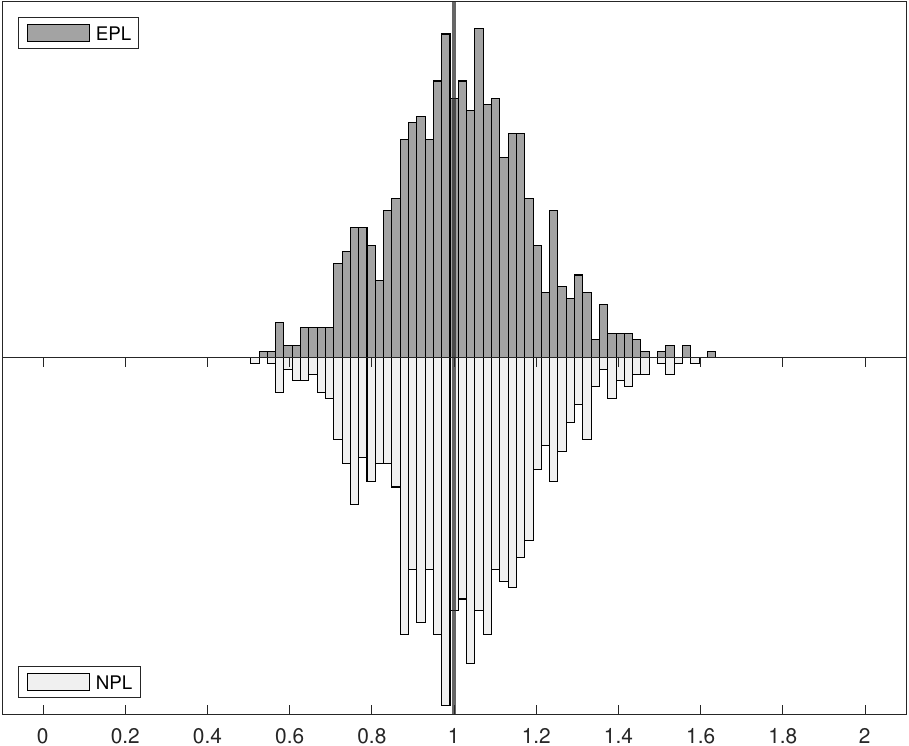}}

\subfloat[Experiment 2: $\hat{\theta}_{\text{RN}}$ when $\theta_{\text{RN}}=2.5$]{\includegraphics[width=1\textwidth]{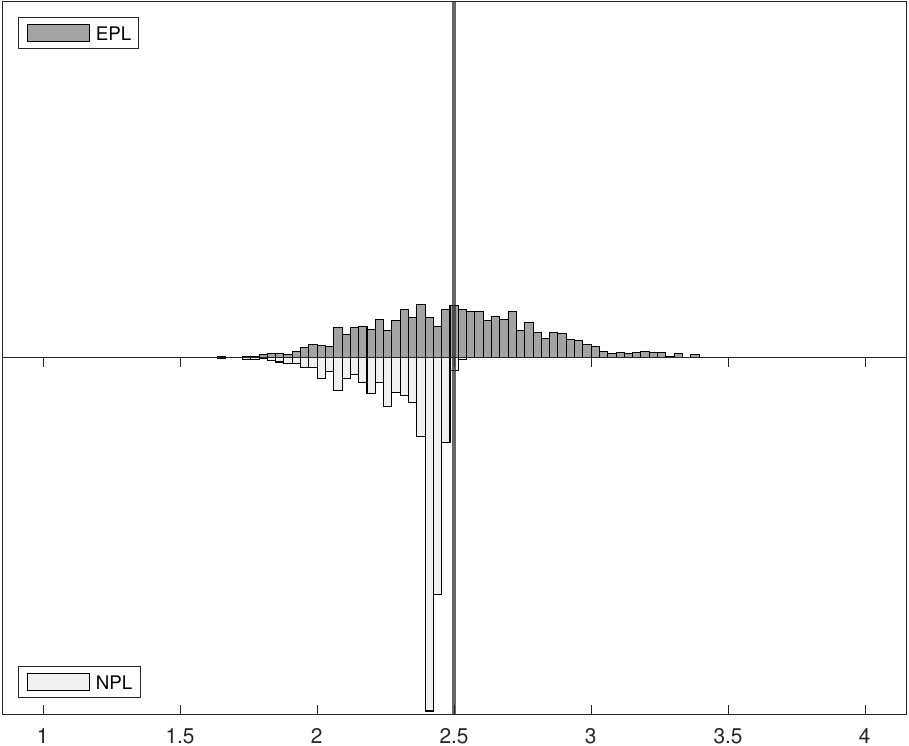}}

\subfloat[Experiment 3: $\hat{\theta}_{\text{RN}}$ when $\theta_{\text{RN}}=4$]{\includegraphics[width=1\textwidth]{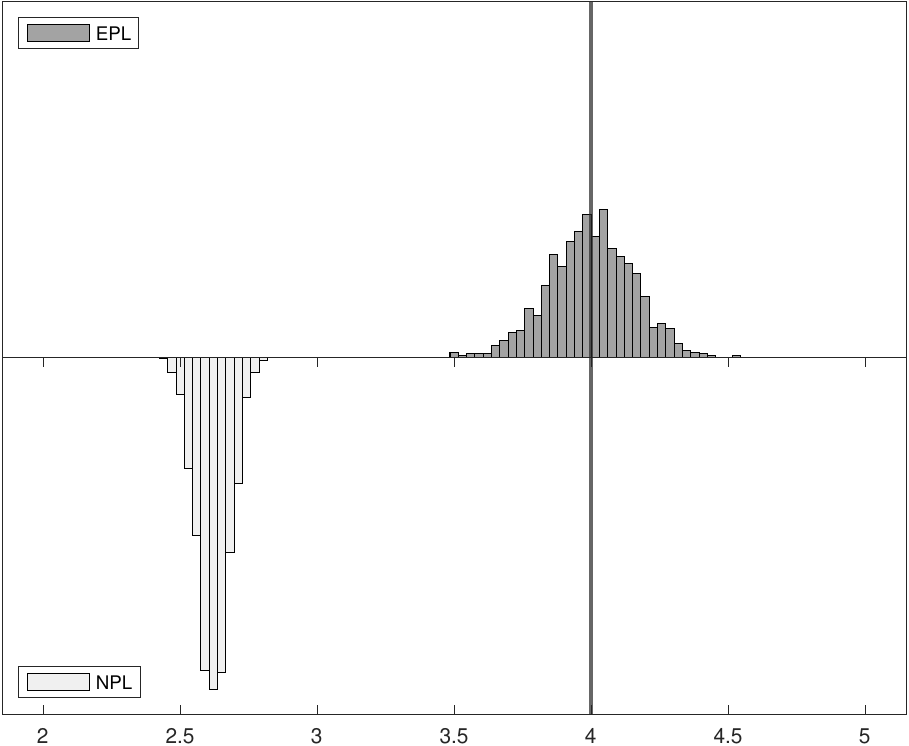}}%
\end{minipage}\hfill{}%
\begin{minipage}[t]{0.49\columnwidth}%
\subfloat[Experiment 1: $\hat{\theta}_{\text{EC}}$ when $\theta_{\text{RN}}=1$]{\includegraphics[width=1\textwidth]{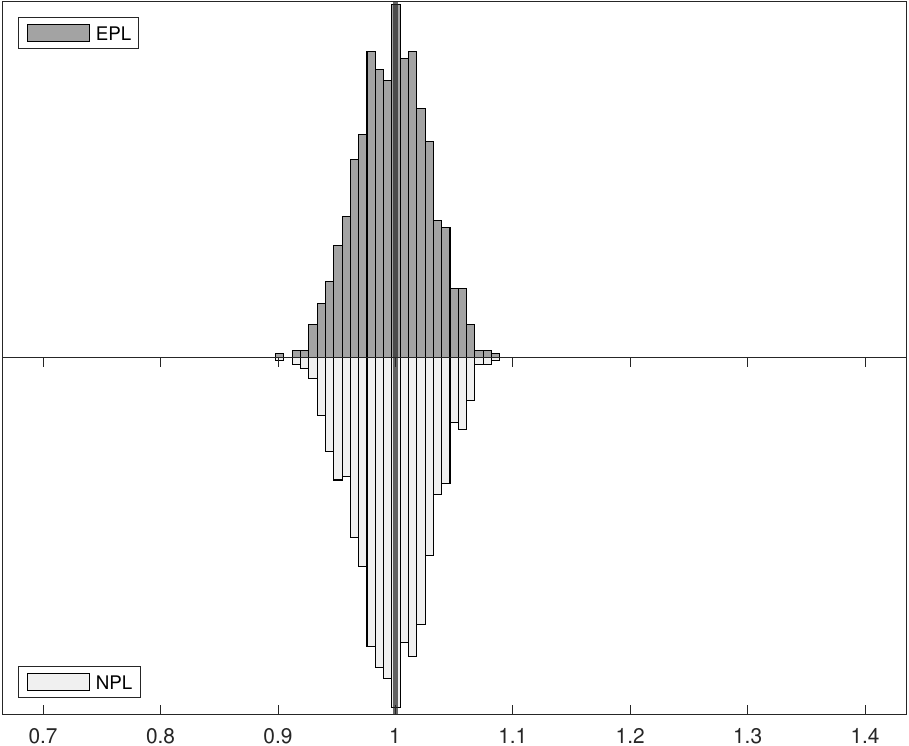}}

\subfloat[Experiment 2: $\hat{\theta}_{\text{EC}}$ when $\theta_{\text{RN}}=2.5$]{\includegraphics[width=1\textwidth]{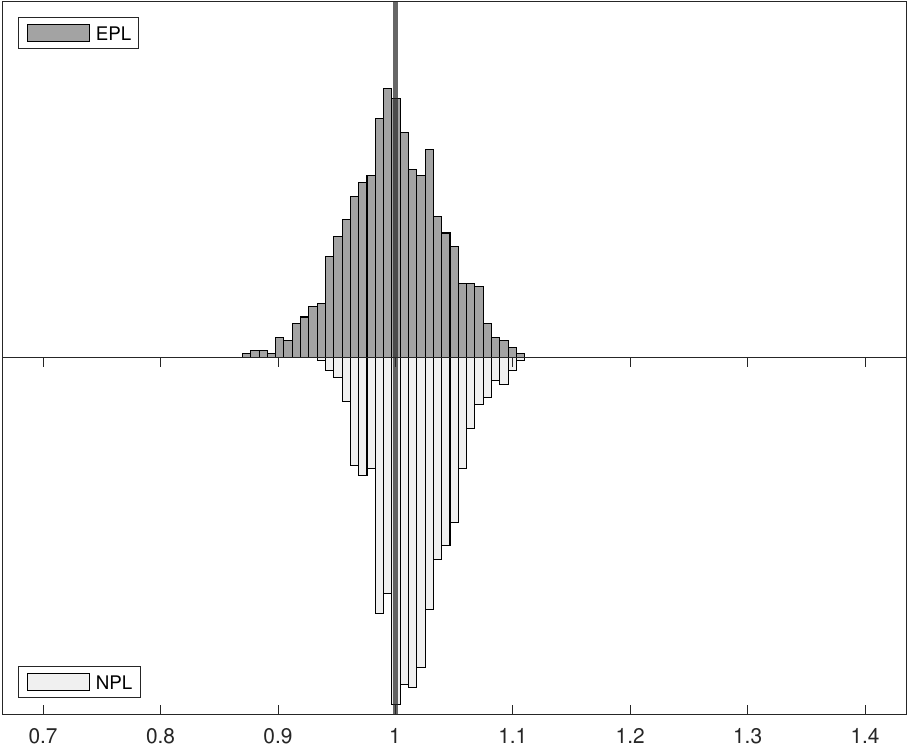}}

\subfloat[Experiment 3: $\hat{\theta}_{\text{EC}}$ when $\theta_{\text{RN}}=4$]{\includegraphics[width=1\textwidth]{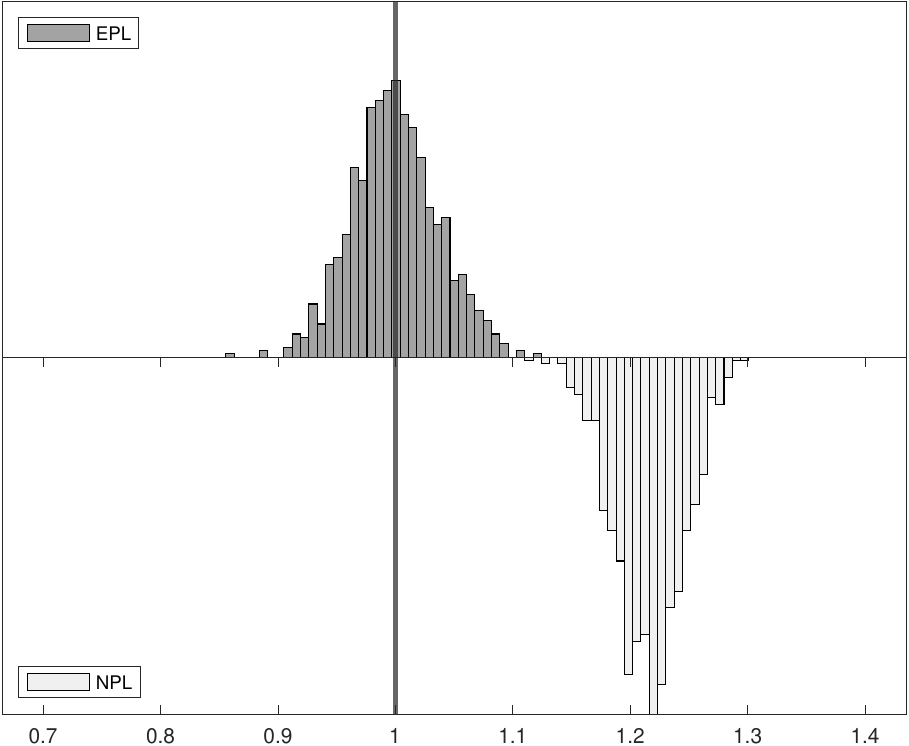}}%
\end{minipage}

\caption{\label{fig:ParameterHistograms}Distributions of $\hat{\theta}_{\text{RN}}$
and $\hat{\theta}_{\text{EC}}$ for $\infty$-EPL and $\infty$-NPL
Over 1000 Replications ($N=6400$)}
\end{figure}

Figure\ \ref{fig:ParameterHistograms} shows the Monte Carlo distributions
of two key parameter estimates in the model: $\hat{\theta}_{\text{RN}}$,
the competitive effect, and $\hat{\theta}_{\text{EC}}$, the entry
cost. Each panel compares two histograms: the histogram above in dark
gray corresponds to $\infty$-EPL and the histogram below in light
gray is for $\infty$-NPL. These histograms are based on data for
1,000 replications of each experiment with 6,400 observations.

The left panels of Figure\ \ref{fig:ParameterHistograms} show the
distributions of $\hat{\theta}_{\text{RN}}$, which is related to
the strength of competition in the market and is closely related with
the spectral radius of the NPL operator (while the spectral radius
of the EPL operator is always zero). Indeed, we can see that as the
competitive effect becomes large the distribution of $\infty$-NPL
estimates is truncated at around $\hat{\theta}_{\text{RN}}=2.5$ in
Experiment 2. It remains concentrated around $\hat{\theta}_{\text{RN}}=2.6$
in Experiment 3 even though the value used in the data generating
process was $\theta_{\text{RN}}=4$. Yet for all experiments the distributions
of $\infty$-EPL estimates are centered around the true values.

Although the true entry cost parameter is fixed at $\theta_{\text{EC}}=1$
in all three experiments, the distribution of estimates can be affected
when we vary the competitive effect from $\theta_{\text{RN}}=1$ to
$\theta_{\text{RN}}=4$. In the right panels of Figure\ \ref{fig:ParameterHistograms},
the truncation from above of $\hat{\theta}_{\text{RN}}$ in the $\infty$-NPL
distribution induces truncation from below in $\hat{\theta}_{\text{EC}}$,
as lower estimated values of competition result in higher estimated
entry costs, leading to distortion of both distributions and to parameter
estimates that would lead a policy-maker to possibly very different
economic implications. In Experiment 3, although the distribution
of $\infty$-NPL estimates again appear to be normally distributed
for both parameters, they are biased and are not centered around the
true values.

\begin{figure}
\begin{minipage}[t]{0.49\columnwidth}%
\subfloat[Experiment 1: Iterations]{\includegraphics[width=1\textwidth]{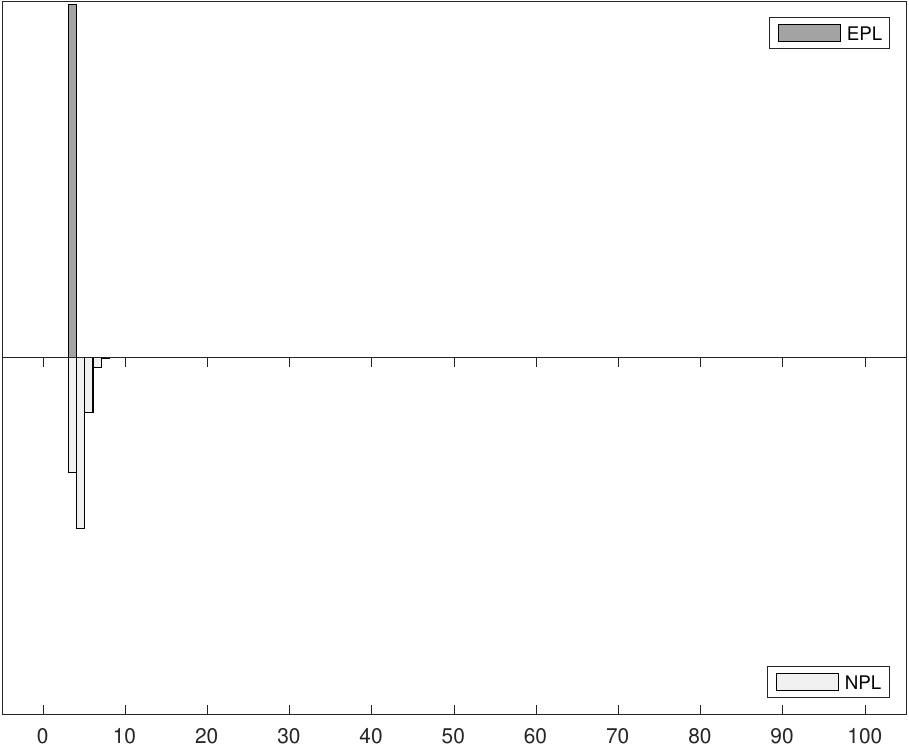}}

\subfloat[Experiment 2: Iterations]{\includegraphics[width=1\textwidth]{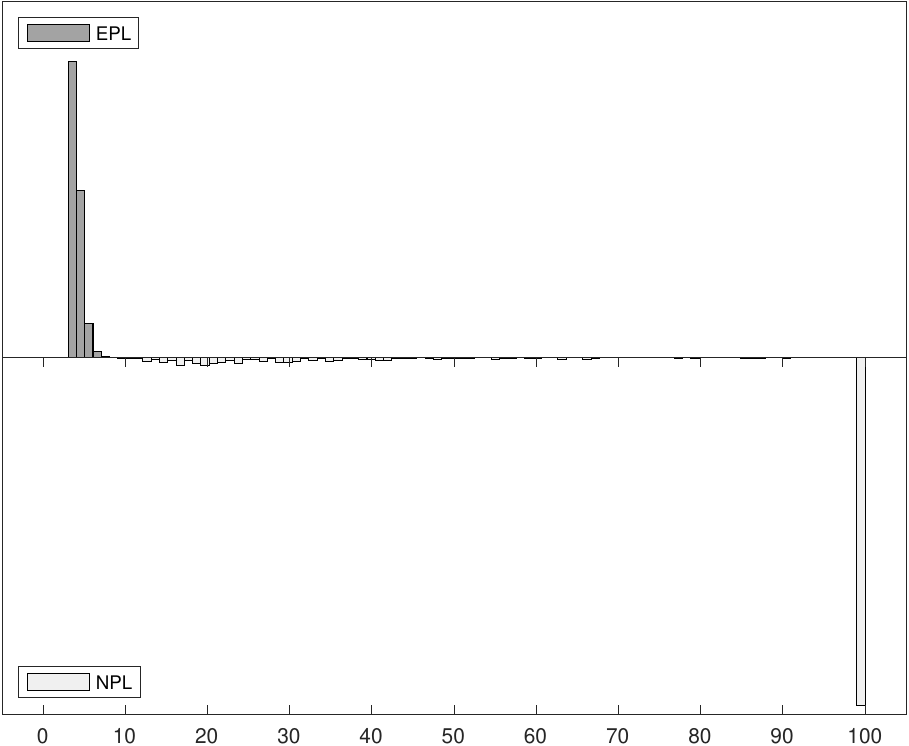}}

\subfloat[Experiment 3: Iterations]{\includegraphics[width=1\textwidth]{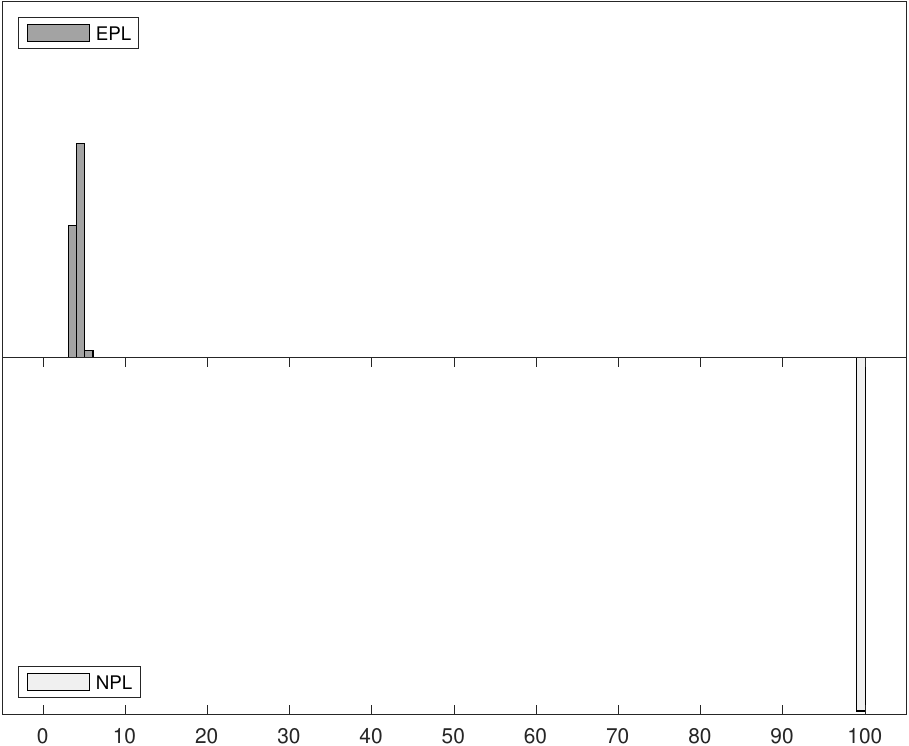}}%
\end{minipage}\hfill{}%
\begin{minipage}[t]{0.49\columnwidth}%
\subfloat[Experiment 1: Computational Times (sec.)]{\includegraphics[width=1\textwidth]{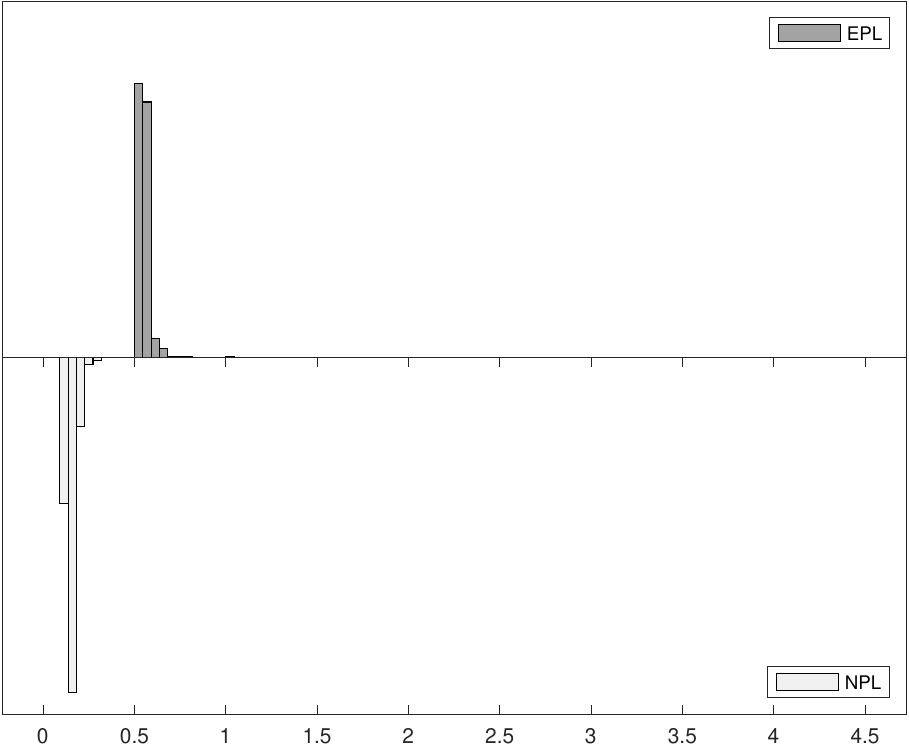}}

\subfloat[Experiment 2: Computational Times (sec.)]{\includegraphics[width=1\textwidth]{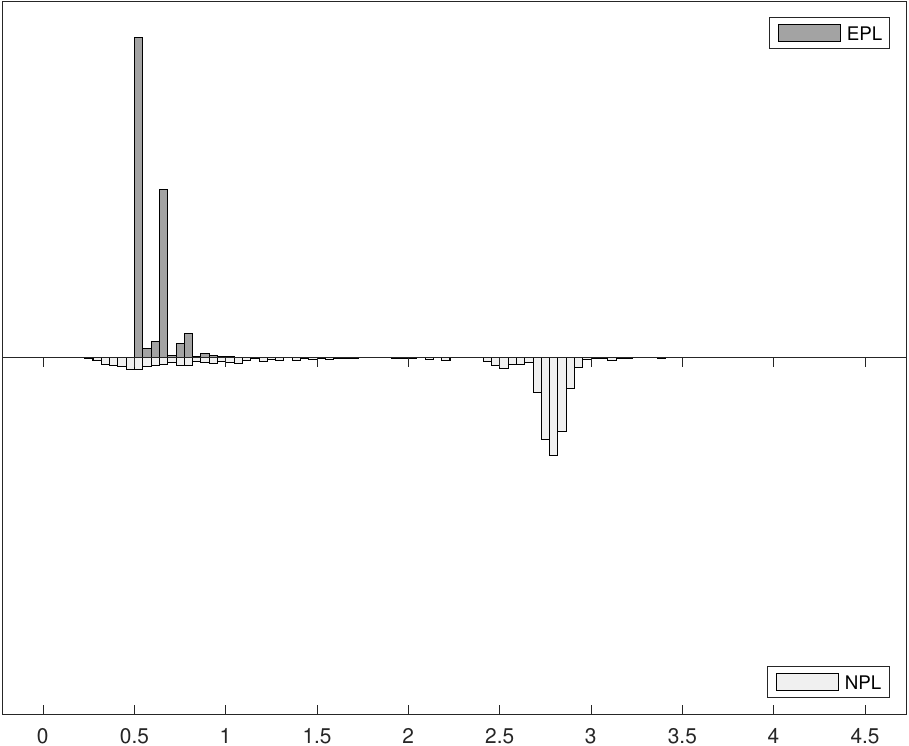}}

\subfloat[Experiment 3: Computational Times (sec.)]{\includegraphics[width=1\textwidth]{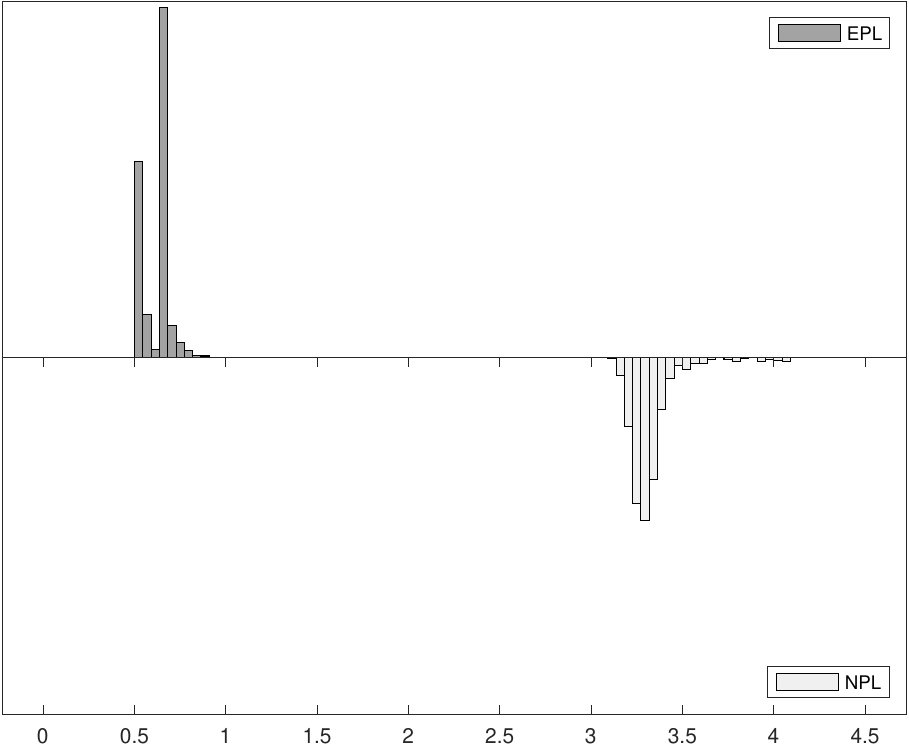}}%
\end{minipage}

\caption{\label{fig:IterationsTimes}Distributions of Iterations and Computational
Time for $\infty$-EPL and $\infty$-NPL Over 1000 Replications ($N=6400$)}
\end{figure}

Figure\ \ref{fig:IterationsTimes} shows the distributions of iteration
counts and computational times for $\infty$-EPL and $\infty$-NPL
across the three experiments.\footnote{Computational times are reported using Matlab R2020b on a 2019 Mac
Pro with a 3.5 GHz 8-Core Intel Xeon W processor.} The histograms in the left panels show the distribution of iteration
counts required to achieve convergence for both estimators. Recall
that the maximum number of iterations allowed was 100. $\infty$-EPL
requires fewer iterations for all experiments, especially for Experiments
2 and 3 where $\infty$-NPL sometimes fails to converge in Experiment
2 and always fails to converge in Experiment 3. $\infty$-EPL converged
for all replications in all experiments.

The histograms in the right panels of Figure\ \ref{fig:IterationsTimes}
show the distribution of total computational time, in seconds, across
the 1000 replications. In Experiment 1, where $k$-NPL is stable,
$\infty$-NPL is faster even though it requires more iterations on
average. In other words, each $k$-EPL iteration is more expensive
on average, but fewer are required. This is in line with the analysis
in Section~\ref{subsec:Computational-Implementation}. However, in
Experiments 2 and 3 the computational times for $\infty$-NPL increase
as it requires more iterations, eventually overtaking the time required
for $\infty$-EPL yet still frequently (Experiment 2) or always (Experiment
3) failing to converge.

Summary statistics for all parameter estimates and all three experiments
can be found in Tables~\ref{tab:exp1small}--\ref{tab:exp3large}.
Each table reports summary statistics for 1-, 2-, 3-, and $\infty$-NPL
and 1-, 2-, 3-, and $\infty$-EPL over the 1000 Monte Carlo replications
for $N=1600$ or $N=6400$. The upper two panels report the mean bias
and mean square error (MSE) for each parameter and each sequential
estimator. The next panel reports summary statistics for the number
of iterations completed until either convergence or failure at 100
iterations. This includes the median, maximum, and inter-quartile
range (IQR) of iteration counts across the replications as well as
the number of replications for which $\infty$-NPL failed to converge.
Finally, the bottom panel reports the computational times for each
estimator including the mean and median total time (for all completed
iterations) across the 1000 replications and the median time per iteration.

\begin{table}
\begin{tabular}{lrrrrrrrrr}
\hline 
\hline Bias & True & 1-NPL & 1-EPL & 2-NPL & 2-EPL & 3-NPL & 3-EPL & $\infty$-NPL & $\infty$-EPL\tabularnewline
\hline 
$\theta_{\text{FC},1}$ & -1.9 & -0.018 & 0.007 & 0.008 & 0.004 & 0.002 & 0.004 & 0.004 & 0.005\tabularnewline
$\theta_{\text{FC},2}$ & -1.8 & -0.016 & 0.007 & 0.007 & 0.003 & 0.002 & 0.004 & 0.003 & 0.004\tabularnewline
$\theta_{\text{FC},3}$ & -1.7 & -0.018 & 0.003 & 0.003 & -0.001 & -0.002 & 0.000 & -0.001 & 0.000\tabularnewline
$\theta_{\text{FC},4}$ & -1.6 & -0.015 & 0.003 & 0.004 & 0.000 & -0.001 & 0.001 & 0.000 & 0.001\tabularnewline
$\theta_{\text{FC},5}$ & -1.5 & -0.013 & 0.003 & 0.004 & 0.000 & -0.001 & 0.000 & 0.000 & 0.000\tabularnewline
$\theta_{\text{RS}}$ & 1.0 & -0.018 & 0.010 & 0.010 & 0.013 & 0.014 & 0.015 & 0.014 & 0.015\tabularnewline
$\theta_{\text{RN}}$ & 1.0 & -0.068 & 0.034 & 0.034 & 0.038 & 0.040 & 0.044 & 0.041 & 0.044\tabularnewline
$\theta_{\text{EC}}$ & 1.0 & -0.000 & 0.002 & 0.003 & -0.001 & -0.002 & -0.001 & -0.001 & -0.001\tabularnewline
\hline 
MSE & True & 1-NPL & 1-EPL & 2-NPL & 2-EPL & 3-NPL & 3-EPL & $\infty$-NPL & $\infty$-EPL\tabularnewline
\hline 
$\theta_{\text{FC},1}$ & -1.9 & 0.012 & 0.013 & 0.013 & 0.013 & 0.013 & 0.013 & 0.013 & 0.013\tabularnewline
$\theta_{\text{FC},2}$ & -1.8 & 0.012 & 0.012 & 0.012 & 0.012 & 0.012 & 0.012 & 0.012 & 0.012\tabularnewline
$\theta_{\text{FC},3}$ & -1.7 & 0.011 & 0.012 & 0.012 & 0.012 & 0.012 & 0.012 & 0.012 & 0.012\tabularnewline
$\theta_{\text{FC},4}$ & -1.6 & 0.010 & 0.011 & 0.011 & 0.011 & 0.011 & 0.011 & 0.011 & 0.011\tabularnewline
$\theta_{\text{FC},5}$ & -1.5 & 0.009 & 0.010 & 0.010 & 0.009 & 0.009 & 0.009 & 0.009 & 0.009\tabularnewline
$\theta_{\text{RS}}$ & 1.0 & 0.010 & 0.013 & 0.013 & 0.013 & 0.014 & 0.014 & 0.014 & 0.014\tabularnewline
$\theta_{\text{RN}}$ & 1.0 & 0.091 & 0.123 & 0.123 & 0.128 & 0.128 & 0.130 & 0.129 & 0.131\tabularnewline
$\theta_{\text{EC}}$ & 1.0 & 0.004 & 0.004 & 0.004 & 0.004 & 0.004 & 0.004 & 0.004 & 0.004\tabularnewline
\hline 
\multicolumn{2}{l}{Iterations} & 1-NPL & 1-EPL & 2-NPL & 2-EPL & 3-NPL & 3-EPL & $\infty$-NPL & $\infty$-EPL\tabularnewline
\hline 
\multicolumn{2}{l}{Median} & 1 & 1 & 2 & 2 & 3 & 3 & 5 & 4\tabularnewline
\multicolumn{2}{l}{Max} & 1 & 1 & 2 & 2 & 3 & 3 & 100 & 7\tabularnewline
\multicolumn{2}{l}{IQR} &  &  &  &  &  &  & 2 & 0\tabularnewline
\multicolumn{2}{l}{Non-Conv.} &  &  &  &  &  &  & 0.1\% & 0\%\tabularnewline
\hline 
\multicolumn{2}{l}{Time (sec.)} & 1-NPL & 1-EPL & 2-NPL & 2-EPL & 3-NPL & 3-EPL & $\infty$-NPL & $\infty$-EPL\tabularnewline
\hline 
\multicolumn{2}{l}{Total} & 17.88 & 118.42 & 31.33 & 234.82 & 44.35 & 350.07 & 76.55 & 481.70\tabularnewline
\multicolumn{2}{l}{Mean} & 0.02 & 0.12 & 0.03 & 0.23 & 0.04 & 0.35 & 0.08 & 0.48\tabularnewline
\multicolumn{2}{l}{Median} & 0.02 & 0.12 & 0.03 & 0.23 & 0.04 & 0.35 & 0.07 & 0.46\tabularnewline
\multicolumn{2}{l}{Med./Iter.} & 0.018 & 0.118 & 0.016 & 0.117 & 0.015 & 0.116 & 0.014 & 0.116\tabularnewline
\hline 
\end{tabular}

\caption{\label{tab:exp1small}Monte Carlo Results for \citet{AgMir2007}:
Experiment 1 ($\theta_{\text{RN}}=1)$, Small Sample ($N=1600$)}
\end{table}

\begin{table}
\begin{tabular}{lrrrrrrrrr}
\hline 
\hline Bias & True & 1-NPL & 1-EPL & 2-NPL & 2-EPL & 3-NPL & 3-EPL & $\infty$-NPL & $\infty$-EPL\tabularnewline
\hline 
$\theta_{\text{FC},1}$ & -1.9 & -0.018 & 0.003 & 0.003 & 0.000 & -0.001 & 0.000 & 0.000 & 0.000\tabularnewline
$\theta_{\text{FC},2}$ & -1.8 & -0.017 & 0.002 & 0.002 & -0.001 & -0.002 & -0.001 & -0.001 & -0.001\tabularnewline
$\theta_{\text{FC},3}$ & -1.7 & -0.016 & 0.001 & 0.001 & -0.002 & -0.002 & -0.001 & -0.001 & -0.001\tabularnewline
$\theta_{\text{FC},4}$ & -1.6 & -0.013 & 0.003 & 0.003 & -0.000 & -0.001 & -0.000 & -0.000 & -0.000\tabularnewline
$\theta_{\text{FC},5}$ & -1.5 & -0.012 & 0.001 & 0.001 & -0.002 & -0.003 & -0.002 & -0.002 & -0.002\tabularnewline
$\theta_{\text{RS}}$ & 1.0 & -0.022 & 0.003 & 0.003 & 0.004 & 0.005 & 0.005 & 0.005 & 0.005\tabularnewline
$\theta_{\text{RN}}$ & 1.0 & -0.078 & 0.010 & 0.010 & 0.011 & 0.012 & 0.013 & 0.013 & 0.013\tabularnewline
$\theta_{\text{EC}}$ & 1.0 & -0.001 & 0.000 & 0.001 & -0.002 & -0.003 & -0.002 & -0.002 & -0.002\tabularnewline
\hline 
MSE & True & 1-NPL & 1-EPL & 2-NPL & 2-EPL & 3-NPL & 3-EPL & $\infty$-NPL & $\infty$-EPL\tabularnewline
\hline 
$\theta_{\text{FC},1}$ & -1.9 & 0.003 & 0.003 & 0.003 & 0.003 & 0.003 & 0.003 & 0.003 & 0.003\tabularnewline
$\theta_{\text{FC},2}$ & -1.8 & 0.003 & 0.003 & 0.003 & 0.003 & 0.003 & 0.003 & 0.003 & 0.003\tabularnewline
$\theta_{\text{FC},3}$ & -1.7 & 0.003 & 0.003 & 0.003 & 0.003 & 0.003 & 0.003 & 0.003 & 0.003\tabularnewline
$\theta_{\text{FC},4}$ & -1.6 & 0.003 & 0.003 & 0.003 & 0.003 & 0.003 & 0.003 & 0.003 & 0.003\tabularnewline
$\theta_{\text{FC},5}$ & -1.5 & 0.003 & 0.003 & 0.003 & 0.002 & 0.003 & 0.002 & 0.003 & 0.002\tabularnewline
$\theta_{\text{RS}}$ & 1.0 & 0.003 & 0.003 & 0.003 & 0.003 & 0.003 & 0.003 & 0.003 & 0.003\tabularnewline
$\theta_{\text{RN}}$ & 1.0 & 0.027 & 0.030 & 0.030 & 0.030 & 0.030 & 0.030 & 0.031 & 0.030\tabularnewline
$\theta_{\text{EC}}$ & 1.0 & 0.001 & 0.001 & 0.001 & 0.001 & 0.001 & 0.001 & 0.001 & 0.001\tabularnewline
\hline 
\multicolumn{2}{l}{Iterations} & 1-NPL & 1-EPL & 2-NPL & 2-EPL & 3-NPL & 3-EPL & $\infty$-NPL & $\infty$-EPL\tabularnewline
\hline 
\multicolumn{2}{l}{Median} & 1 & 1 & 2 & 2 & 3 & 3 & 5 & 4\tabularnewline
\multicolumn{2}{l}{Max} & 1 & 1 & 2 & 2 & 3 & 3 & 8 & 4\tabularnewline
\multicolumn{2}{l}{IQR} &  &  &  &  &  &  & 1 & 0\tabularnewline
\multicolumn{2}{l}{Non-Conv.} &  &  &  &  &  &  & 0\% & 0\%\tabularnewline
\hline 
\multicolumn{2}{l}{Time (sec.)} & 1-NPL & 1-EPL & 2-NPL & 2-EPL & 3-NPL & 3-EPL & $\infty$-NPL & $\infty$-EPL\tabularnewline
\hline 
\multicolumn{2}{l}{Total} & 43.91 & 148.28 & 75.66 & 284.75 & 105.45 & 416.53 & 154.35 & 548.53\tabularnewline
\multicolumn{2}{l}{Mean} & 0.04 & 0.15 & 0.08 & 0.28 & 0.11 & 0.42 & 0.15 & 0.55\tabularnewline
\multicolumn{2}{l}{Median} & 0.04 & 0.15 & 0.07 & 0.28 & 0.10 & 0.41 & 0.15 & 0.55\tabularnewline
\multicolumn{2}{l}{Med./Iter.} & 0.043 & 0.147 & 0.037 & 0.141 & 0.034 & 0.138 & 0.031 & 0.136\tabularnewline
\hline 
\end{tabular}

\caption{\label{tab:exp1large}Monte Carlo Results for \citet{AgMir2007}:
Experiment 1 ($\theta_{\text{RN}}=1)$, Large Sample ($N=6400$)}
\end{table}

For Experiment 1, both the $k$-NPL and $k$-EPL estimators perform
equally well with low bias, as can be seen in Tables \ref{tab:exp1small}
and \ref{tab:exp1large}. The parameter with the most finite-sample
variation is also the main parameter of interest in our study: $\theta_{\text{RN}}.$
Note that in this model, in order to obtain estimates with performance
similar to the converged estimates it would suffice to stop at 3 iterations
with either estimator. Typically, both $\infty$-NPL and $\infty$-EPL
converge in 4 to 5 iterations. However, even in this specification
where $k$-NPL performs well, in one case out of 1000, $\infty$-NPL
fails to converge in 100 iterations or less while $\infty$-EPL always
converges in at most 7 iterations. In terms of computational time,
in this model due to the computational complexity, in the median experiment
one iteration of $k$-EPL is more expensive (0.3 seconds) than one
iteration of $k$-NPL (0.045 seconds). Because roughly the same number
of iterations are required in this model, the overall times for $\infty$-EPL
are also longer than for $\infty$-NPL. However, even with the increased
complexity each replication of $\infty$-EPL takes only around 1.3
seconds to estimate, so it remains quite feasible.

\begin{table}
\begin{tabular}{lrrrrrrrrr}
\hline 
\hline Bias & True & 1-NPL & 1-EPL & 2-NPL & 2-EPL & 3-NPL & 3-EPL & $\infty$-NPL & $\infty$-EPL\tabularnewline
\hline 
$\theta_{\text{FC},1}$ & -1.9 & -0.005 & -0.002 & -0.002 & -0.001 & -0.008 & -0.003 & 0.006 & -0.003\tabularnewline
$\theta_{\text{FC},2}$ & -1.8 & -0.002 & -0.004 & -0.005 & -0.002 & -0.005 & -0.004 & 0.010 & -0.004\tabularnewline
$\theta_{\text{FC},3}$ & -1.7 & 0.003 & -0.006 & -0.007 & -0.001 & -0.000 & -0.004 & 0.014 & -0.004\tabularnewline
$\theta_{\text{FC},4}$ & -1.6 & 0.007 & -0.008 & -0.009 & -0.002 & 0.002 & -0.006 & 0.016 & -0.006\tabularnewline
$\theta_{\text{FC},5}$ & -1.5 & 0.011 & -0.009 & -0.010 & -0.002 & 0.005 & -0.007 & 0.017 & -0.008\tabularnewline
$\theta_{\text{RS}}$ & 1.0 & 0.012 & 0.007 & 0.006 & -0.006 & -0.006 & 0.007 & -0.063 & 0.008\tabularnewline
$\theta_{\text{RN}}$ & 1.0 & 0.056 & 0.024 & 0.017 & -0.026 & -0.027 & 0.025 & -0.261 & 0.030\tabularnewline
$\theta_{\text{EC}}$ & 1.0 & -0.002 & -0.001 & -0.001 & 0.004 & 0.001 & -0.002 & 0.027 & -0.003\tabularnewline
\hline 
MSE & True & 1-NPL & 1-EPL & 2-NPL & 2-EPL & 3-NPL & 3-EPL & $\infty$-NPL & $\infty$-EPL\tabularnewline
\hline 
$\theta_{\text{FC},1}$ & -1.9 & 0.016 & 0.015 & 0.015 & 0.015 & 0.016 & 0.015 & 0.014 & 0.015\tabularnewline
$\theta_{\text{FC},2}$ & -1.8 & 0.016 & 0.015 & 0.015 & 0.014 & 0.015 & 0.015 & 0.013 & 0.015\tabularnewline
$\theta_{\text{FC},3}$ & -1.7 & 0.016 & 0.015 & 0.015 & 0.014 & 0.015 & 0.015 & 0.013 & 0.015\tabularnewline
$\theta_{\text{FC},4}$ & -1.6 & 0.016 & 0.015 & 0.015 & 0.014 & 0.015 & 0.015 & 0.012 & 0.015\tabularnewline
$\theta_{\text{FC},5}$ & -1.5 & 0.018 & 0.017 & 0.017 & 0.016 & 0.016 & 0.016 & 0.013 & 0.017\tabularnewline
$\theta_{\text{RS}}$ & 1.0 & 0.029 & 0.023 & 0.023 & 0.019 & 0.018 & 0.021 & 0.010 & 0.021\tabularnewline
$\theta_{\text{RN}}$ & 1.0 & 0.468 & 0.365 & 0.359 & 0.310 & 0.282 & 0.327 & 0.162 & 0.335\tabularnewline
$\theta_{\text{EC}}$ & 1.0 & 0.008 & 0.007 & 0.007 & 0.006 & 0.006 & 0.006 & 0.004 & 0.006\tabularnewline
\hline 
\multicolumn{2}{l}{Iterations} & 1-NPL & 1-EPL & 2-NPL & 2-EPL & 3-NPL & 3-EPL & $\infty$-NPL & $\infty$-EPL\tabularnewline
\hline 
\multicolumn{2}{l}{Median} & 1 & 1 & 2 & 2 & 3 & 3 & 100 & 5\tabularnewline
\multicolumn{2}{l}{Max} & 1 & 1 & 2 & 2 & 3 & 3 & 100 & 30\tabularnewline
\multicolumn{2}{l}{IQR} &  &  &  &  &  &  & 75.5 & 2\tabularnewline
\multicolumn{2}{l}{Non-Conv.} &  &  &  &  &  &  & 61.2\% & 0\%\tabularnewline
\hline 
\multicolumn{2}{l}{Time (sec.)} & 1-NPL & 1-EPL & 2-NPL & 2-EPL & 3-NPL & 3-EPL & $\infty$-NPL & $\infty$-EPL\tabularnewline
\hline 
\multicolumn{2}{l}{Total} & 18.26 & 120.53 & 32.92 & 239.61 & 47.60 & 358.72 & 1030.36 & 672.18\tabularnewline
\multicolumn{2}{l}{Mean} & 0.02 & 0.12 & 0.03 & 0.24 & 0.05 & 0.36 & 1.03 & 0.67\tabularnewline
\multicolumn{2}{l}{Median} & 0.02 & 0.12 & 0.03 & 0.24 & 0.05 & 0.36 & 1.44 & 0.60\tabularnewline
\multicolumn{2}{l}{Med./Iter.} & 0.018 & 0.120 & 0.016 & 0.119 & 0.016 & 0.119 & 0.014 & 0.118\tabularnewline
\hline 
\end{tabular}

\caption{\label{tab:exp2small}Monte Carlo Results for \citet{AgMir2007}:
Experiment 2 ($\theta_{\text{RN}}=2.5)$, Small Sample ($N=1600$)}
\end{table}

\begin{table}
\begin{tabular}{lrrrrrrrrr}
\hline 
\hline Bias & True & 1-NPL & 1-EPL & 2-NPL & 2-EPL & 3-NPL & 3-EPL & $\infty$-NPL & $\infty$-EPL\tabularnewline
\hline 
$\theta_{\text{FC},1}$ & -1.9 & -0.002 & -0.000 & -0.001 & -0.002 & -0.005 & -0.002 & 0.006 & -0.002\tabularnewline
$\theta_{\text{FC},2}$ & -1.8 & 0.002 & -0.002 & -0.003 & -0.002 & -0.002 & -0.002 & 0.008 & -0.002\tabularnewline
$\theta_{\text{FC},3}$ & -1.7 & 0.006 & -0.004 & -0.005 & -0.002 & 0.000 & -0.003 & 0.009 & -0.003\tabularnewline
$\theta_{\text{FC},4}$ & -1.6 & 0.012 & -0.004 & -0.005 & -0.001 & 0.004 & -0.002 & 0.010 & -0.002\tabularnewline
$\theta_{\text{FC},5}$ & -1.5 & 0.019 & -0.004 & -0.005 & -0.001 & 0.007 & -0.003 & 0.008 & -0.003\tabularnewline
$\theta_{\text{RS}}$ & 1.0 & -0.001 & 0.001 & 0.001 & -0.003 & -0.005 & 0.001 & -0.040 & 0.001\tabularnewline
$\theta_{\text{RN}}$ & 1.0 & 0.003 & 0.002 & -0.003 & -0.013 & -0.018 & 0.001 & -0.163 & 0.002\tabularnewline
$\theta_{\text{EC}}$ & 1.0 & 0.003 & -0.000 & 0.000 & 0.000 & 0.000 & -0.001 & 0.015 & -0.001\tabularnewline
\hline 
MSE & True & 1-NPL & 1-EPL & 2-NPL & 2-EPL & 3-NPL & 3-EPL & $\infty$-NPL & $\infty$-EPL\tabularnewline
\hline 
$\theta_{\text{FC},1}$ & -1.9 & 0.004 & 0.004 & 0.004 & 0.004 & 0.004 & 0.004 & 0.004 & 0.004\tabularnewline
$\theta_{\text{FC},2}$ & -1.8 & 0.004 & 0.004 & 0.004 & 0.004 & 0.004 & 0.004 & 0.003 & 0.004\tabularnewline
$\theta_{\text{FC},3}$ & -1.7 & 0.004 & 0.004 & 0.004 & 0.004 & 0.004 & 0.004 & 0.003 & 0.004\tabularnewline
$\theta_{\text{FC},4}$ & -1.6 & 0.004 & 0.004 & 0.004 & 0.004 & 0.004 & 0.004 & 0.003 & 0.004\tabularnewline
$\theta_{\text{FC},5}$ & -1.5 & 0.004 & 0.004 & 0.004 & 0.004 & 0.004 & 0.004 & 0.004 & 0.004\tabularnewline
$\theta_{\text{RS}}$ & 1.0 & 0.007 & 0.006 & 0.006 & 0.005 & 0.005 & 0.006 & 0.003 & 0.006\tabularnewline
$\theta_{\text{RN}}$ & 1.0 & 0.109 & 0.096 & 0.095 & 0.085 & 0.084 & 0.087 & 0.048 & 0.088\tabularnewline
$\theta_{\text{EC}}$ & 1.0 & 0.002 & 0.002 & 0.002 & 0.002 & 0.002 & 0.002 & 0.001 & 0.002\tabularnewline
\hline 
\multicolumn{2}{l}{Iterations} & 1-NPL & 1-EPL & 2-NPL & 2-EPL & 3-NPL & 3-EPL & $\infty$-NPL & $\infty$-EPL\tabularnewline
\hline 
\multicolumn{2}{l}{Median} & 1 & 1 & 2 & 2 & 3 & 3 & 100 & 4\tabularnewline
\multicolumn{2}{l}{Max} & 1 & 1 & 2 & 2 & 3 & 3 & 100 & 8\tabularnewline
\multicolumn{2}{l}{IQR} &  &  &  &  &  &  & 52 & 1\tabularnewline
\multicolumn{2}{l}{Non-Conv.} &  &  &  &  &  &  & 68.9\% & 0\%\tabularnewline
\hline 
\multicolumn{2}{l}{Time (sec.)} & 1-NPL & 1-EPL & 2-NPL & 2-EPL & 3-NPL & 3-EPL & $\infty$-NPL & $\infty$-EPL\tabularnewline
\hline 
\multicolumn{2}{l}{Total} & 35.34 & 135.62 & 64.32 & 266.26 & 92.87 & 395.24 & 2186.69 & 586.21\tabularnewline
\multicolumn{2}{l}{Mean} & 0.04 & 0.14 & 0.06 & 0.27 & 0.09 & 0.40 & 2.19 & 0.59\tabularnewline
\multicolumn{2}{l}{Median} & 0.03 & 0.13 & 0.06 & 0.26 & 0.09 & 0.39 & 2.75 & 0.53\tabularnewline
\multicolumn{2}{l}{Med./Iter.} & 0.034 & 0.135 & 0.032 & 0.132 & 0.031 & 0.131 & 0.028 & 0.130\tabularnewline
\hline 
\end{tabular}

\caption{\label{tab:exp2large}Monte Carlo Results for \citet{AgMir2007}:
Experiment 2 ($\theta_{\text{RN}}=2.5)$, Large Sample ($N=6400$)}
\end{table}

In Experiment 2, we begin to see a divergence between the two methods.
As reported by \citet{AgMarcoux2019}, the spectral radius of the
population NPL operator is slightly larger than one for this specification.
In finite random samples, sometimes the sample counterpart is stable
and sometimes it is unstable. This leads to the situation illustrated
by Table~\ref{tab:exp2small}, where $\infty$-NPL fails to converge
in 612 of 1000 replications. $\infty$-EPL, on the other hand, is
stable and converges for all 1000 replications. Importantly, the 1-NPL
estimator obtained without further iterations is always consistent.
However, there is substantial bias in the $\infty$-NPL estimates.
In an apparent contradiction, the MSE for $\theta_{\text{RN}}$ is
actually lower for $\infty$-NPL than for $\infty$-EPL. This pattern
of larger bias and lower MSE is seen again with the large sample size
in Table~\ref{tab:exp2large}, however it can be understood simply
by recalling the histogram of the $\theta_{\text{RN}}$ estimates
(panel (b) of Figure~\ref{fig:ParameterHistograms}). The $\infty$-NPL
sampling distribution appears to be truncated near 2.4, which perhaps
not coincidentally is near the value where the spectral radius exceeds
one (\citet{AgMarcoux2019}). Since this happens to be close to the
true parameter value, the MSE is artificially low. However, the sampling
distribution is neither normally distributed nor centered at the true
value. A K-S test for normality of the $\infty$-NPL estimates has
a $p$-value equal to zero up to three decimal places, while for $\infty$-EPL
the $p$-value is 0.622.

In Experiment 2 we also see a reversal of the order of computational
times: the non-convergent $\infty$-NPL cases require more iterations
and more time per iteration, while $\infty$-EPL always converges
in 8 or fewer iterations. Thus, in thinking about the trade-off between
robustness and computational time we should also consider convergence.
A non-convergent estimator may take longer and yield worse results
in the end. $k$-EPL does require more time per iteration in this
model, but it is more robust to the strength of competition in the
model.

\begin{table}
\begin{tabular}{lrrrrrrrrr}
\hline 
\hline Bias & True & 1-NPL & 1-EPL & 2-NPL & 2-EPL & 3-NPL & 3-EPL & $\infty$-NPL & $\infty$-EPL\tabularnewline
\hline 
$\theta_{\text{FC},1}$ & -1.9 & -0.100 & 0.010 & 0.008 & -0.002 & -0.095 & 0.002 & 0.012 & 0.002\tabularnewline
$\theta_{\text{FC},2}$ & -1.8 & -0.088 & 0.000 & -0.001 & -0.000 & -0.074 & 0.002 & 0.031 & 0.002\tabularnewline
$\theta_{\text{FC},3}$ & -1.7 & -0.066 & -0.010 & -0.011 & 0.000 & -0.040 & 0.001 & 0.062 & 0.002\tabularnewline
$\theta_{\text{FC},4}$ & -1.6 & -0.016 & -0.016 & -0.017 & 0.004 & 0.026 & 0.002 & 0.131 & 0.002\tabularnewline
$\theta_{\text{FC},5}$ & -1.5 & 0.050 & -0.007 & -0.005 & 0.009 & 0.156 & 0.000 & 0.211 & -0.000\tabularnewline
$\theta_{\text{RS}}$ & 1.0 & 0.042 & -0.005 & -0.019 & -0.009 & -0.108 & 0.001 & -0.244 & 0.001\tabularnewline
$\theta_{\text{RN}}$ & 1.0 & 0.198 & -0.067 & -0.147 & -0.055 & -0.640 & 0.004 & -1.382 & 0.008\tabularnewline
$\theta_{\text{EC}}$ & 1.0 & -0.009 & 0.025 & 0.031 & 0.010 & 0.073 & 0.001 & 0.218 & -0.000\tabularnewline
\hline 
MSE & True & 1-NPL & 1-EPL & 2-NPL & 2-EPL & 3-NPL & 3-EPL & $\infty$-NPL & $\infty$-EPL\tabularnewline
\hline 
$\theta_{\text{FC},1}$ & -1.9 & 0.038 & 0.022 & 0.022 & 0.023 & 0.032 & 0.023 & 0.019 & 0.023\tabularnewline
$\theta_{\text{FC},2}$ & -1.8 & 0.034 & 0.020 & 0.020 & 0.021 & 0.026 & 0.021 & 0.018 & 0.021\tabularnewline
$\theta_{\text{FC},3}$ & -1.7 & 0.030 & 0.020 & 0.019 & 0.019 & 0.021 & 0.020 & 0.020 & 0.020\tabularnewline
$\theta_{\text{FC},4}$ & -1.6 & 0.028 & 0.019 & 0.019 & 0.019 & 0.019 & 0.019 & 0.032 & 0.019\tabularnewline
$\theta_{\text{FC},5}$ & -1.5 & 0.039 & 0.022 & 0.021 & 0.019 & 0.041 & 0.019 & 0.057 & 0.019\tabularnewline
$\theta_{\text{RS}}$ & 1.0 & 0.022 & 0.006 & 0.006 & 0.004 & 0.015 & 0.004 & 0.061 & 0.004\tabularnewline
$\theta_{\text{RN}}$ & 1.0 & 0.602 & 0.126 & 0.131 & 0.090 & 0.456 & 0.086 & 1.924 & 0.086\tabularnewline
$\theta_{\text{EC}}$ & 1.0 & 0.022 & 0.007 & 0.007 & 0.005 & 0.009 & 0.005 & 0.051 & 0.005\tabularnewline
\hline 
\multicolumn{2}{l}{Iterations} & 1-NPL & 1-EPL & 2-NPL & 2-EPL & 3-NPL & 3-EPL & $\infty$-NPL & $\infty$-EPL\tabularnewline
\hline 
\multicolumn{2}{l}{Median} & 1 & 1 & 2 & 2 & 3 & 3 & 100 & 5\tabularnewline
\multicolumn{2}{l}{Max} & 1 & 1 & 2 & 2 & 3 & 3 & 100 & 9\tabularnewline
\multicolumn{2}{l}{IQR} &  &  &  &  &  &  & 0 & 1\tabularnewline
\multicolumn{2}{l}{Non-Conv.} &  &  &  &  &  &  & 100\% & 0\%\tabularnewline
\hline 
\multicolumn{2}{l}{Time (sec.)} & 1-NPL & 1-EPL & 2-NPL & 2-EPL & 3-NPL & 3-EPL & $\infty$-NPL & $\infty$-EPL\tabularnewline
\hline 
\multicolumn{2}{l}{Total} & 18.64 & 125.88 & 33.88 & 246.85 & 48.96 & 367.25 & 1512.68 & 654.05\tabularnewline
\multicolumn{2}{l}{Mean} & 0.02 & 0.13 & 0.03 & 0.25 & 0.05 & 0.37 & 1.51 & 0.65\tabularnewline
\multicolumn{2}{l}{Median} & 0.02 & 0.12 & 0.03 & 0.24 & 0.05 & 0.36 & 1.48 & 0.62\tabularnewline
\multicolumn{2}{l}{Med./Iter.} & 0.018 & 0.124 & 0.017 & 0.122 & 0.016 & 0.120 & 0.015 & 0.119\tabularnewline
\hline 
\end{tabular}

\caption{\label{tab:exp3small}Monte Carlo Results for \citet{AgMir2007}:
Experiment 3 ($\theta_{\text{RN}}=4)$, Small Sample ($N=1600$)}
\end{table}

\begin{table}
\begin{tabular}{lrrrrrrrrr}
\hline 
\hline Bias & True & 1-NPL & 1-EPL & 2-NPL & 2-EPL & 3-NPL & 3-EPL & $\infty$-NPL & $\infty$-EPL\tabularnewline
\hline 
$\theta_{\text{FC},1}$ & -1.9 & -0.098 & 0.008 & 0.007 & 0.000 & -0.098 & 0.000 & 0.013 & 0.000\tabularnewline
$\theta_{\text{FC},2}$ & -1.8 & -0.086 & -0.003 & -0.003 & 0.000 & -0.078 & 0.000 & 0.031 & 0.000\tabularnewline
$\theta_{\text{FC},3}$ & -1.7 & -0.063 & -0.015 & -0.016 & 0.001 & -0.045 & -0.001 & 0.062 & -0.001\tabularnewline
$\theta_{\text{FC},4}$ & -1.6 & -0.013 & -0.024 & -0.025 & 0.004 & 0.017 & 0.000 & 0.130 & 0.000\tabularnewline
$\theta_{\text{FC},5}$ & -1.5 & 0.057 & -0.025 & -0.023 & 0.009 & 0.139 & 0.000 & 0.207 & 0.000\tabularnewline
$\theta_{\text{RS}}$ & 1.0 & 0.034 & 0.011 & -0.003 & -0.010 & -0.096 & 0.000 & -0.243 & 0.000\tabularnewline
$\theta_{\text{RN}}$ & 1.0 & 0.155 & 0.024 & -0.054 & -0.057 & -0.570 & -0.001 & -1.377 & 0.000\tabularnewline
$\theta_{\text{EC}}$ & 1.0 & 0.001 & 0.011 & 0.016 & 0.008 & 0.060 & 0.000 & 0.216 & 0.000\tabularnewline
\hline 
MSE & True & 1-NPL & 1-EPL & 2-NPL & 2-EPL & 3-NPL & 3-EPL & $\infty$-NPL & $\infty$-EPL\tabularnewline
\hline 
$\theta_{\text{FC},1}$ & -1.9 & 0.016 & 0.005 & 0.005 & 0.005 & 0.015 & 0.005 & 0.005 & 0.005\tabularnewline
$\theta_{\text{FC},2}$ & -1.8 & 0.014 & 0.005 & 0.005 & 0.005 & 0.011 & 0.005 & 0.005 & 0.005\tabularnewline
$\theta_{\text{FC},3}$ & -1.7 & 0.011 & 0.005 & 0.005 & 0.005 & 0.007 & 0.005 & 0.008 & 0.005\tabularnewline
$\theta_{\text{FC},4}$ & -1.6 & 0.007 & 0.006 & 0.006 & 0.005 & 0.005 & 0.005 & 0.020 & 0.005\tabularnewline
$\theta_{\text{FC},5}$ & -1.5 & 0.012 & 0.006 & 0.006 & 0.005 & 0.024 & 0.005 & 0.046 & 0.005\tabularnewline
$\theta_{\text{RS}}$ & 1.0 & 0.006 & 0.002 & 0.001 & 0.001 & 0.010 & 0.001 & 0.060 & 0.001\tabularnewline
$\theta_{\text{RN}}$ & 1.0 & 0.167 & 0.037 & 0.034 & 0.023 & 0.335 & 0.023 & 1.900 & 0.023\tabularnewline
$\theta_{\text{EC}}$ & 1.0 & 0.005 & 0.002 & 0.002 & 0.001 & 0.004 & 0.001 & 0.048 & 0.001\tabularnewline
\hline 
\multicolumn{2}{l}{Iterations} & 1-NPL & 1-EPL & 2-NPL & 2-EPL & 3-NPL & 3-EPL & $\infty$-NPL & $\infty$-EPL\tabularnewline
\hline 
\multicolumn{2}{l}{Median} & 1 & 1 & 2 & 2 & 3 & 3 & 100 & 5\tabularnewline
\multicolumn{2}{l}{Max} & 1 & 1 & 2 & 2 & 3 & 3 & 100 & 6\tabularnewline
\multicolumn{2}{l}{IQR} &  &  &  &  &  &  & 0 & 1\tabularnewline
\multicolumn{2}{l}{Non-Conv.} &  &  &  &  &  &  & 100\% & 0\%\tabularnewline
\hline 
\multicolumn{2}{l}{Time (sec.)} & 1-NPL & 1-EPL & 2-NPL & 2-EPL & 3-NPL & 3-EPL & $\infty$-NPL & $\infty$-EPL\tabularnewline
\hline 
\multicolumn{2}{l}{Total} & 40.48 & 145.93 & 75.22 & 278.93 & 108.36 & 410.64 & 3332.77 & 620.36\tabularnewline
\multicolumn{2}{l}{Mean} & 0.04 & 0.15 & 0.08 & 0.28 & 0.11 & 0.41 & 3.33 & 0.62\tabularnewline
\multicolumn{2}{l}{Median} & 0.04 & 0.14 & 0.07 & 0.28 & 0.11 & 0.41 & 3.30 & 0.65\tabularnewline
\multicolumn{2}{l}{Med./Iter.} & 0.039 & 0.145 & 0.037 & 0.138 & 0.035 & 0.136 & 0.033 & 0.133\tabularnewline
\hline 
\end{tabular}

\caption{\label{tab:exp3large}Monte Carlo Results for \citet{AgMir2007}:
Experiment 3 ($\theta_{\text{RN}}=4)$, Large Sample ($N=6400$)}
\end{table}

Turning to Experiment 3, we increase the effect of competition even
further to where $\theta_{\text{RN}}=4$, well beyond the point where
$k$-NPL becomes unstable. In this case, both with small samples and
large samples, $\infty$-NPL fails to converge in all 1000 replications
but $\infty$-EPL converges in all 1000 replications. In these cases,
the bias in $\infty$-NPL estimates is larger and in this case and
so is the MSE, since the value of $\theta_{\text{RN}}$ is farther
from the point of truncation than in Experiment 2. The $\infty$-NPL
estimator systematically underestimates the competitive effect $\theta_{\text{RN}}$
and overestimates the entry cost $\theta_{\text{EC}}$.

Overall, across the three experiments the performance of $k$-EPL
is stable and of similar quality despite the increasing competitive
effect. This result agrees with our theoretical analysis showing that
$k$-EPL is stable, convergent, and efficient.

Finally, we note that with the $k$-EPL estimator there is very little
performance improvement after the first three iterations. The performance
of $\infty$-EPL is achieved already, up to two decimal places, by
3-EPL. Thus, for this model one can reduce the computational time
required while retaining efficiency and robustness by carrying out
only a few iterations of $k$-EPL.

\section{\label{sec:Empirical-Application}Application to U.S. Wholesale Club
Competition}

The U.S. wholesale club store industry is a retail segment that offers
members a wide range of merchandise at discounted prices. The industry
is dominated by three major players: Sam's Club, Costco, and BJ's
Wholesale Club. These companies operate as membership-based clubs,
with a focus on bulk purchasing and high-volume sales.

The modern wholesale club store industry emerged in the 1970s with
the founding of Price Club in San Diego, California. Price Club was
a pioneer in the industry, offering its members deep discounts on
products, including groceries, electronics, and household goods. In
1983, Costco was founded in Seattle, Washington, and quickly became
a major player in the industry. Sam's Club, a subsidiary of Walmart,
was also founded in 1983. BJ's Wholesale Club was founded in 1984
in Massachusetts. In 1993, Price Club merged with Costco and adopted
the Costco name (\citet{costco-about-us}).

Today, the wholesale club store industry is a major force in retail,
with the three main players generating over \$260 billion in annual
revenue, according to data from their 2021 annual reports. Costco
is the largest company in the industry, with over 800 stores worldwide
and annual revenue of over \$190 billion (\citet{costco-2021}). Sam's
Club is the second-largest, with around 600 stores and annual revenue
of over \$60 billion (\citet{sams-2021}). BJ's Wholesale Club is
the smallest of the three, with over 200 stores and annual revenue
of over \$16 billion (\citet{bj-2021}).

Although we focus on Costco, Sam's Club, and BJ's, which are by far
the largest firms in the industry, there are other smaller players.
An example is DirectBuy, which was the 4th largest firm (by total
markets served) in our sample. While it was a significant player in
the home furnishings and home improvement space, it was not a direct
competitor of Costco, Sam's Club, and BJ's and was significantly smaller.
DirectBuy faced financial difficulties and filed for bankruptcy in
2018 (\citet{furniture-today}).

\subsection{Data}

Our data come from the Data Axle Business Database which contains
information on businesses across the United States. The company uses
a variety of sources to gather information on businesses, including
public records, government filings, and proprietary data sources.
Our data begins in 2009 and ends in 2021. From this database we first
extract records for each Sam's Club, Costco, and BJ's location. We
augment this with ZIP code level population data obtained from NHGIS
(\citet{manson2022ipums}). We then aggregate to county-level markets
using the 2021 ZIP code to county crosswalk files provided by the
U.S. Department of Housing and Urban Development (HUD). We consider
all counties in the 50 states of the United States and the District
of Columbia that have between 20,000 and 600,000 residents. This serves
to exclude very small counties that would clearly not be considered
for entry as well as some very large, atypical markets.

Overall, our final sample consists of $N=1,600$ counties observed
over $T=12$ years.\footnote{Although we have 13 years of data, we require lagged actions to construct
the incumbency status state variable. As a result, the time dimension
of our sample is reduced to $T=12$.} Among these markets, the average peak population during our sample
was 104,841 with standard deviation 112,759.\footnote{We compute peak population for a market as the maximum over the years
in the sample.} In the model, our market size variable, $s_{it}$, is the logarithm
of population discretized into 5 equal bins.\footnote{We also tried 10 bins without any major substantive changes in the
results.} Table \ref{tab:summary_staistics} presents summary statistics for
our sample. On average, there are 0.348 active firms in each market
with a standard deviation of 0.622. The autoregressive coefficient
for the number of active firms is 0.987, indicating a strong positive
correlation between the number of active firms in the current period
and the previous period.\footnote{This refers to the estimated autoregressive coefficient in an AR(1)
regression of the number of current active firms on the number of
active firms in previous period.} The average number of entrants is 0.010, while the average number
of exits is 0.006. Excess turnover, defined as (\#entrants + \#exits)
- |\#entrants - \#exits|, is effectively zero. The correlation between
entries and exits is -0.007. The probability of being active is highest
for Sam's Club at 0.201, followed by Costco at 0.093 and BJ's at 0.054.
The distribution of market size is such that there are relatively
more small-to-medium size markets markets and relatively fewer large
markets.

\begin{table}[htbp]
\centering\caption{Wholesale Clubs: Summary Statistics}

\label{tab:summary_staistics} %
\begin{tabular}{lr}
\hline 
\hline Statistic & Value\tabularnewline
\hline 
Average active firms & 0.348\tabularnewline
S.D. active firms & 0.622\tabularnewline
AR(1) for active firms & 0.987\tabularnewline
Average entrants & 0.010\tabularnewline
Average exits & 0.006\tabularnewline
Excess turnover & 0.000\tabularnewline
Correlation between entries and exits & -0.007\tabularnewline
\hline 
Probability of being active & \tabularnewline
\quad{}Sam's Club & 0.201\tabularnewline
\quad{}Costco & 0.093\tabularnewline
\quad{}BJ's & 0.054\tabularnewline
\hline 
Distribution of market size & \tabularnewline
\quad{}$s=1$ & 0.332\tabularnewline
\quad{}$s=2$ & 0.295\tabularnewline
\quad{}$s=3$ & 0.179\tabularnewline
\quad{}$s=4$ & 0.125\tabularnewline
\quad{}$s=5$ & 0.069\tabularnewline
\hline 
Markets & 1,610\tabularnewline
Years & 12\tabularnewline
Observations (Markets $\times$ Years) & 19,320\tabularnewline
\hline 
\end{tabular}
\end{table}

\subsection{Model}

Our model of wholesale club competition follows the dynamic oligopoly
model with heterogeneous firms described in Example \ref{exa:Wholesale-Club-Store}
and also used in our Monte Carlo experiments in Section \ref{sec:Monte-Carlo-Simulations}.
In our application the firms are denoted $\mathcal{J}=\{\text{SC},\text{CC},\text{BJ}\}$.
In each market $i=1,\dots,N$ and time period $t=1,\dots,T$, firms
decide whether to operate in a market ($a_{it}^{j}=1$) or not ($a_{it}^{j}=0$).
The profit function has the following form for an active firm $j$:
\[
\bar{u}^{j}(x_{it},a_{it}^{j}=1,a_{it}^{-j};\theta)=\theta_{\text{FC},j}+\theta_{\text{RS}}s_{it}-\theta_{\text{RN}}\ln\left(1+\sum_{l\neq j}a_{it}^{l}\right)-\theta_{\text{EC}}(1-a_{i,t-1}^{j}).
\]
Here, $\theta_{FC,j}$ is the fixed cost of operation for firm j,
$\theta_{EC}$ is the cost incurred by a new entrant, $\theta_{RS}$
represents the effect of market size $s_{it}$ (the discretized logarithm
of population, defined above), and $\theta_{RN}$ captures the effect
of competition. When firm $j$ is inactive, $\bar{u}^{j}(x_{it},a_{it}^{j}=0,a_{it}^{-j};\theta)=0$.

\subsection{Structural Parameter Estimates}

Using $k$-EPL, we estimate the heterogeneous fixed costs $\theta_{\text{FC},\text{SC}}$,
$\theta_{\text{FC},\text{CC}}$, and $\theta_{\text{FC},\text{BJ}}$
as well as the entry cost $\theta_{\text{EC}}$, the coefficient on
market size $\theta_{\text{RS}}$, and the competitive effect $\theta_{\text{EC}}$
parameters of the model.

Table~\ref{tab:app:epl} reports the point estimates from the observed
sample along with standard errors and 95\% confidence intervals estimated
using 250 cross-sectional bootstrap replications. The estimates all
have the expected sign and are all significantly different from zero
at the 5\% level. Fixed costs are lowest for Sam's Club and highest
for BJ's. Firms are more profitable in larger markets and entry by
competitors reduces profits. The entry cost is large relative to fixed
costs, as expected.

\begin{table}[h]
\centering \caption{Wholesale Clubs: Parameter Estimates}
\label{tab:app:epl} %
\begin{tabular}{lrrr}
\hline 
\hline Parameter & Estimate & S.E. & 95\% CI\tabularnewline
\hline 
$\theta_{\text{FC},\text{SC}}$ & -0.136 & (0.030) & {[}-0.196, -0.079{]}\tabularnewline
$\theta_{\text{FC},\text{CC}}$ & -0.130 & (0.031) & {[}-0.188, -0.073{]}\tabularnewline
$\theta_{\text{FC},\text{BJ}}$ & -0.197 & (0.030) & {[}-0.255, -0.143{]}\tabularnewline
$\theta_{\text{RS}}$ & 0.106 & (0.009) & {[}0.090, 0.124{]}\tabularnewline
$\theta_{\text{RN}}$ & 0.137 & (0.030) & {[}0.087, 0.210{]}\tabularnewline
$\theta_{\text{EC}}$ & 8.855 & (0.163) & {[}8.559, 9.186{]}\tabularnewline
\hline 
\end{tabular}
\end{table}

We note that for this application $k$-NPL yields very similar results
to $k$-EPL. The estimated competitive effect is small, implying that
$k$-NPL is likely to be stable. However, we could not know this a
priori (recall that in our Monte Carlo Experiments, the $k$-NPL estimates
of the competitive effect are biased towards zero). However, ex post,
with the stable $k$-EPL estimates in hand, we can understand the
performance of $k$-NPL this setting.

\subsection{Counterfactual}

The industry under investigation is characterized by a significant
number of monopoly markets. In the latest year of our sample, 2021,
we observed a mere 14 triopoly markets (less than 1\% of our sample).
Only 7\% of the markets had a duopoly, with two firms present, while
20\% of the markets were monopolies, where a single firm operated.
Interestingly, 72\% of counties in our sample had no wholesale club
stores at all in 2021. Our counterfactual exercise aims to explore
the reasons behind this relative scarcity of duopoly and triopoly
markets. Are strong competitive effects or high costs responsible?

To address this question, we conduct a counterfactual simulation in
which we entirely eliminate the competitive effect in the model, allowing
firms to operate as independent agents without considering their competitors'
actions. If we observe a substantial increase in new entries, we may
deduce that strong competition is deterring other firms from entering
the market. Conversely, if we notice minimal change in entry behavior,
it may suggest that competitive effects are insignificant, and costs
are the primary factor driving firms' entry decisions.

To investigate this, we take the estimated structural parameters from
the observed sample, set $\theta_{\text{RN}}=0$, and compute the
counterfactual equilibrium. We solve the nonlinear system of equilibrium
equations using the estimated equilibrium as the starting value. Subsequently,
we compare the results of simulations conducted under this counterfactual
scenario to those obtained using the estimated model parameters. We
perform this procedure for each bootstrap replication to calculate
standard errors for the counterfactual quantities of interest. For
every simulation, we employ the observed market configuration at the
beginning of the sample in 2009 and simulate until the end of our
sample period in 2021.

In Table~\ref{tab:counterfactual_1}, we present aggregate statistics
from simulations using both the estimated and counterfactual parameters,
where the competitive effect has been set to zero. The reported figures
represent the estimated means and standard errors derived from simulated
sample paths using the parameters from the 250 bootstrap replications.
Before delving into the counterfactual analysis, it is worth noting
that this can also serve as a test of the model's fit. We observe
that the average number of active firms, entries, and exits, as well
as the distribution of the number of firms present, align reasonably
well with the observed values from the data. Upon examining the counterfactual,
we discover that removing the effect of competition does not significantly
reduce the number of unserved markets (from 1164 to 1156). However,
it does shift the distribution away from monopoly markets (from 341
down to 300) towards a higher prevalence of duopoly and triopoly markets
(from 93 up to 118, and from 11 up to 34, respectively). These findings
indicate that competition in this industry is relatively weak, and
costs play a more significant role in determining entry behavior.

\begin{table}[h]
\centering \caption{Wholesale Clubs: Counterfactual (Aggregate)}
\label{tab:counterfactual_1} %
\begin{tabular}{lrrrrr}
\toprule 
\hline  &  & \multicolumn{2}{c}{Simulated} & \multicolumn{2}{c}{Counterfactual}\tabularnewline
 & Observed & Mean & S.E. & Mean & S.E.\tabularnewline
\midrule 
Active Firms & 0.348 & 0.349 & (0.017) & 0.398 & (0.023)\tabularnewline
Entries & 0.010 & 0.010 & (0.001) & 0.016 & (0.002)\tabularnewline
Exits & 0.006 & 0.006 & (0.001) & 0.005 & (0.001)\tabularnewline
\midrule 
Markets with &  &  &  &  & \tabularnewline
\quad{}0 Firms & 1156 & 1164.665 & (27.805) & 1156.434 & (28.576)\tabularnewline
\quad{}1 Firms & 321 & 340.598 & (24.424) & 300.371 & (22.650)\tabularnewline
\quad{}2 Firms & 119 & 93.426 & (12.663) & 118.948 & (15.066)\tabularnewline
\quad{}3 Firms & 14 & 11.311 & (4.940) & 34.247 & (9.687)\tabularnewline
\bottomrule
\end{tabular}
\end{table}

Table~\ref{tab:counterfactual_2} provides a more detailed breakdown
of our simulations by firm and market size state $s$. We estimate
that, on average, Sam's Club would enter approximately 3 additional
small markets ($s=1,2,3$) and around 15 additional large markets
($s=4,5$). Costco would enter, on average, 4 additional small markets
and 28 large markets. BJ's would enter, on average, 3 additional small
markets and 27 large markets. Consequently, eliminating the competitive
effect appears to have a proportionally much larger impact on BJ's
compared to Costco and Sam's Club. This effect is also considerably
more pronounced in larger markets than in smaller ones.

\begin{table}[p]
\centering \caption{Wholesale Clubs: Counterfactual (By Market Size)}
\label{tab:counterfactual_2} %
\begin{tabular}{@{}lrrrrrr}
\hline 
\hline  & \multicolumn{2}{c}{Sam's Club} & \multicolumn{2}{c}{Costco} & \multicolumn{2}{c}{BJ's}\tabularnewline
\cmidrule(lr){2-3} \cmidrule(lr){4-5} \cmidrule(lr){6-7} & Mean & S.E. & Mean & S.E. & Mean & S.E.\tabularnewline
\hline 
\multicolumn{7}{l}{Estimated}\tabularnewline
\quad{}$s=1$ & 3.103 & (1.359) & 3.876 & (1.522) & 1.474 & (0.854)\tabularnewline
\quad{}$s=2$ & 19.384 & (4.153) & 12.480 & (3.134) & 5.739 & (2.011)\tabularnewline
\quad{}$s=3$ & 86.238 & (8.263) & 24.618 & (4.534) & 18.299 & (4.044)\tabularnewline
\quad{}$s=4$ & 127.656 & (11.218) & 54.597 & (7.234) & 29.962 & (5.098)\tabularnewline
\quad{}$s=5$ & 81.056 & (8.290) & 59.769 & (7.252) & 33.458 & (5.734)\tabularnewline
\hline 
\multicolumn{7}{l}{Counterfactual}\tabularnewline
\quad{}$s=1$ & 3.169 & (1.441) & 3.943 & (1.799) & 1.494 & (0.819)\tabularnewline
\quad{}$s=2$ & 19.892 & (4.232) & 12.889 & (3.268) & 5.784 & (2.208)\tabularnewline
\quad{}$s=3$ & 88.295 & (8.959) & 28.575 & (4.945) & 21.138 & (4.094)\tabularnewline
\quad{}$s=4$ & 135.488 & (11.393) & 70.326 & (8.880) & 41.387 & (7.177)\tabularnewline
\quad{}$s=5$ & 87.731 & (8.902) & 72.334 & (8.468) & 49.016 & (8.026)\tabularnewline
\hline 
\end{tabular}
\end{table}

Our analysis also incorporates average profit simulations using both
the estimated parameters and the counterfactual parameters. To achieve
this, for each simulation, we calculate $v^{j}(x,a^{j})+e(x,a^{j};P^{j})$
for each firm and average it over each simulated sample path of states
and choices.\footnote{We simulate actions by sampling from the distribution defined by $P^{j}$,
so $e(x,a^{j};P^{j})\equiv E[\varepsilon^{j}(a^{j})\mid x,a^{j},P^{j}]$
must be included to capture the contribution of $\varepsilon^{j}$
to average profit. With Type 1 Extreme Value shocks, $e(x,a^{j};P^{j})=-\ln P^{j}(x,a^{j})+\bar{\gamma}$,
where $\bar{\gamma}$ is the Euler-Mascheroni constant.} We perform this separately for the simulations based on the model
estimates and those under the counterfactual parameters. Comparing
these two unit-less measures of profitability offers another perspective
on our research question.

The results of this exercise are reported in Table~\ref{tab:counterfactual_3}.
The resulting changes in average profits are relatively minor: a 2.0\%
increase for Sam's Club, a 2.5\% increase for Costco, and a 1.9\%
increase for BJ's. These results provide further insights regarding
the relative significance of competition and costs in determining
market entry behavior. The modest changes in average profits under
the counterfactual scenario, where competitive effects are removed,
suggest that competition does not play a dominant role in shaping
firms' entry decisions. Instead, this reinforces our earlier findings
that costs are a more substantial determinant of entry behavior in
this industry.

\begin{table}[h]
\centering \caption{Wholesale Clubs: Counterfactual (Profits)}
\label{tab:counterfactual_3} %
\begin{tabular}{lrrrr}
\hline 
\hline  & \multicolumn{2}{c}{Simulated} & \multicolumn{2}{c}{Counterfactual}\tabularnewline
 & Mean & S.E. & Mean & S.E.\tabularnewline
\hline 
Sam's Club & 20.342 & (0.114) & 20.753 & (0.155)\tabularnewline
Costco & 19.566 & (0.080) & 20.051 & (0.146)\tabularnewline
BJ's & 19.020 & (0.051) & 19.384 & (0.131)\tabularnewline
\hline 
\end{tabular}
\end{table}

\section{\label{sec:Conclusion}Conclusion}

We proposed an iterative $k$-step Efficient Pseudo-Likelihood ($k$-EPL)
estimation sequence that extends the attractive econometric and computational
properties of the single-agent $k$-NPL sequence to games. The nice
econometric properties arise because $k$-EPL uses Newton-like steps
on the fixed point constraint at each iteration. As a result, $k$-EPL
is stable for all regular Markov perfect equilibria, each EPL iteration
has the same limiting distribution as the MLE, and further iterations
achieve higher-order equivalence and quickly converge to the finite-sample
MLE almost surely. Computational advantages follow from defining the
equilibrium conditions with choice-specific value functions, with
standard modeling assumptions reducing each EPL iteration to two steps:
(i) solving linear systems to generate pseudo-regressors, followed
by (ii) solving a globally concave static logit/probit maximum likelihood
problem using the pseudo-regressors.

In a real-world application, we use $k$-EPL to investigate the effect
of competition on entry and exit of U.S. wholesale club stores. Our
estimated model indicates that competition among wholesale club stores
has a relatively mild effect on their entry and exit.Our Monte Carlo
simulations show that $k$-EPL performs favorably in finite samples,
is robust to data-generating processes where standard $k$-NPL encounters
serious problems, and scales better than other iterative alternatives
to $k$-NPL.

One limitation of our analysis is that we did not consider time-invariant
unobserved heterogeneity in estimating dynamic discrete games. $k$-EPL
can easily accommodate a proxy variable approach (e.g., \citet{CollardWexler2013}),
where an observed time-invariant variable is used to proxy for the
the time-invariant unobserved heterogeneity. Without a proxy variable,
it may also be possible to modify the $k$-EPL algorithm to incorporate
time-invariant unobserved heterogeneity while preserving computational
convenience and econometric efficiency. However, we leave such a substantial
and challenging extension as an avenue for future research.

\newpage

\begin{center}
\textbf{Acknowledgements}\\
\end{center}

We thank J\'er\^ome Adda (editor) and three anonymous referees for
numerous comments which greatly improved the paper. We also thank
Dan Ackerberg, Victor Aguirregabiria, Lanier Benkard, Chris Conlon,
Arvind Magesan, Mathieu Marcoux, Robert Miller, Salvador Navarro,
John Rust, and Eduardo Souza-Rodrigues for comments and insightful
discussions. The paper has benefited from comments by seminar participants
at Boston College, Cornell University (Johnson), Georgetown University,
University of Notre Dame, University of Pennsylvania (Wharton), University
of Toronto, the 2019 International Industrial Organization Conference
(IIOC, Boston), the 2019 University of Calgary Empirical Microeconomics
Workshop (Banff), the 2020 ASSA Annual Meeting (San Diego), the Spring
2020 $\textrm{IO}^{2}$ Seminar, the 2020 Econometric Society World
Congress (Milan), the 2021 North American Meeting of the Econometric
Society, and the 2023 International Association for Applied Econometrics
conference (Oslo).

\bigskip{}

\begin{center}
\textbf{Data Availability Statement}\\
\end{center}

The data and code underlying this research is available on Zenodo
at \url{https://dx.doi.org/10.5281/zenodo.10582004}.

\bigskip{}

\bibliographystyle{chicago}
\bibliography{Biblio/bib1}

\newpage{}

\appendix

\section{Proofs}

Proofs are presented in order of appearance in the main text.

\subsection{Proof of Lemma \ref{lem:Rep-Lemma}}

We show that $v=\Phi(\theta,v)$ is a Bellman-like representation
of the best-response equilibrium conditions, $P=\sigma(\theta,P)$.
First, note that $P^{j}=\Lambda^{j}(v^{j})$ for all $j$. Using the
definition of $\Phi(\cdot)$ from Section \ref{sec:Dynamic-Discrete-Game}
and stacking all states for player $j$, we have 
\begin{align*}
\Phi^{j}(a^{j};v^{j},v^{-j},\theta) & =u^{j}(a^{j};\Lambda^{-j}(v^{-j}),\theta_{u})+\beta F^{j}(a^{j};\Lambda^{-j}(v^{-j});\theta_{f})S(v^{j}).
\end{align*}
 We see that $v^{-j}$ only influences $\Phi^{j}(\cdot)$ through
its effect on $P^{-j}=\Lambda^{-j}(v^{-j})$ and we can then define
$\phi^{j}(a^{j};\theta,v^{j},P^{-j})=\Phi^{j}(a^{j};\theta,v^{j},v^{-j})$.
It is straightforward to show that $||\nabla_{v^{j}}\phi^{j}||_{\infty}=\beta<1$.
To see this, note that 
\begin{align*}
\frac{\partial\phi^{j}(x,a^{j};\theta,v^{j},P^{-j})}{\partial v^{j}(x',a')} & =\beta f^{j}(x'\mid x,a^{j},P^{-j})P^{j}(x',a').\\
 & =\beta\times Pr(x',{a'}^{j}\mid x,a^{j},P^{-j}),
\end{align*}
where ${a'}^{j}$ denotes the future action of player $j$. So, we
have $\nabla_{v^{j}}\phi^{j}(v^{j})=\beta E^{j}(P^{-j})$, where $E^{j}(P^{-j})$
is a row-stochastic matrix with entries that represent transition
probabilities between state-action pairs: $E^{j}(P^{-j})\{k,l\}=Pr((x',{a'}^{j})=l\mid(x,a^{j})=k,P^{-j})$.
Since $E^{j}$ is row-stochastic, we have $||E^{j}(P^{-j})||_{\infty}=1$,
and it follows that $||\nabla_{v^{j}}\phi^{j}(v^{j})||_{\infty}=||\beta E^{j}(P^{-j})||_{\infty}=\beta||E^{j}(P^{-j})||_{\infty}=\beta$.
Consequently, $\phi^{j}(\cdot)$ is a Bellman-like contraction in
$v^{j}$ (fixing $\theta$ and $P^{-j}$) with a unique fixed point.
Player $j$'s best response is $\sigma^{j}(\theta,P^{-j})=\Lambda^{j}(v^{j})$
where $v^{j}=\phi^{j}(\theta,v^{j},P^{-j})$. Imposing $v=\Phi(\theta,v)$
is therefore equivalent to imposing $P=\sigma(\theta,P)$; their sets
of fixed points for a given $\theta$ are isomorphic.

\subsection{Proof of Lemma \ref{lem:Newt-Props}}

Let $\gamma=\left(\breve{\theta},Y\right)$ (to explicitly distinguish
$\theta$ from $\breve{\theta}$), so that 
\begin{align*}
\Upsilon(\theta,\gamma) & =\Upsilon\left(\theta,\left(\breve{\theta},Y\right)\right)\\
 & =Y-\left(\nabla_{Y}G\left(\breve{\theta},Y\right)\right)^{-1}G(\theta,Y).
\end{align*}
 Additionally, let $\gamma_{\theta}=(\theta,Y_{\theta})$. Then we
have 
\[
\Upsilon(\theta,\gamma_{\theta})=Y_{\theta}-\left(\nabla_{Y}G(\theta,Y_{\theta})\right)^{-1}G(\theta,Y_{\theta}).
\]
 Result 1 follows immediately because $\left(\nabla_{Y}G(\theta,Y_{\theta})\right)^{-1}$
is non-singular and $Y_{\theta}=Y(\theta)$ by definition, so $\Upsilon(\theta,\gamma_{\theta})=Y_{\theta}$
if and only if $G(\theta,Y_{\theta})=0$.

For Results 2 and 3, first consider partial derivatives of $\Upsilon(\theta,\gamma)$,
evaluated at $(\theta,\gamma_{\theta})$: 
\begin{align*}
\nabla_{\theta}\Upsilon(\theta,\gamma)\bigg|_{(\theta,\gamma)=(\theta,(\theta,Y_{\theta}))} & =-\left(\nabla_{Y}G\left(\breve{\theta},Y\right)\right)^{-1}\nabla_{\theta}G(\theta,Y)\quad\bigg|_{(\theta,\gamma)=(\theta,(\theta,Y_{\theta}))}\\
 & =-\left(\nabla_{Y}G\left(\theta,Y_{\theta}\right)\right)^{-1}\nabla_{\theta}G(\theta,Y_{\theta}).
\end{align*}
 Implicit differentiation on $G(\theta,Y_{\theta})=0$ yields
\begin{align*}
Y'(\theta) & =-(\nabla_{Y}G(\theta,Y_{\theta}))^{-1}\nabla_{\theta}G(\theta,Y_{\theta}),
\end{align*}
 proving Result 2. For Result 3, first note that we have 
\begin{equation}
\nabla_{\gamma}\Upsilon(\theta,\gamma)=\left[\begin{array}{c}
-\frac{\partial\nabla_{\gamma}G(\gamma)^{-1}}{\partial\theta}G(\theta,Y)\\
I-\nabla_{Y}G(\breve{\theta},Y)^{-1}\nabla_{Y}G(\theta,Y)-\frac{\partial\nabla_{\gamma}G(\gamma)^{-1}}{\partial\gamma}G(\theta,Y)
\end{array}\right]^{'}.\label{eq:Upsilon-jacobian}
\end{equation}
 This then leads to 
\begin{align*}
\nabla_{\gamma}\Upsilon\left(\theta,\gamma\right)\bigg|_{(\theta,\gamma)=(\theta,(\theta,Y_{\theta}))} & =\left[\begin{array}{c}
-\frac{\partial\nabla_{\gamma}G(\gamma)^{-1}}{\partial\theta}G(\theta,Y)\\
I-\nabla_{Y}G(\breve{\theta},Y)^{-1}\nabla_{Y}G(\theta,Y)-\frac{\partial\nabla_{\gamma}G(\gamma)^{-1}}{\partial\gamma}G(\theta,Y)
\end{array}\right]^{'}\bigg|_{(\theta,\gamma)=(\theta,(\theta,Y_{\theta}))}\\
 & =\left[\begin{array}{c}
-\frac{\partial\nabla_{\gamma}G(\gamma_{\theta})^{-1}}{\partial\theta}G(\theta,Y_{\theta})\\
I-\nabla_{Y}G(\theta,Y_{\theta})^{-1}\nabla_{Y}G(\theta,Y_{\theta})-\frac{\partial\nabla_{\gamma}G(\gamma_{\theta})^{-1}}{\partial\gamma}G(\theta,Y_{\theta})
\end{array}\right]^{'}\\
 & =\left[\begin{array}{c}
-\frac{\partial\nabla_{\gamma}G(\gamma_{\theta})^{-1}}{\partial\theta}\times0\\
I-I-\frac{\partial\nabla_{\gamma}G(\gamma_{\theta})^{-1}}{\partial\gamma}\times0
\end{array}\right]^{'}\\
 & =0.
\end{align*}
 where the second-to-last equality arises because $G(\theta,Y_{\theta})=0$.

\subsection{Proof of Theorem \ref{thm:kEPL-Properties}}

The proofs of Results 1 and 2 adapt the proofs of consistency and
asymptotic normality for the $1$-NPL estimator from \citet{AgMir2007}
to an inductive proof for $k$-EPL.\footnote{See the proofs of Propositions 1 and 2 in the Appendix of \citet{AgMir2007}.}
We do this by showing that strong $\sqrt{N}$-consistency of $\hat{\gamma}_{k-1}=(\hat{\theta}_{k-1},\hat{Y}_{k-1})$
implies the results for $\hat{\gamma}_{k}=(\hat{\theta}_{k},\hat{Y}_{k})$.
The proof of Result 3 follows the arguments very similar to those
used in the proofs of Proposition 2 of \citet{KasShim2008} and Proposition
7 in the supplementary material for \citet{KasShim2012}. Throughout,
we rely heavily on analysis similar to that from the proof of Lemma
\ref{lem:Newt-Props}.

It is helpful up-front to define $\tilde{q}_{i}(\theta,\gamma)=q_{i}(\theta,\Upsilon(\theta,\gamma))$,
$\tilde{Q}_{N}(\theta,\gamma)=N^{-1}\sum_{i=1}^{N}\tilde{q}_{i}(\theta,\gamma)$,
and $\tilde{\theta}_{N}(\gamma)=\textrm{arg max}_{\theta}\medspace\tilde{Q}_{N}(\theta,\gamma)$.
Similarly, $\tilde{Q}^{*}(\theta,\gamma)=E[\tilde{q}_{i}(\theta,\gamma)]$
and $\tilde{\theta}^{*}(\gamma)=\textrm{arg max}_{\theta}\medspace\tilde{Q}^{*}(\theta,\gamma)$.
Then, $\hat{\theta}_{k}=\tilde{\theta}_{N}(\hat{\gamma}_{k-1})$ and
$\hat{Y}_{k}=\Upsilon(\hat{\theta}_{k},\hat{\gamma}_{k-1})$.

\subsubsection{Result 1 (Strong consistency of $\hat{\theta}_{k}$ and $\hat{Y}_{k}$)}

We have uniform continuity of $\tilde{Q}^{*}(\theta,\gamma)$ and
that $\tilde{Q}_{N}(\theta,\gamma)$ converges almost surely and uniformly
in $(\theta,\gamma)\in\Theta\times(\Theta\times\mathcal{Y})$ to $\tilde{Q}^{*}(\theta,\gamma)$.
Also, $\hat{\gamma}_{k-1}$ converges almost surely to $\gamma^{*}$.
Appealing to Lemma 24.1 of \citet{GourierouxMonfort1995v2}, these
imply that $\tilde{Q}_{N}(\theta,\hat{\gamma}_{k-1})$ converges almost
surely and uniformly in $\theta\in\Theta$ to $\tilde{Q}^{*}(\theta,\gamma^{*})$.
Then since $\theta^{*}$ uniquely maximizes $\tilde{Q}^{*}(\theta,\gamma^{*})$
on $\Theta$, $\hat{\theta}_{k}$ converges almost surely to $\theta^{*}$
\citep[Property 24.2]{GourierouxMonfort1995v2}. Continuity of $\Upsilon(\theta,\gamma)$
and the Mann-Wald theorem then give almost sure convergence of $\hat{Y}_{k}$
to $Y^{*}$.

\subsubsection{Result 2 (Asymptotic Distribution of $\hat{\theta}_{k}$ and $\hat{Y}_{k}$)}

We will show that consistency of $\hat{\gamma}_{k-1}$ leads to asymptotic
normality of $\hat{\theta}_{k}$ and $\hat{Y}_{k}$, with their asymptotic
variance the same as the MLE. Using the properties of $\Upsilon(\cdot)$
defined in Algorithm \ref{alg:EPL-Algorithm}, the chain rule, and
the generalized information matrix equality \citep[p. 2163]{McFadNewey1994}
we obtain the following population equalities:
\[
\begin{aligned}\nabla_{\theta\theta'}\tilde{Q}^{*}(\theta^{*},\gamma^{*})= & -\Omega_{\theta\theta}^{*},\\
\nabla_{\theta\gamma'}\tilde{Q}^{*}(\theta^{*},\gamma^{*})= & 0.
\end{aligned}
\]

To establish these, first recall that $\tilde{Q}^{*}(\theta,\gamma)=\mathrm{{E}}[\ln f(w\mid\theta,\Upsilon(\theta,\gamma))]$
and let $\gamma=(\breve{\theta},\,Y)$ denote the components of $\gamma$
(to explicitly distinguish $\theta$ from $\breve{\theta}$). Then
by the generalized information matrix equality we have
\begin{align*}
\nabla_{\theta\theta'}\tilde{Q}^{*}(\theta^{*},\gamma^{*}) & =E[\nabla_{\theta\theta'}\ln f(w\mid\theta^{*},\Upsilon(\theta^{*},\gamma^{*}))]\\
 & =-E[\nabla_{\theta}\ln f(w\mid\theta^{*},\Upsilon(\theta^{*},\gamma^{*}))\nabla_{\theta}\ln f(w\mid\theta^{*},\Upsilon(\theta^{*},\gamma^{*}))']\\
 & =-\Omega_{\theta\theta}^{*},
\end{align*}
where the final equality is a direct implication of Lemma \ref{lem:Newt-Props}
(Result 2) because $\gamma^{*}=\gamma_{\theta^{*}}$. We also have
\begin{align*}
\nabla_{\theta\gamma'}\tilde{Q}^{*}(\theta^{*},\gamma^{*}) & =E[\nabla_{\theta\gamma'}\ln f(w\mid\theta^{*},\Upsilon(\theta^{*},\gamma^{*}))]\\
 & =-E[\nabla_{\theta}\ln f(w\mid\theta^{*},\Upsilon(\theta^{*},\gamma^{*}))\nabla_{\gamma}\ln f(w\mid\theta^{*},\Upsilon(\theta^{*},\gamma^{*}))']\\
 & =0.
\end{align*}
The last equality follows from the chain rule and Lemma \ref{lem:Newt-Props}:
\begin{align*}
\nabla_{\gamma}\ln f(w\mid\theta^{*},\Upsilon(\theta^{*},\gamma^{*})) & =\frac{{1}}{f(w\mid\theta^{*},\Upsilon(\theta^{*},\gamma^{*}))}\nabla_{Y}f(w\mid\theta^{*},\Upsilon(\theta^{*},\gamma^{*}))\nabla_{\gamma}\Upsilon(\theta^{*},\gamma^{*})\\
 & =0
\end{align*}
 where the second equality arises because $\nabla_{\gamma}\Upsilon(\theta^{*},\gamma^{*})=0$
by Lemma \ref{lem:Newt-Props} (Result 3).

Turning to the sample objective function, a Taylor expansion of the
first-order condition gives
\begin{align*}
0 & =\nabla_{\theta}\tilde{Q}_{N}(\theta^{*},\gamma^{*})+\nabla_{\theta\theta'}\tilde{Q}_{N}(\bar{\theta},\bar{\gamma})(\hat{\theta}_{k}-\theta^{*})+\nabla_{\theta\gamma'}\tilde{Q}_{N}(\bar{\theta},\bar{\gamma})(\hat{\gamma}_{k-1}-\gamma^{*}).
\end{align*}
Solving and scaling then yields
\begin{align*}
\sqrt{N}(\hat{\theta}_{k}-\theta^{*}) & =-\nabla_{\theta\theta'}\tilde{Q}_{N}(\bar{\theta},\bar{\gamma})^{-1}\left[\sqrt{N}\nabla_{\theta}\tilde{Q}_{N}(\theta^{*},\gamma^{*})+\nabla_{\theta\gamma'}\tilde{Q}_{N}(\bar{\theta},\bar{\gamma})\sqrt{N}(\hat{\gamma}_{k-1}-\gamma^{*})\right].
\end{align*}
By consistency of $\hat{\gamma}_{k-1}$ and $\hat{\theta}_{k}$ and
the Mann-Wald theorem we have $\nabla_{\theta\theta'}\tilde{Q}_{N}(\bar{\theta},\bar{\gamma})\overset{a.s.}{\to}-\Omega_{\theta\theta}^{*}$
and by the central limit theorem, $\sqrt{N}\nabla_{\theta}\tilde{Q}_{N}(\theta^{*},\gamma^{*})\overset{d}{\to}\mathcal{\mathrm{N}}(0,\Omega_{\theta\theta}^{*})$.
For the last term in square brackets we have $\sqrt{N}(\hat{\gamma}_{k-1}-\gamma^{*})=O_{p}(1)$
and $\nabla_{\theta\gamma'}\tilde{Q}_{N}(\bar{\theta},\bar{\gamma})=o_{p}(1)$.
Therefore,
\[
\sqrt{N}(\hat{\theta}_{k}-\theta^{*})\overset{d}{\to}\mathcal{\mathrm{N}}(0,\Omega_{\theta\theta}^{*-1}).
\]
Furthermore, because $\hat{Y}_{k}=\Upsilon(\hat{\theta}_{k},\hat{\gamma}_{k-1})$,
with $\Upsilon$ twice continuously differentiable in a neighborhood
of $(\theta^{*},Y^{*})$, consistency and asymptotic normality of
$\hat{Y}_{k}$ follow immediately. Asymptotic equivalence of $\hat{Y}_{k}$
and $\hat{Y}_{\text{MLE}}$ follow from asymptotic equivalence of
$\hat{\theta}_{k}$ and $\hat{\theta}_{\text{MLE}}$ and the properties
of $\Upsilon$. Strong $\sqrt{N}$-consistency of $\hat{\gamma}_{0}$
completes the proof by induction.

\subsubsection{Result 3 (Large Sample Convergence)}

This result follows from Kasahara and Shimotsu (\citeyear{KasShim2012},
Proposition 1). The zero Jacobian property ensures that the required
spectral radius is equal to zero: as established in Result 2 above
$\nabla_{\theta\gamma}\tilde{Q}^{*}(\theta^{*},\gamma^{*})=0$, so
the spectral radius is also zero.

\subsection{Proof of Theorem \ref{thm:Finite-Sample-Properties}}

By examining the first-order conditions, we can see that $\hat{\theta}_{\text{MLE}}=\tilde{\theta}_{N}(\hat{\gamma}_{\text{MLE}})$,
so that $\hat{Y}_{\text{MLE}}=\Upsilon(\tilde{\theta}_{N}(\hat{\gamma}_{\text{MLE}}),\hat{\gamma}_{\text{MLE}})$.
This proves Result 1: the MLE is a fixed point of the EPL iterations.

Now let $H_{N}$ denote the EPL iteration mapping by stacking the
updating equations for $\theta$ and $Y$ so that $\hat{\gamma}_{k}=H_{N}(\hat{\gamma}_{k-1})$:

\[
H_{N}(\gamma)=\left[\begin{array}{c}
H_{1,N}(\gamma)\\
H_{2,N}(\gamma)
\end{array}\right]\equiv\left[\begin{array}{c}
\tilde{\theta}_{N}(\gamma)\\
\Upsilon(\tilde{\theta}_{N}(\gamma),\gamma)
\end{array}\right].
\]
We then consider the first-order conditions evaluated at $\hat{\gamma}_{\text{MLE}}$.
First we have
\begin{equation}
\nabla_{\gamma}H_{1,N}(\hat{\gamma}_{\text{MLE}})=\nabla_{\gamma}\tilde{\theta}_{N}(\hat{\gamma}_{\text{MLE}})\overset{a.s.}{\to}\nabla_{\gamma}\tilde{\theta}^{*}(\gamma^{*})=0\label{eq:H1deriv}
\end{equation}
because $\nabla_{\theta\gamma'}\tilde{Q}^{*}(\theta^{*},\gamma^{*})=0$,
as shown in the proof of Theorem \ref{thm:kEPL-Properties} (Result
2). Second, we have
\[
\nabla_{\gamma}H_{2,N}(\hat{\gamma}_{\text{MLE}})=\nabla_{\theta}\Upsilon(\hat{\theta}_{\text{MLE}},\hat{\gamma}_{\text{MLE}})\nabla_{\gamma}\tilde{\theta}_{N}(\hat{\gamma}_{\text{MLE}})+\nabla_{\gamma}\Upsilon(\hat{\theta}_{\text{MLE}},\hat{\gamma}_{\text{MLE}}).
\]
recalling that $\tilde{\theta}_{N}(\hat{\gamma}_{\text{MLE}})=\hat{\theta}_{\text{MLE}}$.
The first term converges to zero in probability as in (\ref{eq:H1deriv}).
The second term is zero due to Lemma \ref{lem:Newt-Props} (Result
3).

The above analysis implies that $\hat{\gamma}_{\text{MLE}}=H_{N}(\hat{\gamma}_{\text{MLE}})$
and $\nabla_{\gamma}H_{N}(\hat{\gamma}_{\text{MLE}})\overset{a.s.}{\to}0$,
which are key to Results 2 and 3. To obtain Result 2, note that the
$\sqrt{N}$-consistency of $\hat{\gamma}_{\text{MLE}}$ implies $\nabla_{\gamma}H_{N}(\hat{\gamma}_{\text{MLE}})=0+O_{p}(N^{-1/2})$,
due to continuity of $\nabla_{\gamma}H_{N}(\cdot)$. So, by a first-order
expansion around $\hat{\gamma}_{\text{MLE}}$, 
\begin{align*}
\hat{\gamma}_{k} & =H_{N}(\hat{\gamma}_{k-1})\\
 & =H_{N}(\hat{\gamma}_{\text{MLE}})+\nabla_{\gamma}H_{N}(\hat{\gamma}_{\text{MLE}})(\hat{\gamma}_{k-1}-\hat{\gamma}_{\text{MLE}})+O_{p}(||\hat{\gamma}_{k-1}-\hat{\gamma}_{\text{MLE}}||^{2})\\
 & =\hat{\gamma}_{\text{MLE}}+(0+O_{p}(N^{-1/2}))(\hat{\gamma}_{k-1}-\hat{\gamma}_{\text{MLE}})+O_{p}(||\hat{\gamma}_{k-1}-\hat{\gamma}_{\text{MLE}}||^{2}).
\end{align*}
 It follows that 
\[
\hat{\gamma}_{k}-\hat{\gamma}_{\text{MLE}}=O_{p}(N^{-1/2}||\hat{\gamma}_{k-1}-\hat{\gamma}_{\text{MLE}}||+||\hat{\gamma}_{k-1}-\hat{\gamma}_{\text{MLE}}||^{2}).
\]
 For Result 3, we appeal to continuity of $\nabla_{\gamma}H_{N}(\cdot)$
and $||\cdot||$. For any $\varepsilon>0$, if $||\nabla_{\gamma}H_{N}(\hat{\gamma}_{\text{MLE}})||<\varepsilon$,
then there exists some neighborhood around $\hat{\gamma}_{\text{MLE}}$,
$\mathcal{B}$, such that $H_{N}(\cdot)$ is a contraction mapping
on $\mathcal{B}$ with Lipschitz constant, $L<\varepsilon$, and fixed
point $\hat{\gamma}_{\text{MLE}}$ (as established in Result 1). We
have $\nabla_{\gamma}H_{N}(\hat{\gamma}_{\text{MLE}})\overset{a.s.}{\to}0$,
so that $||\nabla_{\gamma}H_{N}(\hat{\gamma}_{\text{MLE}})||<\varepsilon$
w.p.a. 1 as $N\to\infty$. Result 3 follows immediately.

\section{\label{sec:Additional-Monte-Carlo}Additional Monte Carlo Results}

\subsection{\citet{PesSD2008}: Dynamic Game with Multiple Equilibria}

Here, we conduct simulations where we estimate the model from \citet{PesSD2008},
a small-scale dynamic game with multiple possible equilibria. There
are two firms indexed by $j\in\{1,2\}$ who choose an action in each
market $i$, denoted $a_{i}^{j}\in\{0,1\}$, where 1 is entry and
0 is exit. The observed state variable $x_{i}=(x_{i}^{1},x_{i}^{2})\in\{(0,0),(0,1),(1,0),(1,1)\}$
represents the incumbency status of firms 1 and 2, respectively. Flow
utilities are period profits:
\[
\tilde{u}^{j}(x_{i},a_{i}^{j}=1)=\theta_{\mathrm{M}}+\theta_{\mathrm{C}}a_{i}^{-j}+\theta_{\mathrm{EC}}(1-x_{i}^{j})+\varepsilon_{1}^{j},
\]
\[
\tilde{u}^{j}(x_{i},a_{i}^{j}=0)=\theta_{\mathrm{SV}}x_{i}^{j}+\varepsilon_{0}^{j},
\]
where $\theta_{\mathrm{EC}}$ represents the entry cost, $\theta_{\mathrm{SV}}$
is the scrap value, $\theta_{\mathrm{M}}$ is the monopoly profit,
and $\theta_{\mathrm{C}}$ is the effect of competition on profit.
The discount factor is $\beta^{*}\in(0,1)$. The data are generated
using the parameter values $(\theta_{\mathrm{M}},\theta_{\mathrm{C}},\theta_{\mathrm{EC}},\theta_{\mathrm{SV}},\beta)=(1.2,-2.4,-0.2,0.1,0.9)$.
The private shocks have distribution $\varepsilon_{a}^{j}\sim\mathcal{\mathrm{N}}(0,0.5)$.
We note that this is a slightly different parameterization of the
model than the one used by \citet{PesSD2008}, but it is straightforward
to show that the resulting flow utility values are the same and hence
the equilibria are also the same.

There are multiple equilibria in the game, and we generate data from
equilibria (i), (ii), and (iii) from \citet{PesSD2008}. The NPL mapping
is unstable for two of the three equilibria, but the EPL mapping,
due to its Newton-like form, is stable for all three equilibria. Specifically,
equilibrium (i) is asymmetric and NPL-stable, equilibrium (ii) is
asymmetric and NPL-unstable, and equilibrium (iii) is symmetric and
NPL-unstable. In their Monte Carlo simulations, \citet{PesSD2008}
find that $k$-NPL performs well for equilibrium (i) but becomes severely
biased for equilibria (ii) and (iii) as $k$ grows, which is in line
with the theoretical analysis of \citet{KasShim2012}.\footnote{\citet{PesSD2008} use the terminology ``$k$-PML'' for $k$-NPL
and iterate until $k=20$.} In contrast, we expect that $k$-EPL will perform well for all three
equilibria.

We estimate $(\theta_{\mathrm{M}},\theta_{\mathrm{C}},\theta_{\mathrm{EC}})$
and assume the other parameter values are known and held fixed at
$\theta_{SV}=0.1$ and $\beta=0.9$. We present results for estimation
using $k$-NPL and $k$-EPL. The initial estimates of the conditional
choice probabilities are sample frequencies, $\hat{P}^{j}(x,a^{j})$.
We generate the data by first taking $N\in\{250,1000\}$ i.i.d. draws
from the stationary distribution of the observed state, $x$, for
each equilibrium. One interpretation of this sampling procedure is
that each of the draws from the stationary distribution of $x$ represents
an independent market. For each of these $N$ draws we then sample
actions for each player using the equilibrium choice probabilities.
We carry out 1000 replications for each sample size. For $\infty$-NPL
and $\infty$-EPL we terminate the algorithm when $|\hat{\theta}_{k}-\hat{\theta}_{k-1}|<10^{-6}$
or after 100 iterations. Computational times reported are minutes
of ``wall clock'' time required to carry out the full set of replications.\footnote{Experiments were carried out using MATLAB R2018a on a 2017 iMac Pro
in parallel using 18 Intel Xeon 2.3 GHz cores.}

In order to implement $k$-EPL in this context, we use the fixed-point
constraint in choice-specific value function space, defined in Section
\ref{sec:Dynamic-Discrete-Game}. Computing $\nabla_{v}G(\hat{\theta}_{k-1},\hat{v}_{k-1})$
for $k=1$ requires initial estimates $(\hat{\theta}_{0},\hat{v}_{0})$.
We use $\hat{\theta}_{0}=\hat{\theta}_{1\textrm{-NPL}}$, the estimate
from $1$-NPL, which is similar to the way \citet{PesSD2008} use
$\hat{\theta}_{1\textrm{-NPL}}$ to obtain an estimate of the efficient
weighting matrix used in their minimum-distance estimator. For each
player, we then set $\hat{v}_{0}^{j}(x,a^{j})=u^{j}(x,a^{j};\hat{\theta}_{0},\hat{P}^{-j})+\beta f^{j}(x,a^{j};\hat{P}^{-j})'\Gamma^{j}(\hat{\theta}_{0},\hat{P})$,
where 
\[
\Gamma^{j}(\theta,P)=\left(I-\beta F^{j}(\theta,P)\right)^{-1}\sum_{a^{j}}P^{j}(a^{j})*\left(u^{j}(a^{j};\theta,P^{-j})+e(a^{j};P^{j})\right)
\]
 maps $(\theta,P)$ into an \emph{ex-ante} value function for player
$j$, as in Aguirregabiria and Mira (\citeyear{AgMir2007}, p. 10).

\subsubsection{Results for NPL-Stable Equilibrium (i)}

Table \ref{tab:PSD-Eq1} shows results for equilibrium (i), for which
$k$-NPL is stable and consistent. We consider both the one-step $k$-NPL
and $k$-EPL estimators (for $k=1$) as well as the converged estimators
($k=\infty$). For $k=\infty$, we report the total estimation time
across all datasets, as well as the median and interquartile range
(IQR) of the number of iterations. For the large sample experiments,
we obtained convergence in fewer than 100 iterations for all algorithms
in almost all datasets, with $\infty$-NPL and $\infty$-EPL failing
to converge in only 5 and 1 out of 1000 datasets, respectively. Convergence
rates were somewhat lower, especially for $\infty$-NPL, with the
smaller sample size. Our reported results include all datasets, including
those where we obtain non-convergence.

\begin{table}[h]
\caption{\label{tab:PSD-Eq1}Monte Carlo Results for \citet{PesSD2008} NPL-Stable
Equilibrium (i)}

\centering{}%
\begin{tabular}{ccccr@{\extracolsep{0pt}.}lr@{\extracolsep{0pt}.}lr@{\extracolsep{0pt}.}lr@{\extracolsep{0pt}.}lr@{\extracolsep{0pt}.}lr@{\extracolsep{0pt}.}lr@{\extracolsep{0pt}.}l}
\hline 
\hline Obs. & Parameter & Statistic &  & \multicolumn{2}{c}{1-NPL} & \multicolumn{2}{c}{} & \multicolumn{2}{c}{1-EPL} & \multicolumn{2}{c}{} & \multicolumn{2}{c}{$\infty$-NPL} & \multicolumn{2}{c}{} & \multicolumn{2}{c}{$\infty$-EPL}\tabularnewline
\hline 
\multirow{10}{*}{\begin{turn}{90}
$N=250$
\end{turn}} & \multirow{2}{*}{$\theta_{\mathrm{M}}=1.2$} & Mean Bias &  & -0&0579 & \multicolumn{2}{c}{} & -0&0277 & \multicolumn{2}{c}{} & -0&0158 & \multicolumn{2}{c}{} & 0&0277\tabularnewline
 &  & MSE &  & 0&0461 & \multicolumn{2}{c}{} & 0&0376 & \multicolumn{2}{c}{} & 0&0368 & \multicolumn{2}{c}{} & 0&0312\tabularnewline
 & \multirow{2}{*}{$\theta_{\mathrm{C}}=-2.4$} & Mean Bias &  & 0&1120 & \multicolumn{2}{c}{} & 0&0425 & \multicolumn{2}{c}{} & 0&0294 & \multicolumn{2}{c}{} & -0&0482\tabularnewline
 &  & MSE &  & 0&1642 & \multicolumn{2}{c}{} & 0&1061 & \multicolumn{2}{c}{} & 0&0585 & \multicolumn{2}{c}{} & 0&0512\tabularnewline
 & \multirow{2}{*}{$\theta_{\mathrm{EC}}=-0.2$} & Mean Bias &  & -0&0393 & \multicolumn{2}{c}{} & -0&0205 & \multicolumn{2}{c}{} & -0&0270 & \multicolumn{2}{c}{} & -0&0039\tabularnewline
 &  & MSE &  & 0&0494 & \multicolumn{2}{c}{} & 0&0338 & \multicolumn{2}{c}{} & 0&0116 & \multicolumn{2}{c}{} & 0&0045\tabularnewline
 & Converged & \% &  & \multicolumn{2}{c}{} & \multicolumn{2}{c}{} & \multicolumn{2}{c}{} & \multicolumn{2}{c}{} & 92&6\% & \multicolumn{2}{c}{} & 97&5\%\tabularnewline
 & \multirow{2}{*}{Iterations} & Median &  & \multicolumn{2}{c}{} & \multicolumn{2}{c}{} & \multicolumn{2}{c}{} & \multicolumn{2}{c}{} & \multicolumn{2}{c}{70} & \multicolumn{2}{c}{} & \multicolumn{2}{c}{8}\tabularnewline
 &  & IQR &  & \multicolumn{2}{c}{} & \multicolumn{2}{c}{} & \multicolumn{2}{c}{} & \multicolumn{2}{c}{} & \multicolumn{2}{c}{28} & \multicolumn{2}{c}{} & \multicolumn{2}{c}{2}\tabularnewline
 & Time (min.) & Total &  & \multicolumn{2}{c}{} & \multicolumn{2}{c}{} & \multicolumn{2}{c}{} & \multicolumn{2}{c}{} & 0&4481 & \multicolumn{2}{c}{} & 0&1047\tabularnewline
\hline 
\multirow{10}{*}{\begin{turn}{90}
$N=1000$
\end{turn}} & \multirow{2}{*}{$\theta_{\mathrm{M}}=1.2$} & Mean Bias &  & -0&0165 & \multicolumn{2}{c}{} & -0&0050 & \multicolumn{2}{c}{} & -0&0044 & \multicolumn{2}{c}{} & 0&0033\tabularnewline
 &  & MSE &  & 0&0116 & \multicolumn{2}{c}{} & 0&0107 & \multicolumn{2}{c}{} & 0&0083 & \multicolumn{2}{c}{} & 0&0059\tabularnewline
 & \multirow{2}{*}{$\theta_{\mathrm{C}}=-2.4$} & Mean Bias &  & 0&0340 & \multicolumn{2}{c}{} & 0&0119 & \multicolumn{2}{c}{} & 0&0076 & \multicolumn{2}{c}{} & -0&0052\tabularnewline
 &  & MSE &  & 0&0423 & \multicolumn{2}{c}{} & 0&0320 & \multicolumn{2}{c}{} & 0&0106 & \multicolumn{2}{c}{} & 0&0076\tabularnewline
 & \multirow{2}{*}{$\theta_{\mathrm{EC}}=-0.2$} & Mean Bias &  & -0&0127 & \multicolumn{2}{c}{} & -0&0061 & \multicolumn{2}{c}{} & -0&0059 & \multicolumn{2}{c}{} & -0&0012\tabularnewline
 &  & MSE &  & 0&0123 & \multicolumn{2}{c}{} & 0&0086 & \multicolumn{2}{c}{} & 0&0018 & \multicolumn{2}{c}{} & 0&0008\tabularnewline
 & Converged & \% &  & \multicolumn{2}{c}{} & \multicolumn{2}{c}{} & \multicolumn{2}{c}{} & \multicolumn{2}{c}{} & 99&5\% & \multicolumn{2}{c}{} & 99&9\%\tabularnewline
 & \multirow{2}{*}{Iterations} & Median &  & \multicolumn{2}{c}{} & \multicolumn{2}{c}{} & \multicolumn{2}{c}{} & \multicolumn{2}{c}{} & \multicolumn{2}{c}{70} & \multicolumn{2}{c}{} & \multicolumn{2}{c}{6}\tabularnewline
 &  & IQR &  & \multicolumn{2}{c}{} & \multicolumn{2}{c}{} & \multicolumn{2}{c}{} & \multicolumn{2}{c}{} & \multicolumn{2}{c}{19} & \multicolumn{2}{c}{} & \multicolumn{2}{c}{1}\tabularnewline
 & Time (min.) & Total &  & \multicolumn{2}{c}{} & \multicolumn{2}{c}{} & \multicolumn{2}{c}{} & \multicolumn{2}{c}{} & 0&5847 & \multicolumn{2}{c}{} & 0&0743\tabularnewline
\hline 
\end{tabular}
\end{table}

Comparing 1-NPL to 1-EPL in Table \ref{tab:PSD-Eq1}, we see that
1-EPL has lower mean bias and MSE for all three parameters of interest.
However, both of these are outperformed by estimators iterated to
convergence, illustrating the gains from such iterations in finite
samples. For the larger sample size $\infty$-EPL has the lowest bias,
MSE, number of iterations, and computation time. Even for this equilibrium
where $k$-NPL is expected to perform well, the efficiency of $k$-EPL
yields improvements. Time per iteration is similar for $k$-NPL and
$k$-EPL here, so the lower computational times are driven by the
significant reduction in the number of iterations to convergence.\footnote{Each iteration reduces to solving a linear system and then estimating
a static binary probit model.} In this case for the smaller sample size, the results are more mixed.
The mean bias of $k$-EPL is higher (in absolute value) for two parameter
values, but the MSE is lower for all three. However, as before, convergence
is much faster.

\subsubsection{Results for NPL-Unstable Equilibria (ii) and (iii)}

\begin{table}[h]
\caption{\label{tab:PSD-Eq2}Monte Carlo Results for \citet{PesSD2008} NPL-Unstable
Equilibrium (ii)}

\centering{}%
\begin{tabular}{ccccr@{\extracolsep{0pt}.}lr@{\extracolsep{0pt}.}lr@{\extracolsep{0pt}.}lr@{\extracolsep{0pt}.}lr@{\extracolsep{0pt}.}lr@{\extracolsep{0pt}.}lr@{\extracolsep{0pt}.}l}
\hline 
\hline Obs. & Parameter & Statistic &  & \multicolumn{2}{c}{1-NPL} & \multicolumn{2}{c}{} & \multicolumn{2}{c}{1-EPL} & \multicolumn{2}{c}{} & \multicolumn{2}{c}{$\infty$-NPL} & \multicolumn{2}{c}{} & \multicolumn{2}{c}{$\infty$-EPL}\tabularnewline
\hline 
\multirow{10}{*}{\begin{turn}{90}
$N=250$
\end{turn}} & \multirow{2}{*}{$\theta_{\mathrm{M}}=1.2$} & Mean Bias &  & -0&1322 & \multicolumn{2}{c}{} & -0&1461 & \multicolumn{2}{c}{} & -0&2099 & \multicolumn{2}{c}{} & -0&0309\tabularnewline
 &  & MSE &  & 0&0902 & \multicolumn{2}{c}{} & 0&0988 & \multicolumn{2}{c}{} & 0&0622 & \multicolumn{2}{c}{} & 0&0740\tabularnewline
 & \multirow{2}{*}{$\theta_{\mathrm{C}}=-2.4$} & Mean Bias &  & 0&2793 & \multicolumn{2}{c}{} & 0&2617 & \multicolumn{2}{c}{} & 0&6719 & \multicolumn{2}{c}{} & 0&0717\tabularnewline
 &  & MSE &  & 0&4643 & \multicolumn{2}{c}{} & 0&5121 & \multicolumn{2}{c}{} & 0&4804 & \multicolumn{2}{c}{} & 0&4106\tabularnewline
 & \multirow{2}{*}{$\theta_{\mathrm{EC}}=-0.2$} & Mean Bias &  & -0&0777 & \multicolumn{2}{c}{} & -0&0764 & \multicolumn{2}{c}{} & -0&3110 & \multicolumn{2}{c}{} & -0&0441\tabularnewline
 &  & MSE &  & 0&1058 & \multicolumn{2}{c}{} & 0&1270 & \multicolumn{2}{c}{} & 0&1117 & \multicolumn{2}{c}{} & 0&1076\tabularnewline
 & Converged & \% &  & \multicolumn{2}{c}{} & \multicolumn{2}{c}{} & \multicolumn{2}{c}{} & \multicolumn{2}{c}{} & 96&1\% & \multicolumn{2}{c}{} & \multicolumn{2}{c}{100\%}\tabularnewline
 & \multirow{2}{*}{Iterations} & Median &  & \multicolumn{2}{c}{} & \multicolumn{2}{c}{} & \multicolumn{2}{c}{} & \multicolumn{2}{c}{} & \multicolumn{2}{c}{34} & \multicolumn{2}{c}{} & \multicolumn{2}{c}{9}\tabularnewline
 &  & IQR &  & \multicolumn{2}{c}{} & \multicolumn{2}{c}{} & \multicolumn{2}{c}{} & \multicolumn{2}{c}{} & \multicolumn{2}{c}{10} & \multicolumn{2}{c}{} & \multicolumn{2}{c}{3}\tabularnewline
 & Time (min.) & Total &  & \multicolumn{2}{c}{} & \multicolumn{2}{c}{} & \multicolumn{2}{c}{} & \multicolumn{2}{c}{} & 0&3011 & \multicolumn{2}{c}{} & 0&0785\tabularnewline
\hline 
\multirow{10}{*}{\begin{turn}{90}
$N=1000$
\end{turn}} & \multirow{2}{*}{$\theta_{\mathrm{M}}=1.2$} & Mean Bias &  & -0&0432 & \multicolumn{2}{c}{} & -0&0385 & \multicolumn{2}{c}{} & -0&2093 & \multicolumn{2}{c}{} & -0&0013\tabularnewline
 &  & MSE &  & 0&0210 & \multicolumn{2}{c}{} & 0&0205 & \multicolumn{2}{c}{} & 0&0480 & \multicolumn{2}{c}{} & 0&0155\tabularnewline
 & \multirow{2}{*}{$\theta_{\mathrm{C}}=-2.4$} & Mean Bias &  & 0&0952 & \multicolumn{2}{c}{} & 0&0612 & \multicolumn{2}{c}{} & 0&6636 & \multicolumn{2}{c}{} & 0&0047\tabularnewline
 &  & MSE &  & 0&1162 & \multicolumn{2}{c}{} & 0&1078 & \multicolumn{2}{c}{} & 0&4459 & \multicolumn{2}{c}{} & 0&0829\tabularnewline
 & \multirow{2}{*}{$\theta_{\mathrm{EC}}=-0.2$} & Mean Bias &  & -0&0269 & \multicolumn{2}{c}{} & -0&0122 & \multicolumn{2}{c}{} & -0&2983 & \multicolumn{2}{c}{} & -0&0039\tabularnewline
 &  & MSE &  & 0&0283 & \multicolumn{2}{c}{} & 0&0279 & \multicolumn{2}{c}{} & 0&0923 & \multicolumn{2}{c}{} & 0&0222\tabularnewline
 & Converged & \% &  & \multicolumn{2}{c}{} & \multicolumn{2}{c}{} & \multicolumn{2}{c}{} & \multicolumn{2}{c}{} & 99&7\% & \multicolumn{2}{c}{} & \multicolumn{2}{c}{100\%}\tabularnewline
 & \multirow{2}{*}{Iterations} & Median &  & \multicolumn{2}{c}{} & \multicolumn{2}{c}{} & \multicolumn{2}{c}{} & \multicolumn{2}{c}{} & \multicolumn{2}{c}{32} & \multicolumn{2}{c}{} & \multicolumn{2}{c}{7}\tabularnewline
 &  & IQR &  & \multicolumn{2}{c}{} & \multicolumn{2}{c}{} & \multicolumn{2}{c}{} & \multicolumn{2}{c}{} & \multicolumn{2}{c}{6} & \multicolumn{2}{c}{} & \multicolumn{2}{c}{1}\tabularnewline
 & Time (min.) & Total &  & \multicolumn{2}{c}{} & \multicolumn{2}{c}{} & \multicolumn{2}{c}{} & \multicolumn{2}{c}{} & 0&3585 & \multicolumn{2}{c}{} & 0&0946\tabularnewline
\hline 
\end{tabular}
\end{table}

Turning to equilibrium (ii), for which the NPL fixed point is unstable,
we have a very different picture. The results are presented in Table
\ref{tab:PSD-Eq2}. For $\infty$-NPL there is substantial bias in
all parameters and seemingly little variation around those biased
values. For example, there is attenuation bias in the competitive
effect $\theta_{\text{{\ensuremath{\mathrm{C}}}}}$, making it seem
less negative. This bias does not appear to decrease with the sample
size. The bias for $\infty$-EPL is lower by an order of magnitude
in all cases for $N=250$ and by two orders of magnitude for $N=1000$.
Relative to 1-NPL and 1-EPL, iterating in the finite sample improves
the estimates for $\infty$-EPL but worsens the estimates for $\infty$-NPL.

Table \ref{tab:PSD-Eq3} reports some results for equilibrium (iii),
which is also NPL-unstable, and therefore the results are qualitatively
similar to those for equilibrium (ii) presented in Table \ref{tab:PSD-Eq2}.

\begin{table}[h]
\caption{\label{tab:PSD-Eq3}Monte Carlo Results for \citet{PesSD2008} Equilibrium
(iii) with $N=1000$}

\centering{}%
\begin{tabular}{cccccccccc}
\hline 
\hline Parameter & Statistic &  & 1-NPL &  & 1-EPL &  & $\infty$-NPL &  & $\infty$-EPL\tabularnewline
\hline 
\multirow{2}{*}{$\theta_{\mathrm{M}}=1.2$} & Mean Bias &  & -0.0419 &  & -0.0383 &  & -0.2099 &  & -0.0003\tabularnewline
 & MSE &  & 0.0204 &  & 0.0194 &  & 0.0480 &  & 0.0174\tabularnewline
\multirow{2}{*}{$\theta_{\mathrm{C}}=-2.4$} & Mean Bias &  & 0.0948 &  & 0.0625 &  & 0.6806 &  & 0.0043\tabularnewline
 & MSE &  & 0.1127 &  & 0.1000 &  & 0.4683 &  & 0.0987\tabularnewline
\multirow{2}{*}{$\theta_{\mathrm{EC}}=-0.2$} & Mean Bias &  & -0.0277 &  & -0.0133 &  & -0.3146 &  & -0.0044\tabularnewline
 & MSE &  & 0.0286 &  & 0.0265 &  & 0.1023 &  & 0.0277\tabularnewline
Converged & \% &  &  &  &  &  & 99.8\% &  & 100\%\tabularnewline
\multirow{2}{*}{Iterations} & Median &  &  &  &  &  & 30 &  & 8\tabularnewline
 & IQR &  &  &  &  &  & 5 &  & 2\tabularnewline
Time (min.) & Total &  &  &  &  &  & 0.3486 &  & 0.1026\tabularnewline
\hline 
\end{tabular}
\end{table}

Overall, the results here illustrate the good performance of $k$-EPL
and in particular $\infty$-EPL. We see that $\infty$-EPL is generally
more efficient, is robust to unstable equilibria, and converges in
fewer iterations than $\infty$-NPL, resulting in substantial time
savings.

\subsubsection{Effects of Noisy, Inconsistent Starting Values}

Next, we consider robustness to starting values. Like convergence
results for Newton's method, our convergence results are local. That
is, the starting values (initial estimates) must be in a neighborhood
of the maximum likelihood estimates to guarantee convergence. We do
not claim, nor should we expect, global convergence results in models
with multiple equilibria. This underscores the importance of good
initial estimates, i.e., either consistent estimates or multiple starting
values, or both. First we explore using a single, polluted version
of the consistent estimates as the starting value. This is for the
purposes of exploring the effect of noise in consistent starting values
with respect to the locality of convergence. In practice, of course,
we recommend using multiple starting values.

\begin{table}[h]
\caption{\label{tab:PSD-Noise}Monte Carlo Results for \citet{PesSD2008} with
Noisy Starting Values ($N=250$)}

\centering{}%
\begin{tabular}{cccr@{\extracolsep{0pt}.}lr@{\extracolsep{0pt}.}lr@{\extracolsep{0pt}.}lr@{\extracolsep{0pt}.}lr@{\extracolsep{0pt}.}lr@{\extracolsep{0pt}.}lr@{\extracolsep{0pt}.}l}
\hline 
\hline  &  &  & \multicolumn{6}{c}{Equilibrium (i)} & \multicolumn{2}{c}{} & \multicolumn{6}{c}{Equilibrium (ii)}\tabularnewline
Parameter & Statistic &  & \multicolumn{2}{c}{$\infty$-NPL} & \multicolumn{2}{c}{} & \multicolumn{2}{c}{$\infty$-EPL} & \multicolumn{2}{c}{} & \multicolumn{2}{c}{$\infty$-NPL} & \multicolumn{2}{c}{} & \multicolumn{2}{c}{$\infty$-EPL}\tabularnewline
\hline 
\multirow{2}{*}{$\theta_{\mathrm{M}}=1.2$} & Mean Bias &  & -0&0827 & \multicolumn{2}{c}{} & 0&0266 & \multicolumn{2}{c}{} & -0&2116 & \multicolumn{2}{c}{} & -0&0855\tabularnewline
 & MSE &  & 0&0648 & \multicolumn{2}{c}{} & 0&0756 & \multicolumn{2}{c}{} & 0&0630 & \multicolumn{2}{c}{} & 0&1054\tabularnewline
\multirow{2}{*}{$\theta_{\mathrm{C}}=-2.4$} & Mean Bias &  & 0&1286 & \multicolumn{2}{c}{} & -0&0431 & \multicolumn{2}{c}{} & 0&6738 & \multicolumn{2}{c}{} & 0&2055\tabularnewline
 & MSE &  & 0&1282 & \multicolumn{2}{c}{} & 0&2234 & \multicolumn{2}{c}{} & 0&4832 & \multicolumn{2}{c}{} & 0&5810\tabularnewline
\multirow{2}{*}{$\theta_{\mathrm{EC}}=-0.2$} & Mean Bias &  & -0&0672 & \multicolumn{2}{c}{} & 0&0020 & \multicolumn{2}{c}{} & -0&3107 & \multicolumn{2}{c}{} & -0&0965\tabularnewline
 & MSE &  & 0&0251 & \multicolumn{2}{c}{} & 0&0073 & \multicolumn{2}{c}{} & 0&1118 & \multicolumn{2}{c}{} & 0&1275\tabularnewline
Converged & \% &  & 89&2\% & \multicolumn{2}{c}{} & 96&6\% & \multicolumn{2}{c}{} & 99&3\% & \multicolumn{2}{c}{} & 99&9\%\tabularnewline
\multirow{2}{*}{Iterations} & Median &  & \multicolumn{2}{c}{72} & \multicolumn{2}{c}{} & \multicolumn{2}{c}{9} & \multicolumn{2}{c}{} & \multicolumn{2}{c}{34} & \multicolumn{2}{c}{} & \multicolumn{2}{c}{10}\tabularnewline
 & IQR &  & \multicolumn{2}{c}{36} & \multicolumn{2}{c}{} & \multicolumn{2}{c}{3} & \multicolumn{2}{c}{} & \multicolumn{2}{c}{11} & \multicolumn{2}{c}{} & \multicolumn{2}{c}{3}\tabularnewline
Time (min.) & Total &  & 0&4024 & \multicolumn{2}{c}{} & 0&1217 & \multicolumn{2}{c}{} & 0&2533 & \multicolumn{2}{c}{} & 0&0889\tabularnewline
\hline 
\end{tabular}
\end{table}

In Table \ref{tab:PSD-Noise} we used initial choice probabilities
that were an equally weighted average of the consistent CCP estimates
and $\mathrm{{Uniform}}(0,1)$ noise. We then re-computed each of
the converged estimates --- $\infty$-NPL and $\infty$-EPL ---
using these noisy starting values with the small sample size $N=250$.
For equilibrium (i), comparing with the consistent starting values
from the top panel of Table \ref{tab:PSD-Eq1}, we see that the added
noise increases the MSE values and decreases convergence rates for
both estimators, but the increase in bias is smallest for $\infty$-EPL.
Furthermore, the convergence rate of $\infty$-EPL decreases less
than the convergence rate for $\infty$-NPL.

For equilibrium (ii) we can compare with the consistent starting values
from the top panel of Table \ref{tab:PSD-Eq2}. In this case the bias
and MSE for $\infty$-NPL only changed slightly because the estimates
were previously also biased. There is only a slight increase in bias
and MSE as a result of the noisy starting values, but the results
are largely the same as before. The convergence percentage for $\infty$-NPL
actually increased with the added noise, but the iterations still
converge to inconsistent estimates. On the other hand, the bias and
MSE values for $\infty$-EPL increased---especially for $\theta_{\text{{C}}}$---while
the convergence percentage only decreased from 100\% to 99.9\%.

Overall, while good starting values are important, these results show
that $k$-EPL is also somewhat robust starting values with fairly
severe estimation errors. Note that we do not actually recommend using
only a single starting value if first-stage CCP estimation accuracy
is a concern. With that in mind, in the next section we consider moving
away from consistent starting values entirely.

\subsubsection{$k$-EPL as a Computational Procedure (Random Starting Values)}

Rather than rely solely on a single consistent estimate, we consider
here using $k$-EPL as a computational procedure to compute the MLE
using multiple starting values (in practice, ideally with a consistent
estimate among them). A similar procedure was suggested by \citet{AgMir2007}
to compute the NPL estimator by attempting to use the $k$-NPL algorithm,
with multiple starting values, to compute all NPL fixed points, and
taking the estimate that maximizes the likelihood. However, for datasets
generated by equilibria for which the NPL mapping is unstable, the
initial guess may need to be exactly correct to reach those fixed
points.\footnote{\citet{AgMarcoux2019} show that this issue can even arise when the
data are generated from a stable equilibrium.} $k$-EPL, however, is stable and has a faster rate of convergence,
with the maximum likelihood estimator as a fixed point. So, in this
section we consider using this approach with $k$-EPL to compute the
MLE.

\begin{table}
\caption{\label{tab:PSD-MultiStart}Monte Carlo Results for \citet{PesSD2008}
Without Consistent Starting Values}

\centering{}%
\begin{tabular}{ccccr@{\extracolsep{0pt}.}lr@{\extracolsep{0pt}.}lr@{\extracolsep{0pt}.}lr@{\extracolsep{0pt}.}lr@{\extracolsep{0pt}.}lr@{\extracolsep{0pt}.}lr@{\extracolsep{0pt}.}l}
\hline 
\hline  &  &  &  & \multicolumn{6}{c}{Equilibrium (i)} & \multicolumn{2}{c}{} & \multicolumn{6}{c}{Equilibrium (ii)}\tabularnewline
Obs. & Parameter & Statistic &  & \multicolumn{2}{c}{$\infty$-NPL} & \multicolumn{2}{c}{} & \multicolumn{2}{c}{$\infty$-EPL} & \multicolumn{2}{c}{} & \multicolumn{2}{c}{$\infty$-NPL} & \multicolumn{2}{c}{} & \multicolumn{2}{c}{$\infty$-EPL}\tabularnewline
\hline 
\multirow{10}{*}{\begin{turn}{90}
$N=250$
\end{turn}} & \multirow{2}{*}{$\theta_{\mathrm{M}}=1.2$} & Mean Bias &  & -0&0127 & \multicolumn{2}{c}{} & 0&0295 & \multicolumn{2}{c}{} & -0&2035 & \multicolumn{2}{c}{} & -0&0083\tabularnewline
 &  & MSE &  & 0&0359 & \multicolumn{2}{c}{} & 0&0325 & \multicolumn{2}{c}{} & 0&0601 & \multicolumn{2}{c}{} & 0&0812\tabularnewline
 & \multirow{2}{*}{$\theta_{\mathrm{C}}=-2.4$} & Mean Bias &  & 0&0252 & \multicolumn{2}{c}{} & -0&0516 & \multicolumn{2}{c}{} & 0&6634 & \multicolumn{2}{c}{} & 0&0088\tabularnewline
 &  & MSE &  & 0&0566 & \multicolumn{2}{c}{} & 0&0552 & \multicolumn{2}{c}{} & 0&4699 & \multicolumn{2}{c}{} & 0&4649\tabularnewline
 & \multirow{2}{*}{$\theta_{\mathrm{EC}}=-0.2$} & Mean Bias &  & -0&0255 & \multicolumn{2}{c}{} & -0&0034 & \multicolumn{2}{c}{} & -0&3078 & \multicolumn{2}{c}{} & -0&0122\tabularnewline
 &  & MSE &  & 0&0113 & \multicolumn{2}{c}{} & 0&0046 & \multicolumn{2}{c}{} & 0&1100 & \multicolumn{2}{c}{} & 0&1228\tabularnewline
 & Converged & \% &  & 89&8\% & \multicolumn{2}{c}{} & 97&2\% & \multicolumn{2}{c}{} & 95&2\% & \multicolumn{2}{c}{} & \multicolumn{2}{c}{100\%}\tabularnewline
 & \multirow{2}{*}{Iterations} & Median &  & \multicolumn{2}{c}{347} & \multicolumn{2}{c}{} & \multicolumn{2}{c}{53} & \multicolumn{2}{c}{} & \multicolumn{2}{c}{176} & \multicolumn{2}{c}{} & \multicolumn{2}{c}{53}\tabularnewline
 &  & IQR &  & \multicolumn{2}{c}{117} & \multicolumn{2}{c}{} & \multicolumn{2}{c}{18} & \multicolumn{2}{c}{} & \multicolumn{2}{c}{46} & \multicolumn{2}{c}{} & \multicolumn{2}{c}{10}\tabularnewline
 & Time (min.) & Total &  & 1&9699 & \multicolumn{2}{c}{} & 0&5606 & \multicolumn{2}{c}{} & 1&2573 & \multicolumn{2}{c}{} & 0&4207\tabularnewline
\hline 
\multirow{10}{*}{\begin{turn}{90}
$N=1000$
\end{turn}} & $\theta_{\mathrm{M}}=1.2$ & Mean Bias &  & -0&0044 & \multicolumn{2}{c}{} & 0&0019 & \multicolumn{2}{c}{} & -0&2002 & \multicolumn{2}{c}{} & -0&0100\tabularnewline
 &  & MSE &  & 0&0083 & \multicolumn{2}{c}{} & 0&0057 & \multicolumn{2}{c}{} & 0&0444 & \multicolumn{2}{c}{} & 0&0173\tabularnewline
 & $\theta_{\mathrm{C}}=-2.4$ & Mean Bias &  & 0&0076 & \multicolumn{2}{c}{} & -0&0027 & \multicolumn{2}{c}{} & 0&6517 & \multicolumn{2}{c}{} & 0&0264\tabularnewline
 &  & MSE &  & 0&0106 & \multicolumn{2}{c}{} & 0&0065 & \multicolumn{2}{c}{} & 0&4303 & \multicolumn{2}{c}{} & 0&0962\tabularnewline
 & $\theta_{\mathrm{EC}}=-0.2$ & Mean Bias &  & -0&0059 & \multicolumn{2}{c}{} & -0&0016 & \multicolumn{2}{c}{} & -0&2915 & \multicolumn{2}{c}{} & -0&0146\tabularnewline
 &  & MSE &  & 0&0018 & \multicolumn{2}{c}{} & 0&0008 & \multicolumn{2}{c}{} & 0&0884 & \multicolumn{2}{c}{} & 0&0255\tabularnewline
 & Converged & \% &  & 95&1\% & \multicolumn{2}{c}{} & 98&7\% & \multicolumn{2}{c}{} & 99&6\% & \multicolumn{2}{c}{} & \multicolumn{2}{c}{100\%}\tabularnewline
 & \multirow{2}{*}{Iterations} & Median &  & \multicolumn{2}{c}{362} & \multicolumn{2}{c}{} & \multicolumn{2}{c}{50} & \multicolumn{2}{c}{} & \multicolumn{2}{c}{161} & \multicolumn{2}{c}{} & \multicolumn{2}{c}{52}\tabularnewline
 &  & IQR &  & \multicolumn{2}{c}{83} & \multicolumn{2}{c}{} & \multicolumn{2}{c}{12} & \multicolumn{2}{c}{} & \multicolumn{2}{c}{18} & \multicolumn{2}{c}{} & \multicolumn{2}{c}{8}\tabularnewline
 & Time (min.) & Total &  & 2&8453 & \multicolumn{2}{c}{} & 0&6559 & \multicolumn{2}{c}{} & 1&5548 & \multicolumn{2}{c}{} & 0&5956\tabularnewline
\hline 
\end{tabular}
\end{table}

Using the same model as before and focusing on equilibria (i) and
(ii), we proceed in the following way for each of 1000 replications.
First, we generated five completely random starting values for the
choice probabilities $P$. For each, we compute and store the corresponding
1-NPL estimate for $\theta$. Then we compute the $\infty$-NPL and
$\infty$-EPL estimates for each starting value.\footnote{As in the previous simulations, we allow for up to 100 iterations
per starting value but terminate early if convergence is achieved.} Finally, for both $\infty$-NPL and $\infty$-EPL we select from
among the five estimates the one that maximizes the likelihood function.
The results in Table \ref{tab:PSD-MultiStart} are calculated using
the best estimates for each of the 1000 replications. Reported iteration
counts, convergence percentages, and computational times include all
5 starting values over all replications.

Overall, the comparisons between $k$-EPL and $k$-NPL are qualitatively
similar to the case before with initial consistent estimates. For
the NPL-stable equilibrium (i), the small sample results for bias
are mixed, but $k$-EPL is faster, converges more often, and has smaller
MSE. In the large sample, $k$-EPL always has lower bias and MSE in
addition to being more stable and computationally lighter. Also as
before, $k$-EPL is robust to the NPL-unstable equilibrium. These
results show that $\infty$-EPL can be used as a computational procedure
to compute the MLE without initial consistent estimates, by using
multiple starting values and choosing the best estimates. This is
true even for equilibrium data generating processes where using the
same procedure with $\infty$-NPL would yield inconsistent estimates.

\subsection{\citet{PesSD2010}: Static Game, Unstable Equilibrium}

As a simple illustration of the performance of $k$-EPL, we consider
estimating the static game ($\beta=0$) of incomplete information
from \citet{PesSD2010}. This example is particularly interesting
because it was constructed as an example where inconsistency of $\infty$-NPL
can be shown analytically. We discuss only some relevant details of
the model and refer the reader to \citet{PesSD2010} for a full description.

There are two agents (players), $j\in\{1,2\}$, and two possible actions,
$a\in\{0,1\}$. The structural parameter is a scalar: $\theta\in[-10,-1]$.
The choice probabilities are $\Pr(a^{j}=1\mid\theta,P^{-j})=1-F_{\alpha}(-\theta P^{-j})$,
where $0<\alpha<1$. $F_{\alpha}$ is an approximate uniform distribution
with $F_{\alpha}(v)=v$ for $v\in[\alpha,1-\alpha)$ and a more complicated
form for $v\in\mathbb{R}\setminus[\alpha,1-\alpha)$ to guarantee
that it is a proper distribution function with full support. The probability
mass in the uniform region can be made arbitrarily close to 1 by taking
$\alpha\to0$. Given a value of $\theta$, the model has three equilibria
for $\alpha$ sufficiently close to zero. The equilibrium generating
the data is described by the following fixed point equation in $P$
space:
\[
\begin{aligned}\left[\begin{array}{c}
P^{1}\\
P^{2}
\end{array}\right]= & \left[\begin{array}{c}
1+\theta P^{2}\\
1+\theta P^{1}
\end{array}\right]\\
= & \left[\begin{array}{c}
1\\
1
\end{array}\right]+\left[\begin{array}{cc}
0 & \theta\\
\theta & 0
\end{array}\right]\left[\begin{array}{c}
P^{1}\\
P^{2}
\end{array}\right],
\end{aligned}
\]
or more compactly, $P=\Psi(\theta,P)$, so this can be used in $k$-NPL.
This linear system has a unique solution if and only if $\theta\neq-1$,
and the solution is $P^{1}=P^{2}=\frac{1}{1-\theta}$. \citet{PesSD2010}
consider $\theta^{*}=-2$, so that $P_{1}^{*}=P_{2}^{*}=\frac{1}{3}$.
They show that as $N\to\infty$, $\hat{\theta}_{\infty-NPL}\overset{p}{\to}-1$.
Rather than repeat their full explanation of this result, we instead
focus on explaining why the sequence does not converge to $\theta^{*}=-2$.
The reason is, essentially, because the equilibrium is unstable. Notice
that 
\[
\nabla_{P}\Psi(\theta^{*},P^{*})=\left[\begin{array}{cc}
0 & -2\\
-2 & 0
\end{array}\right],
\]
 which has eigenvalues $\lambda=\pm2$, implying that the equilibrium
is unstable. \citet{KasShim2012} show that the non-convergence issue
in $k$-NPL can be rectified by estimating separate parameters for
each player. However, this type of adjustment may not induce convergence
in more general settings.

Consider, instead, estimating $\theta^{*}$ with $k$-EPL, with a
change of variable to $v$ space. Noting that $P^{-j}=F_{\alpha}(v^{-j})$,
in the equilibrium we have $v^{j}=\theta F_{\alpha}(v^{-j})$ where
$F_{\alpha}(v^{-j})=v^{-j}$ in the region of interest. The fixed-point
equation then reduces to: 
\[
\left[\begin{array}{c}
v^{1}\\
v^{2}
\end{array}\right]=\left[\begin{array}{c}
(1+v^{2})\theta\\
(1+v^{1})\theta
\end{array}\right].
\]
 So, we can define $Y\equiv v=(v^{1},v^{2})$ and therefore
\begin{align*}
G(\theta,v) & =v-\left(\left[\begin{array}{cc}
0 & 1\\
1 & 0
\end{array}\right]v+\left[\begin{array}{c}
1\\
1
\end{array}\right]\right)\theta\\
 & =v-\left(Av+b\right)\theta
\end{align*}
 Because $\theta$ is a scalar, $G(\theta,v)$ is linear in $\theta$
and $v$ separately (holding the other fixed) but not jointly. Linearity
in $v$ is important because $G(\theta,v)=0$ can be solved analytically
via a linear system. So, we can easily compute the finite-sample MLE
via nested fixed point as well as via EPL iterations. Additionally,
we see that $\nabla_{v}G(\theta,v)=I-A\theta$ and we can easily verify
that this is invertible if and only if $\theta\neq-1$. And since
$v^{j}=\theta P^{-j}$, we also have $q_{i}(\theta,v)=q_{i}(v)$,
so that $\theta$ only influences $q_{i}$ through $v(\theta)$. This
modification is made without loss of generality in full MLE subject
to the equilibrium constraint, so it is also valid here. For $k$-EPL,
we have $\Upsilon(\theta,\hat{\gamma}_{k-1})=\hat{v}_{k-1}-\nabla_{v}G(\hat{\theta}_{k-1},\hat{v}_{k-1}){}^{-1}G(\theta,\hat{v}_{k-1})$,
as in Algorithm (\ref{alg:EPL-Algorithm}). %
\begin{comment}
Because $G(\theta,Y)$ is linear in $\theta$, $\Upsilon$ is also
linear in $\theta$. This means that $\hat{q}_{i}(\theta,\Upsilon(\theta,\hat{\gamma}_{k-1}))=\hat{q}_{i}(\Upsilon(\theta,\hat{\gamma}_{k-1}))$
is the log of a linear function of $\theta$, so there are no additional
nonlinearities in the objective function, relative to $k$-NPL. If
we instead tried to use $Y=(P_{1},P_{2})$, then we would still get
$\Upsilon$ linear in $\theta$, but $q_{i}$ would then depend on
$\theta\Upsilon(\theta,Y)$, which is non-linear in $\theta$.
\end{comment}

All that remains now is to obtain $\hat{v}_{0}$ and $\hat{\theta}_{0}$.
Notice that the best response equations imply $\theta=\frac{P^{j}-1}{P^{-j}}$
for $j\in\{1,2\}$. So first, we obtain frequency estimators $\hat{P}_{0}^{1}$
and $\hat{P}_{0}^{2}$. We then use these to construct
\[
\hat{\theta}_{0}=\frac{\frac{\hat{P}_{0}^{1}-1}{\hat{P}_{0}^{2}}+\frac{\hat{P}_{0}^{2}-1}{\hat{P}_{0}^{1}}}{2},
\]
\[
\hat{v}_{0}^{j}=\hat{\theta}_{0}\hat{P}_{0}^{-j}.
\]

We run Monte Carlo simulations of this model to illustrate the performance
of the estimators. We simulate 500 samples, each with 5,000 observations.
We estimate the model using MLE, $\infty$-EPL, and $\infty$-NPL.\footnote{By $\infty$-EPL and $\infty$-NPL, we mean that we iterate until
$||\hat{\theta}_{k}-\hat{\theta}_{k-1}||_{\infty}<10^{-6}$ or $k$
reaches 20. Estimation was performed with Matlab R2017a using \texttt{fmincon}.
The default tolerance of $10^{-6}$ is used for the solver.} The results are summarized in Table \ref{tab:PesSD2008-Monte-Carlo}.
The MLE and $\infty$-EPL estimates achieve mean $-2.0017$ and $-2.0014$,
respectively, and mean squared error (MSE) $0.0017$ and $0.0017$.
The two-sample Kolmogorov-Smirnov p-value is equal to 1. Furthermore,
$k$-EPL obtained convergence at $k=2$ in all 500 datasets. This
is unsurprising: with so many observations and only two players/actions,
we get \emph{very} precise initial estimates, so iteration converges
very quickly. The slight difference in means and MSE are likely due
to a combination of the tolerance used in estimation and non-linearity
in the full MLE objective function.

\begin{table}[h]
\caption{\label{tab:PesSD2008-Monte-Carlo}\citet{PesSD2010} Monte Carlo Results}

\centering{}%
\begin{tabular}{ccc}
\hline 
\hline Estimator & Mean & MSE\tabularnewline
\hline 
MLE & -2.0017 & 0.0017\tabularnewline
$\infty$-EPL & -2.0014 & 0.0017\tabularnewline
$\infty$-NPL & -1.0342 & 0.9652\tabularnewline
\hline 
\end{tabular}
\end{table}

On the other hand, $\infty$-NPL performs poorly as expected since
this model was constructed to be an example where $\infty$-NPL is
inconsistent. The estimate has a mean of $-1.0342$ and MSE of $0.9651$.
Almost all of the MSE is due to the asymptotic bias, so the estimate
is reliably converging to the wrong number.\footnote{There were 17 samples for which NPL converged in 3 or fewer iterations.
The mean and MSE for these samples were $-1.9932$ and $0.0012$,
respectively. For the other 483 samples, convergence took at least
12 iterations. These had a mean and MSE of $-0.9991$ and $0.9991$,
respectively. \citet{AgMarcoux2019} explain why the estimates converge
to ``good'' values in some samples even though the equilibrium generating
the data is unstable.}

\section{Alternative Choices of $\Upsilon(\cdot)$}

We have already mentioned that the choice of $\Upsilon$ used in the
algorithm could be replaced with full Newton steps without affecting
the asymptotic results. These are only two of several choices that
yield the same asymptotic results, as shown in the next theorem.
\begin{thm}
\label{thm:Alternative-Properties}(Asymptotically Equivalent Definitions
of $\Upsilon$) The results of Theorems \ref{thm:kEPL-Properties}
and \ref{thm:Finite-Sample-Properties} hold when $\Upsilon$ is defined
as any of the following, where $\hat{\gamma}_{k-1}=(\hat{\theta}_{k-1},\hat{Y}_{k-1})$:
\begin{enumerate}
\item $\Upsilon(\theta,\hat{\gamma}_{k-1})=\hat{Y}_{k-1}-\nabla_{Y}G(\hat{\theta}_{k-1},\hat{Y}_{k-1}){}^{-1}G(\theta,\hat{Y}_{k-1})$.
\item $\Upsilon(\theta,\hat{\gamma}_{k-1})=\hat{Y}_{k-1}-Z(\hat{\theta}_{k-1},\hat{Y}_{k-1})^{-1}G(\theta,\hat{Y}_{k-1})$,
where $Z$ is a continuously differentiable function and $Z(\theta,Y_{\theta})=\nabla_{Y}G(\theta,Y_{\theta})$
for all $\theta$.
\item $\Upsilon(\theta,\hat{\gamma}_{k-1})=\hat{Y}_{k-1}-\nabla_{Y}G(\theta,\hat{Y}_{k-1}){}^{-1}G(\theta,\hat{Y}_{k-1})$.
\end{enumerate}
\end{thm}
\begin{proof}
All of the listed $\Upsilon(\cdot)$ functions satisfy the zero Jacobian
property, so the results from the proofs of Theorems \ref{thm:kEPL-Properties}
and \ref{thm:Finite-Sample-Properties} carry through.
\end{proof}
The first definition of $\Upsilon$ in the Theorem \ref{thm:Alternative-Properties}
is the one we have worked with so far. The second definition is a
generalization of the first and can allow researchers to circumvent
the need for an initial $\hat{\theta}_{0}$ if they can find some
$Z(\hat{\theta}_{k-1},\hat{Y}_{k-1})=Z(\hat{Y}_{k-1})$ or even $Z(\hat{\theta}_{k-1},\hat{Y}_{k-1})=A$
that has the required properties.\footnote{Of course, $Z(\hat{\theta}_{k-1},\hat{Y}_{k-1})=\nabla_{Y}G(\hat{\theta}_{k-1},\hat{Y}_{k-1})$
is an option, so definition 2 includes definition 1.} We will show later on that this definition can be used in single-agent
dynamic discrete choice models. The third definition, which is an
exact Newton step, is likely the least useful because it requires
inverting $G_{Y}(\theta,\hat{Y}_{k-1})$ at multiple values of $\theta$,
which can be computationally burdensome and also will introduce additional
nonlinearities in the objective function for optimization.\footnote{More precisely, it only requires solving the linear system $G_{Y}(\theta,\hat{Y}_{k-1})b=G(\theta,\hat{Y}_{k-1})$
for $b$. However, the point about computational burden remains.} However, we include it for completeness. For all of the definitions
of $\Upsilon$ in the theorem, the results of Lemma \ref{lem:Newt-Props}
hold when all appropriate terms are replaced with $(\theta^{*},Y^{*})$
or $(\hat{\theta}_{\text{{MLE}}},\hat{Y}_{\text{{MLE}}})$. So, the
proof techniques from Theorems \ref{thm:kEPL-Properties} and \ref{thm:Finite-Sample-Properties}
can be used to prove Theorem \ref{thm:Alternative-Properties}.
\end{document}